\documentclass[prx,twocolumn,floatfix,superscriptaddress,longbibliography,notitlepage]{revtex4-1}

\usepackage{graphicx, color, graphpap}% Include figure files
\usepackage{enumitem}
\usepackage{amssymb}
\usepackage{amsthm}
\usepackage{multirow}
\usepackage[colorlinks=true,citecolor=blue,linkcolor=magenta]{hyperref}
\usepackage[T1]{fontenc}
\usepackage{bbm}
\usepackage{thmtools,thm-restate}
\usepackage{verbatim}
\usepackage{mathtools}
\usepackage{titlesec}
\usepackage{color}  
\usepackage{xcolor}

\long\def\ca#1\cb{} %Use for commenting out: \ca...\cb

\newcommand{\ketbra}[2]{| \hspace{1pt} #1 \rangle \langle #2 \hspace{1pt} |}

\newcommand{\avg}[1]{\langle #1\rangle }
\newcommand{\ket}[1]{|#1\rangle}               %ket
\newcommand{\bra}[1]{\langle #1|}              %bra
\newcommand{\dya}[1]{\ket{#1}\!\bra{#1}}
\newcommand{\poly}{\operatorname{poly}}

\newcommand{\rank}{\text{rank}}

\newcommand{\AC}{\mathcal{A}}
\newcommand{\BC}{\mathcal{B}}

\newcommand{\HC}{\mathcal{H}}

\newcommand{\LC}{\mathcal{L}}

\newcommand{\OC}{\mathcal{O}}

\newcommand{\SC}{\mathcal{S}}

\newcommand{\HS}{\text{HS}}

\newcommand{\Tr}{{\rm Tr}}

\newcommand{\Var}{{\rm Var}}

\renewcommand{\geq}{\geqslant}
\renewcommand{\leq}{\leqslant}
\newcommand{\mte}[2]{\langle#1|#2|#1\rangle }

\newcommand{\LCb}{\overline{\LC}}
\newcommand{\VL}{V_{\LC}}
\newcommand{\VLb}{V_{\LCb}}

\renewcommand{\vec}[1]{\boldsymbol{#1}}  % Bold vectors instead of arrow vectors
\newcommand{\ot}{\otimes}
\newcommand{\ad}{^\dagger}
\newcommand*{\id}{\openone}

\newcommand{\thv}{\vec{\theta}}

{}%         Body font
{}%         Indent amount (empty = no indent, \parindent 

\newtheoremstyle{mystyle}%                % Name
  {}%                                     % Space above
  {}%                                     % Space below
  {}%                                     % Body font
  {}%                                     % Indent amount
  {\itshape}%                            % Theorem head font
  {.}%                                    % Punctuation after theorem head
  { }%                                    % Space after theorem head, ' ', or \newline
  {}%                                     % Theorem head spec (can be left empty, meaning `normal')

\theoremstyle{mystyle}

\newtheorem{theorem}{Theorem}
\newtheorem{lemma}{Lemma}
\newtheorem{corollary}{Corollary}
\newtheorem{proposition}{Proposition}

\begin{document}
%\title{Cost-Function-Dependent Barren Plateaus in Shallow Quantum Neural Networks}
\title{Cost Function Dependent Barren Plateaus in Shallow Parametrized Quantum Circuits}

\author{M. Cerezo}
\thanks{Corresponding authors: \\ Marco Cerezo - cerezo@lanl.gov\\ Patrick J. Coles - pcoles@lanl.gov}
\affiliation{Theoretical Division, Los Alamos National Laboratory, Los Alamos, NM, USA.}
\affiliation{Center for Nonlinear Studies, Los Alamos National Laboratory, Los Alamos, NM, USA
}

\author{Akira Sone}
\affiliation{Theoretical Division, Los Alamos National Laboratory, Los Alamos, NM, USA.}
\affiliation{Center for Nonlinear Studies, Los Alamos National Laboratory, Los Alamos, NM, USA
}

\author{Tyler Volkoff}
\affiliation{Theoretical Division, Los Alamos National Laboratory, Los Alamos, NM, USA.}

\author{Lukasz Cincio}
\affiliation{Theoretical Division, Los Alamos National Laboratory, Los Alamos, NM, USA.}

\author{Patrick J. Coles}
\thanks{Corresponding authors: \\ Marco Cerezo - cerezo@lanl.gov\\ Patrick J. Coles - pcoles@lanl.gov}
\affiliation{Theoretical Division, Los Alamos National Laboratory, Los Alamos, NM, USA.}

\begin{abstract}
Variational quantum algorithms (VQAs) optimize the parameters $\vec{\theta}$ of a parametrized quantum circuit $V(\vec{\theta})$ to minimize a cost function $C$. While VQAs may enable practical applications of noisy quantum computers, they are nevertheless heuristic methods with unproven scaling. Here, we rigorously prove two results, assuming $V(\vec{\theta})$  is an alternating layered  ansatz composed  of blocks forming local 2-designs. Our first result states that defining $C$ in terms of global observables leads to exponentially vanishing gradients (i.e., barren plateaus) even when $V(\vec{\theta})$ is shallow. Hence, several VQAs in the literature must revise their proposed costs. On the other hand, our second result states that defining $C$ with local observables leads to at worst a polynomially vanishing gradient, so long as the depth of $V(\vec{\theta})$ is $\mathcal{O}(\log n)$. Our results establish a connection between locality and trainability. We illustrate these ideas with large-scale simulations, up to 100 qubits, of a quantum autoencoder implementation.
\end{abstract}
\maketitle

\section*{INTRODUCTION}\label{sc:intro}

One of the most important technological questions is whether Noisy Intermediate-Scale Quantum (NISQ) computers will have practical applications~\cite{preskill2018quantum}. NISQ devices are limited both in qubit count and in gate fidelity, hence preventing the use of quantum error correction.

The leading strategy to make use of these devices are variational quantum algorithms (VQAs)~\cite{mcclean2016theory}. VQAs employ a quantum computer to efficiently evaluate a cost function $C$, while a classical optimizer trains the parameters $\thv$ of a Parametrized Quantum Circuit (PQC) $V(\thv)$. The benefits of VQAs are three-fold. First, VQAs allow for task-oriented programming of quantum computers, which is important since designing quantum algorithms is non-intuitive. Second, VQAs make up for small qubit counts by leveraging classical computational power. Third, pushing complexity onto classical computers, while only running short-depth quantum circuits, is an effective strategy for error mitigation on NISQ devices.

There are very few rigorous scaling results for VQAs (with exception of one-layer approximate optimization~\cite{qaoa2014, hadfield2019quantum, hastings2019classical}).  Ideally, in order to reduce gate overhead that arises when implementing on quantum hardware one would like to employ a hardware-efficient ansatz~\cite{kandala2017hardware} for $V(\thv)$. As recent large-scale implementations for chemistry~\cite{arute2020hartree} and optimization~\cite{arute2020quantum} applications have shown, this ansatz leads to smaller errors due to hardware noise. However, one of the few known scaling results is that deep versions of randomly initialized  hardware-efficient ansatzes lead to exponentially vanishing gradients~\cite{mcclean2018barren}. Very little is known about the  scaling of the gradient in such ansatzes for shallow depths, and it would be especially useful to have a converse bound that guarantees  non-exponentially vanishing gradients for certain depths.
This motivates our work, where we rigorously investigate the  gradient scaling of VQAs as a function of the circuit depth.

The other motivation for our work is the recent explosion in the number of proposed VQAs. The Variational Quantum Eigensolver (VQE) is the most famous VQA. It aims to prepare the ground state of a given Hamiltonian $H = \sum_{\alpha} c_{\alpha} \sigma_{\alpha}$, with $H$ expanded as a sum of local Pauli operators~\cite{VQE}. In VQE, the cost function is obviously the energy $C = \mte{\psi}{H}$ of the trial state $\ket{\psi}$. However, VQAs have been proposed for other applications, like quantum data compression~\cite{Romero}, quantum error correction~\cite{johnson2017qvector}, quantum metrology~\cite{koczor2019variational}, quantum compiling~\cite{QAQC, jones2018quantum, sharma2020noise, heya2018variational}, quantum state diagonalization~\cite{VQSD, bravo2020quantum}, quantum simulation~\cite{Li, heya2019subspace, cirstoiu2019variational, otten2019noise}, fidelity estimation~\cite{cerezo2020variational}, unsampling~\cite{carolan2019variational}, consistent histories~\cite{arrasmith2019variational}, and linear systems~\cite{bravo-prieto2019, Xiaosi, huang2019near}. For these applications, the choice of $C$ is less obvious. Put another way, if one reformulates these VQAs as ground-state problems (which can be done in many cases), the choice of Hamiltonian $H$ is less intuitive. This is because many of these applications are abstract, rather than associated with a physical Hamiltonian.

We remark that polynomially vanishing gradients imply that the number of shots needed to estimate the gradient should grow as $\OC(\poly(n))$. In contrast, exponentially vanishing gradients (i.e., barren plateaus) imply that derivative-based optimization will have exponential scaling~\cite{cerezo2020impact}, and this scaling can also apply to derivative-free optimization~\cite{arrasmith2020effect}. Assuming a polynomial number of shots per optimization step, one will be able to resolve against finite sampling noise and train the parameters if the gradients vanish polynomially. 
Hence, we employ the term ``trainable'' for polynomially vanishing gradients.

In this work we connect the trainability of VQAs to the choice of $C$. For the abstract applications in Refs.~\cite{Romero, johnson2017qvector,koczor2019variational, QAQC, jones2018quantum, sharma2020noise, heya2018variational, VQSD, bravo2020quantum, Li, heya2019subspace, cirstoiu2019variational, otten2019noise, cerezo2020variational, carolan2019variational, arrasmith2019variational, bravo-prieto2019, Xiaosi, huang2019near}, it is important for $C$ to be operational, so that small values of $C$ imply that the task is almost accomplished. Consider an example of state preparation, where the goal is to find a gate sequence that prepares a target state $\ket{\psi_0}$. A natural cost function is the square of the trace distance $D_{\text{T}}$ between $\ket{\psi_0}$ and $\ket{\psi} = V(\thv)\ad\ket{\vec{0}}$, given by $C_{\text{G}} = D_{\text{T}}(\ket{\psi_0}, \ket{\psi})^2 $,  which is equivalent to 
\begin{equation}
C_{\text{G}} = \Tr[O_{\text{G}} V(\thv)\dya{\psi_0}V(\thv)\ad]\,, \label{eq:Cost-TD} 
\end{equation}
with  $O_{\text{G}}=\id - \dya{\vec{0}}$. Note that $\sqrt{C_{\text{G}}} \geq |\mte{\psi}{M} - \mte{\psi_0}{M}|$ has a nice operational meaning as a bound on the expectation value difference for a POVM element~$M$. 

However, here we argue that this cost function and others like it  exhibit exponentially vanishing gradients. Namely, we consider global cost functions, where one directly compares states or operators living in exponentially large Hilbert spaces (e.g., $\ket{\psi}$ and $\ket{\psi_0}$). These are precisely the cost functions that have operational meanings for tasks of interest, including all tasks in Refs.~\cite{Romero, johnson2017qvector,koczor2019variational, QAQC, jones2018quantum, sharma2020noise, heya2018variational, VQSD, bravo2020quantum, Li, heya2019subspace, cirstoiu2019variational, otten2019noise, cerezo2020variational, carolan2019variational, arrasmith2019variational, bravo-prieto2019, Xiaosi, huang2019near}. Hence, our results imply that a non-trivial subset of these references will need to revise their choice of $C$. 

\begin{figure}[t]
    \centering
    \includegraphics[width=.94\columnwidth]{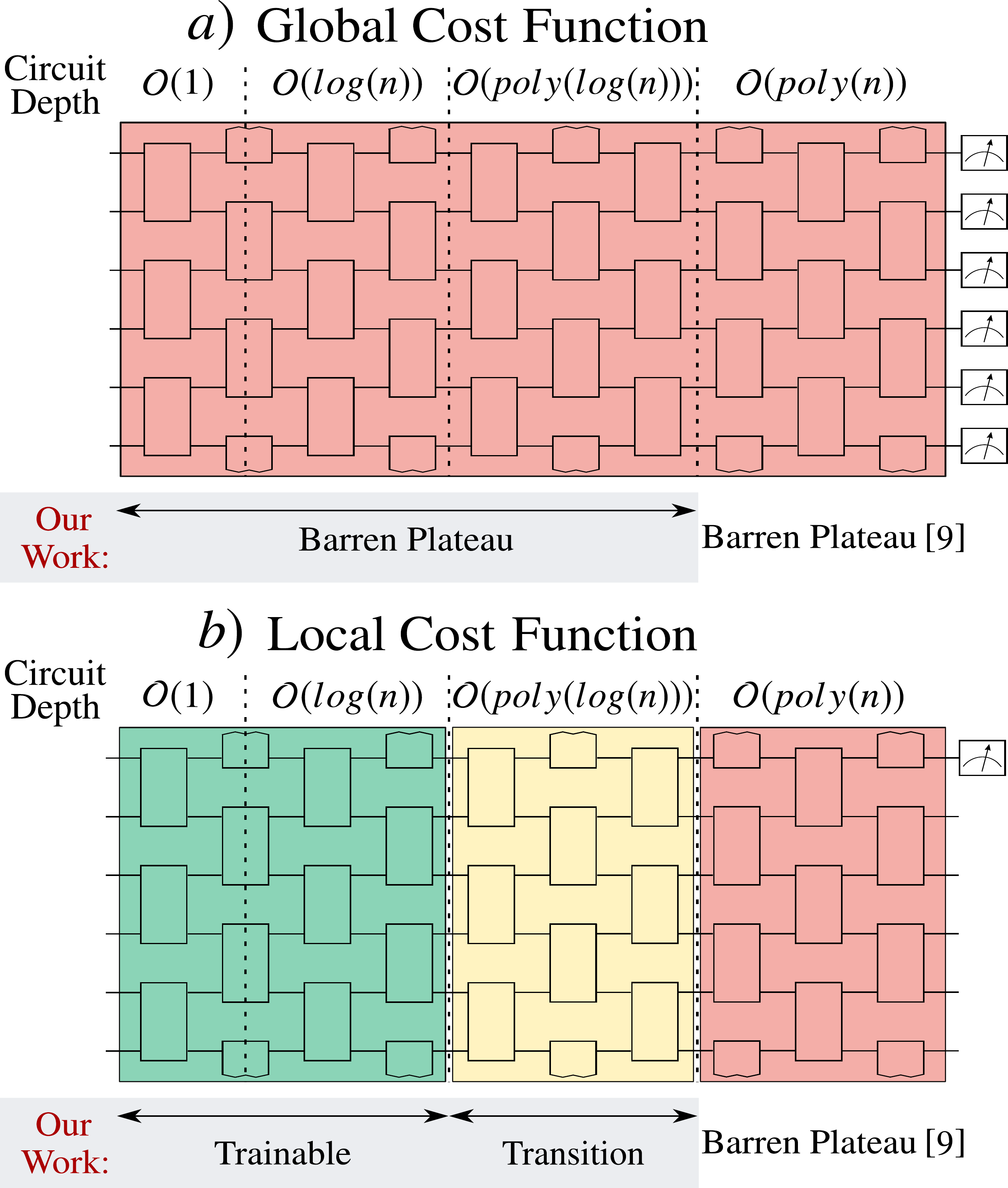}
    \caption{Summary of our main results. McClean {\it et al.}~\cite{mcclean2018barren} proved that a barren plateau can occur when the depth $D$ of a hardware-efficient ansatz is $D\in\OC(\poly(n))$. Here we extend these results by providing bounds for the variance of the gradient of global and local cost functions as a function of $D$. In particular, we find that the barren plateau phenomenon is cost-function dependent. a) For global cost functions (e.g., Eq.~\eqref{eq:Cost-TD}), the landscape will exhibit a barren plateau essentially for all depths $D$. b) For local cost functions (e.g., Eq.~\eqref{eq:localcostfunction}), the gradient vanishes at worst polynomially and hence is trainable when $D\in \OC(\log(n))$, while barren plateaus occur for $D\in \OC(\poly(n))$, and between these two regions the gradient transitions from polynomial to exponential decay.}
    \label{fig:Overview}
\end{figure}

Interestingly, we demonstrate vanishing gradients for shallow  PQCs. This is in contrast to McClean {\it et al}.~\cite{mcclean2018barren}, who showed vanishing gradients for deep  PQCs. They noted that randomly initializing $\thv$ for a $V(\thv)$ that forms a 2-design leads to a barren plateau, i.e., with the gradient vanishing exponentially in the number of qubits, $n$. Their work implied that researchers must develop either clever parameter initialization strategies~\cite{grant2019initialization, verdon2019learning} or clever PQCs ansatzes~\cite{lee2018generalized,hadfield2019quantum,verdon2019quantum}. Similarly, our work implies that researchers must carefully weigh the balance between trainability and operational relevance when choosing $C$. 

While our work is for general VQAs, barren plateaus for global cost functions were noted for specific VQAs and for a very specific tensor-product example by our research group~\cite{QAQC,VQSD}, and more recently in~\cite{huang2019near}. This motivated the proposal of local cost functions~\cite{QAQC, sharma2020noise, VQSD, cirstoiu2019variational, arrasmith2019variational, bravo-prieto2019, carolan2019variational}, where one compares objects (states or operators) with respect to each individual qubit, rather than in a global sense, and therein it was shown that these local cost functions have indirect operational meaning.

Our second result is that these local cost functions have gradients that vanish polynomially rather than exponentially in $n$, and hence have the potential to be trained. This holds for $V(\thv)$ with depth $\OC(\log n)$. Figure~\ref{fig:Overview} summarizes our two main results.

Finally, we illustrate our main results for an important example: quantum autoencoders~\cite{Romero}. Our large-scale numerics show that the global cost function proposed in \cite{Romero}  has a barren plateau. On the other hand, we propose a novel local cost function that is trainable, hence making quantum autoencoders a scalable application.

\section*{RESULTS}\label{sec:Results}

\subsection*{Warm-up example} \label{sec:warmup}

To illustrate cost-function-dependent barren plateaus, we first consider a toy problem corresponding to the state preparation problem in the Introduction with the target state being $\ket{\vec{0}}$. We assume a tensor-product ansatz of the form $V(\vec{\theta})=\bigotimes_{j=1}^{n}e^{-i \theta^j \sigma_{x}^{(j)} /2} $, with the goal of finding the angles $\theta^j$ such that $V(\vec{\theta})\ket{\vec{0}}=\ket{\vec{0}}$. Employing the global cost of \eqref{eq:Cost-TD} results in $C_{\text{G}} = 1-\prod_{j=1}^n\cos^2{\theta^j\over 2}$. The barren plateau can be detected via the variance of its gradient: $\Var[\frac{\partial C_{\text{G}}}{\partial \theta^j}]=\frac{1}{8}(\frac{3}{8})^{n-1}$, which is exponentially vanishing in $n$. Since the mean value is $\avg{\frac{\partial C_{\text{G}}}{\partial \theta^j}}=0$,  the gradient concentrates exponentially around zero.

On the other hand, consider a local cost function:
\begin{align}
C_{\text{L}}&=\Tr\left[O_{\text{L}}V(\thv)\dya{\vec{0}} V(\thv)\ad\right],
\label{eq:localcostfunction}\\
\text{with}\quad O_{\text{L}}&=\id - \frac{1}{n}\sum_{j=1}^n \dya{0}_j\otimes\id_{\overline{j}}\,,\label{eq:localcostfunction22}
\end{align}
where $\id_{\overline{j}}$ is the identity on all qubits except qubit $j$. Note that $C_{\text{L}}$ vanishes under the same conditions as $C_{\text{G}}$~\cite{QAQC,sharma2020noise},  $C_{\text{L}}=0 \Leftrightarrow C_{\text{G}}=0$. We find $C_{\text{L}} = 1-\frac{1}{n}\sum_{j=1}^n \cos^2{\theta^j \over 2}$, and the variance of its gradient is  $\Var[\frac{\partial C_{\text{L}}}{\partial \theta^j}]=\frac{1}{8n^2}$, which vanishes polynomially with $n$ and hence exhibits no barren plateau. Figure~\ref{fig:f2} depicts the cost landscapes of $C_{\text{G}}$ and $C_{\text{L}}$ for two values of $n$ and shows that the barren plateau can be avoided here via a local cost function.

Moreover, this example allows us to delve deeper into the cost landscape to see a phenomenon that we refer to as a narrow gorge. While a barren plateau is associated with a flat landscape, a narrow gorge refers to the steepness of the valley that contains the global minimum. This phenomenon is illustrated in Fig.~\ref{fig:f2}, where each dot corresponds to cost values obtained from randomly selected parameters $\vec{\theta}$. For $C_{\text{G}}$ we see that very few dots fall inside the narrow gorge, while for $C_{\text{L}}$ the narrow gorge is not present. Note that the narrow gorge makes it harder to train $C_{\text{G}}$ since the learning rate of descent-based optimization algorithms must be exponentially small in order not to overstep the narrow gorge. The following proposition (proved in the Supplementary Note 2) formalizes the narrow gorge for $C_{\text{G}}$ and its absence for $C_{\text{L}}$ by characterizing the dependence on $n$ of the probability $C \leq\delta$. This probability is associated with the parameter space volume that leads to $C\leq\delta$. 

\begin{proposition}\label{prop1}
Let $\theta^{j}$ be uniformly distributed on $[-\pi,\pi]$ $\forall j$. For any $\delta \in (0,1)$, the probability that $C_{\text{G}}\le \delta$ satisfies
\begin{equation}
\text{\normalfont Pr}\lbrace C_{\text{G}}\le \delta \rbrace \le (1-\delta)^{-1}\left({1\over 2} \right)^{n}.
\end{equation}
For any $\delta \in [{1\over 2},1]$, the probability that $C_{\text{L}}\le \delta$ satisfies
\begin{equation}
\text{\normalfont Pr}\lbrace C_{\text{L}}\le \delta \rbrace \ge { (2\delta -1)^{2}\over {1\over 2n} + (2\delta -1)^{2} } \underset{n\to\infty}{\longrightarrow}1\,.
\end{equation}
\label{prop:one}
\end{proposition}

\begin{figure}[t]
    \centering
    \includegraphics[width=1.0\columnwidth]{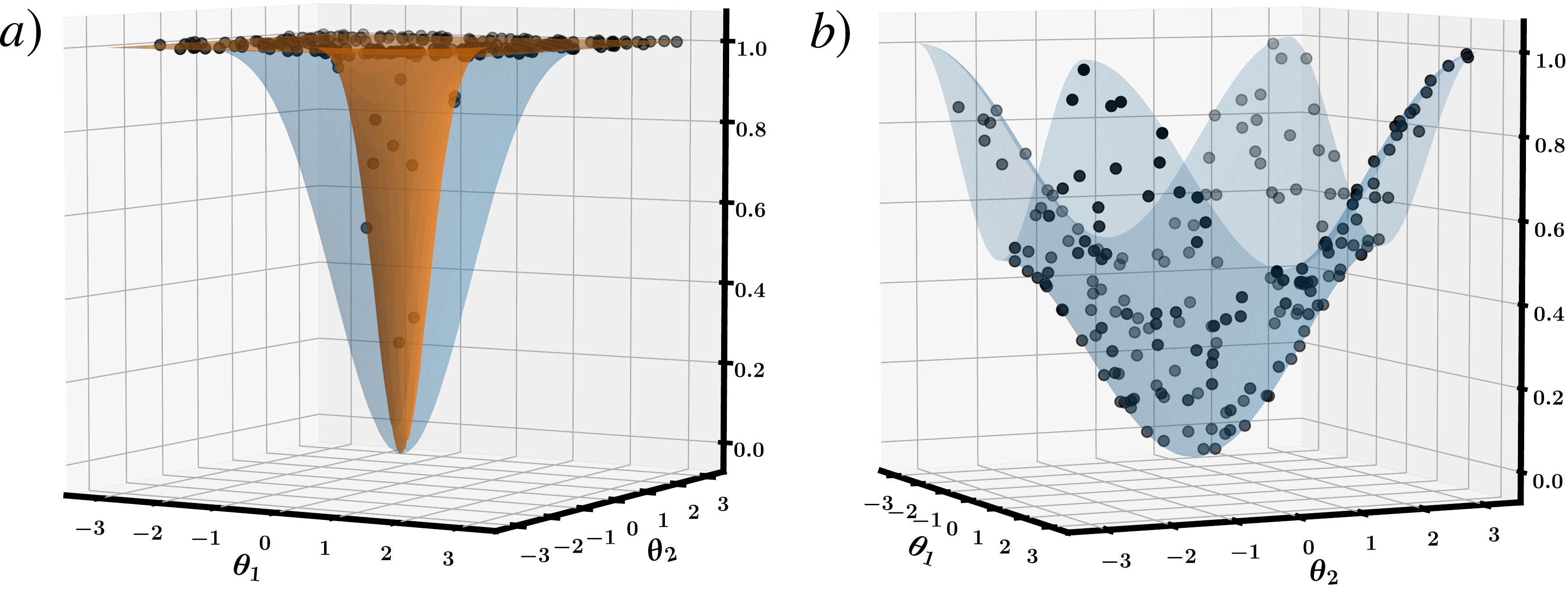}
    \caption{Cost function landscapes. a) Two-dimensional cross-section through the landscape of $C_{\text{G}} = 1-\prod_{j=1}^n\cos^2 (\theta^j/2)$ for $n=4$ (blue) and $n=24$ (orange). b) The same cross-section through the landscape of $C_{\text{L}} = 1-\frac{1}{n}\sum_{j=1}^n \cos^2 (\theta^j/2)$ is independent of $n$. In both cases, 200 Haar distributed points are shown, with very few (most) of these points lying in the valley containing the global minimum of $C_{\text{G}}$ ($C_{\text{L}}$).}
    \label{fig:f2}
\end{figure}

\subsection*{General framework}
\label{sec:Framework}

%\subsubsection{General cost function}

For our general results, we consider a family of cost functions that can be expressed as the expectation value of an operator $O$ as follows
\begin{equation}
    C=\Tr\left[ O V(\vec{\theta}) \rho V\ad(\vec{\theta}) \right]\,, \label{eq:cost_function}
\end{equation}
where $\rho$ is an arbitrary quantum state on $n$ qubits. Note that this framework includes the special case where $\rho$ could be a pure state, as well as the more special case where $\rho = \dya{\vec{0}}$, which is the input state for many VQAs such as VQE. Moreover, in VQE, one chooses $O = H$, where $H$ is the physical Hamiltonian. In general, the choice of $O$ and $\rho$ essentially defines the application of interest of the particular VQA.

It is typical to express $O$ as a linear combination of the form  $O=c_0\id+\sum_{i=1}^{N} c_i O_i$. Here $O_i\neq\id$,  $c_i\in\mathbb{R}$, and we assume that at least one $c_i\neq 0$. Note that $C_{\text{G}}$ and $C_{\text{L}}$ in \eqref{eq:Cost-TD} and \eqref{eq:localcostfunction} fall under this framework. In our main results below, we will consider two different choices of $O$ that respectively capture our general notions of global and local cost functions and also generalize the aforementioned $C_{\text{G}}$ and $C_{\text{L}}$.

%\subsubsection{Ansatz} \label{sebsec:Ansatz}

As shown in Fig.~\ref{fig:f3}(a), $V(\vec{\theta})$ consists of $L$ layers of $m$-qubit unitaries $W_{kl}(\vec{\theta}_{kl})$, or  blocks, acting on alternating groups of $m$ neighboring qubits.  We refer to this as an  Alternating Layered Ansatz. We remark that the Alternating Layered Ansatz will be a hardware-efficient ansatz so long as the gates that compose each block are taken from a set of gates native to a specific device. As depicted in Fig.~\ref{fig:f3}(c), the one dimensional Alternating Layered Ansatz can be readily implemented in devices with one-dimensional connectivity, as well as in devices with two-dimensional connectivity (such as that of IBM's~\cite{ibmqxdevices} and Google's~\cite{arute2019quantum} quantum devices). That is, with both one- and two-dimensional hardware connectivity one can group qubits to form an Alternating Layered Ansatz as in Fig.~\ref{fig:f3}(a). 

The index  $l=1,\ldots,L$ in $W_{kl}(\vec{\theta}_{kl})$ indicates the layer that contains the block, while $k=1,\ldots,\xi$ indicates the qubits it acts upon. We assume $n$ is a multiple of $m$, with $n=m \xi$, and that $m$ does not scale with $n$. As depicted in Fig.~\ref{fig:f3}(a), we define  $S_{k}$ as  the $m$-qubit subsystem on which $W_{kL}$ acts, and we define $\SC = \{S_{k}\}$ as the set of all such subsystems.  Let us now consider a block $W_{kl}(\vec{\theta}_{kl})$ in the $l$-th layer of the ansatz. For simplicity we henceforth use  $W$ to refer to a given $W_{kl}(\vec{\theta}_{kl})$. As shown in the Methods section, given a $\theta^\nu\in\vec{\theta}_{kl} $ that parametrizes a rotation $e^{-i\theta^\nu \sigma_\nu /2 }$ (with $\sigma_\nu$ a Pauli operator)   inside a given block $W$, one can always express
\begin{equation}\label{eq:newWaWb}
    \frac{\partial W}{\partial\theta^\nu}:=\partial_\nu W = \frac{-i}{2} W_{\text{A}}\sigma_\nu W_{\text{B}},
\end{equation}
where $W_{\text{A}}$ and $W_{\text{B}}$ contain all remaining gates in $W$, and are properly defined in the Methods section.

\begin{figure}
    \centering
    \includegraphics[width=.9\columnwidth]{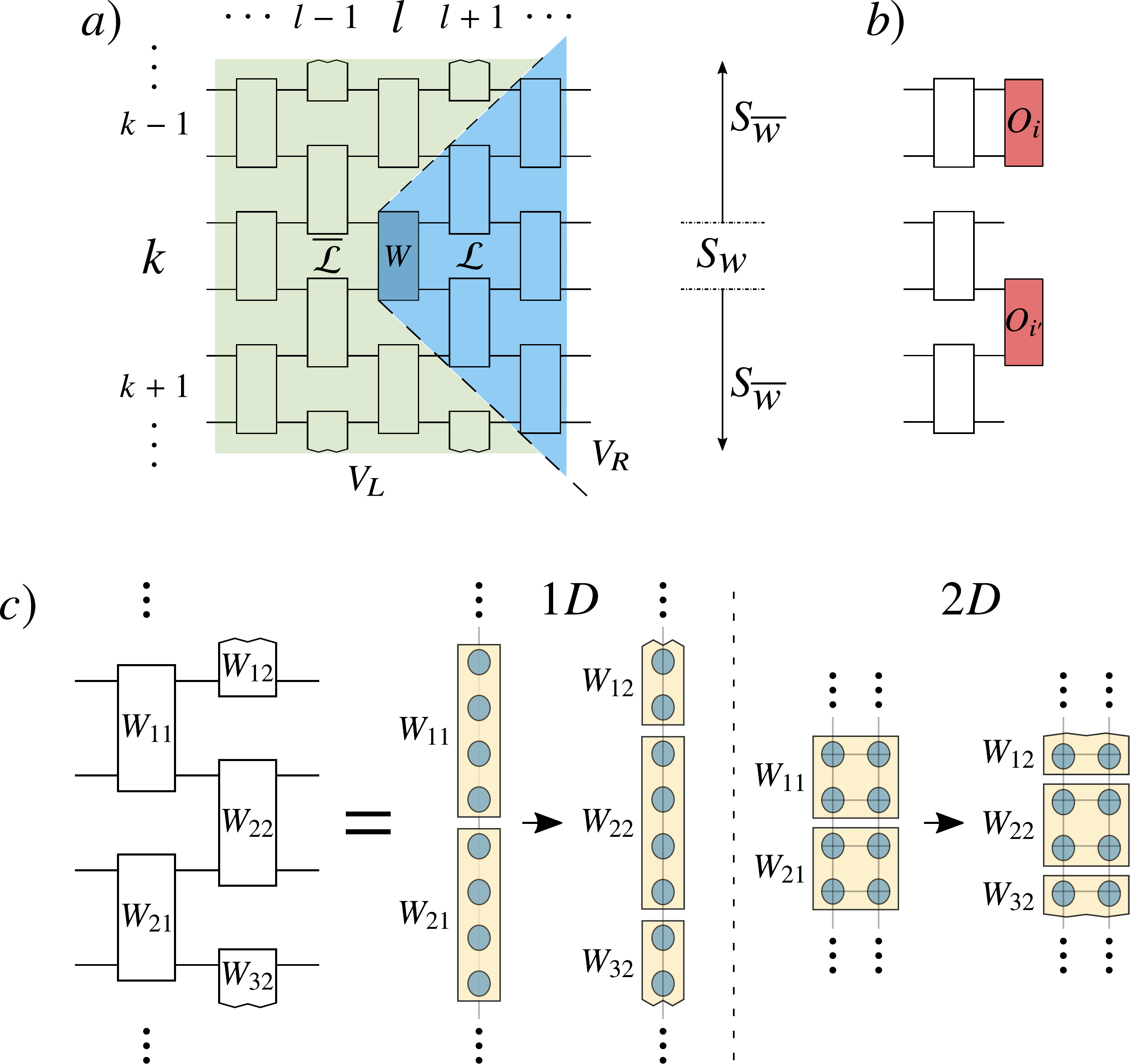}
    \caption{Alternating Layered Ansatz.   a) Each block $W_{kl}$ acts on $m$ qubits and is parametrized via \eqref{eq;W-gates}. As shown, we define  $S_{k}$ as  the $m$-qubit subsystem on which $W_{kL}$ acts, where $L$ is the last layer of $V(\thv)$. Given some block $W$, it is useful for our proofs (outlined in the Methods) to write $V(\vec{\theta})=V_{\text{R}} (\id_{\overline{w}}\otimes W)V_{\text{L}}$, where $V_{\text{R}}$ contains all gates in the forward light-cone $\LC$ of $W$. The forward light-cone $\LC$ is defined as all gates with at least one input qubit causally connected to the output qubits of $W$.   We define $\overline{\mathcal{L}}$ as the compliment of $\mathcal{L}$, $S_w$ as the $m$-qubit subsystem on which $W$ acts, and $S_{\overline{w}}$ as the $n-m$ qubit subsystem on which $W$ acts trivially. b) The operator $O_i$ acts non-trivially only in  subsystem $S_{k-1}\in\mathcal{S}$, while $O_{i'}$ acts non-trivially on the first $m/2$ qubits of $S_{k+1}$, and on the second $m/2$ qubits of $S_{k}$.  c) Depiction of the Alternating Layered Ansatz with one- and two-dimensional connectivity. Each circle represents a physical qubit.
    }
    \label{fig:f3}
\end{figure}

%\subsubsection{ Gradient of the cost function}\label{subsubsec:variance}

The contribution to the gradient $\nabla C$ from a parameter $ \theta^\nu$ in the block $W$ is given by the partial derivative $ \partial_\nu C$. While the value of $\partial_\nu C$ depends on the specific parameters $\vec{\theta}$, it is useful to compute $\langle\partial_\nu C\rangle_V$, i.e., the average gradient over all possible unitaries $V(\vec{\theta})$ within the ansatz. Such an average may not be representative near the minimum of $C$, although it does provide a good estimate of the expected gradient when randomly initializing the angles in $V(\vec{\theta})$. In the Methods Section we explicitly show how to compute averages of the form $\langle\dots\rangle_V$, and in the Supplementary Note 3 we provide a proof for the following Proposition.

\begin{proposition}
\label{prop2}
The average of the partial derivative of any cost function of the form~\eqref{eq:cost_function} with respect to a parameter $\theta^\nu$ in a block $W$ of the ansatz in Fig.~\ref{fig:f3} is
\begin{equation}
    \langle\partial_\nu C\rangle_V=0\,,
\end{equation}
provided that either $W_{\text{A}}$ or $W_{\text{B}}$ of ~\eqref{eq:newWaWb} form a $1$-design.
\end{proposition}

Here we recall that a $t$-design is an ensemble of unitaries, such that sampling over their distribution yields the same properties as sampling random unitaries from the unitary group with respect to the Haar measure up to the first $t$ moments~\cite{dankert2009exact}. The Methods section provides a formal definition of a $t$-design.

Proposition~\ref{prop2} states that the gradient is not biased in any particular direction. To analyze the trainability of $C$, we  consider the second moment of its partial derivatives: 
\begin{equation}
    \Var[\partial_\nu C]= \left\langle\left( \partial_\nu C\right)^2\right\rangle_V\,,\label{eq:variance}
\end{equation}
where we used the fact that $\langle\partial_\nu C\rangle_V=0$. The magnitude of $\Var[\partial_\nu C]$ quantifies how much the partial derivative concentrates around zero, and hence small values in~\eqref{eq:variance} imply that the slope of the landscape will typically be insufficient to provide a cost-minimizing direction. Specifically, from Chebyshev's inequality, $\Var[\partial_\nu C]$ bounds the probability that the cost-function partial derivative deviates from its mean value (of zero) as $\Pr \left(|\partial_\nu C|\geq c \right) \leq \Var[\partial_\nu C]/c^2$ for all $c>0$.

\subsection*{Main results}

Here we present our main theorems and corollaries, with the proofs sketched in the Methods and detailed in the Supplementary Information. In addition, in the Methods section we provide some intuition behind our main results by  analyzing a generalization of the warm-up example where $V(\thv)$ is composed of a single layer  of  the ansatz in Fig.~\ref{fig:f3}. This case bridges the gap between the warm-up example and our main theorems and also showcases the tools used to derive our main result.

The following theorem provides an upper bound on the variance of the partial derivative of a global cost function which can be expressed as the expectation value of an operator of the form
\begin{equation}
O=c_0\id+\sum_{i=1}^{N} c_i \widehat{O}_{i1}\otimes \widehat{O}_{i2}\otimes \dots\otimes \widehat{O}_{i\xi}\,.\label{eq-global-O}
\end{equation}
Specifically, we consider two cases of interest:  (i) When $N=1$ and each $\widehat{O}_{1k}$ is a non-trivial projector ($\widehat{O}_{1k}^2=\widehat{O}_{1k}\neq\id$) of rank $r_k$ acting on subsystem $S_k$, or (ii) When $N$ is arbitrary and $\widehat{O}_{ik}$ is traceless with  $\Tr[\widehat{O}_{ik}^2]\leq 2^m$ (for example, when $\widehat{O}_{ik}=\bigotimes_{j=1}^m \sigma^\mu_j$ is a tensor product of Pauli operators $\sigma^\mu_j\in\{\id_j,\sigma^x_j,\sigma^y_j,\sigma^z_j\}$, with at least one $\sigma^\mu_j\neq\id$). Note that case (i) includes $C_{\text{G}}$ of~\eqref{eq:Cost-TD} as a special case.

\begin{theorem}\label{thm1}
 Consider a trainable parameter  $\theta^\nu$ in a block $W$ of the ansatz in Fig.~\ref{fig:f3}.  Let $\Var[\partial_\nu C]$ be the variance of the partial derivative of a global cost function $C$ (with $O$ given by~\eqref{eq-global-O}) with respect to $\theta^\nu$.  If $W_{\text{A}}$, $W_{\text{B}}$ of~\eqref{eq:newWaWb}, and each block in $V(\vec{\theta})$ form a local $2$-design, then $\Var[\partial_\nu C]$ is upper bounded by 
\begin{equation}
    \Var[\partial_\nu C]\leq F_{n}(L,l)\,.\label{eq:varMain1} 
\end{equation}
\begin{itemize}
\item[(i)] For $N=1$ and when each $\widehat{O}_{1k}$ is a non-trivial projector, then defining  $R=\prod_{k=1}^{\xi}  r^2_{k}$, we have
\begin{align}
\label{eq:varMain2}
F_{n}(L,l) = \frac{2^{2m+(2m-1)(L-l)}}{(2^{2m}-1)\cdot3^{\frac{n}{m}}\cdot2^{(2-\frac{3}{m})n}}c_1^2 R\,.
\end{align}
\item[(ii)] For arbitrary $N$ and when each $\widehat{O}_{ik}$ satisfies $\Tr[\widehat{O}_{ik}]=0$ and $\Tr[\widehat{O}_{ik}^2]\leq 2^m$, then
\begin{equation}
    F_{n}(L,l)=\frac{2^{2m(L-l+1)+1}}{3^{\frac{2n}{m}}\cdot 2^{\left(3-\frac{4}{m}\right)n}}\sum_{i,j=1}^N c_i c_j \,.
    \label{eq:globalUpperPauliMTsm}
\end{equation}
\end{itemize}
\end{theorem}

From Theorem~\ref{thm1} we  derive the following corollary.
\begin{corollary}\label{cor1}
Consider the function $F_n(L,l)$. 

\begin{itemize}
    \item[(i)] Let $N=1$ and let each $\widehat{O}_{1k}$ be a non-trivial projector, as in case (i) of  Theorem~\ref{thm1}. If $c_1^2 R \in\OC(2^n)$ and if the number of layers $L\in\OC(\poly(\log(n)))$, then
    \begin{equation}
F_{n}\left(L,l\right)\in\OC\left(2^{-\left(1-\frac{1}{m}\log_23\right) n}\right)\,,
\end{equation}
which implies that $\Var[\partial_\nu C]$ is exponentially vanishing in $n$ if $m\geq2$.

\item[(ii)] Let $N$ be arbitrary, and let each $\widehat{O}_{ik}$ satisfy $\Tr[\widehat{O}_{ik}]=0$ and $\Tr[\widehat{O}_{ik}^2]\leq 2^m$, as in case (ii) of  Theorem~\ref{thm1}. If $N\in\OC(2^n)$,  $c_i\in\OC(1)$, and if the number of layers $L\in\OC(\poly(\log(n)))$, then
\begin{equation}
F_{n}\left(L,l\right)\in\OC\left(\frac{1}{2^{\left(1-\frac{1}{m}\right)n}}\right)\,,
\end{equation}
which implies that $\Var[\partial_\nu C]$ is exponentially vanishing in $n$ if $m\geq2$.
\end{itemize}

\end{corollary}

Let us now make several important remarks. First, note that part (i) of Corollary~\ref{cor1} includes as a particular example the cost function $C_{\text{G}}$ of~\eqref{eq:Cost-TD}. Second, part (ii) of this corollary also includes as particular examples operators with $N\in\OC(1)$, as well as $N\in\OC(\poly(n))$. Finally, we remark that $F_{n}(L,l)$ becomes trivial when  the number of layers $L$ is $\Omega(\poly(n))$, however, as we discuss below, we can still find that $\Var[\partial_\nu C_{\text{G}}]$ vanishes exponentially in this case.

Our second main theorem shows that barren plateaus can be avoided for shallow circuits by employing local cost functions. Here we consider $m$-local cost functions where each $\widehat{O}_i$ acts non-trivially on at most $m$ qubits and (on these qubits) can be expressed as $\widehat{O}_i=\widehat{O}_i^{\mu_i}\otimes \widehat{O}_{i}^{\mu_i'}$:
\begin{equation}\label{eq:Oml}
    O=c_0\id+\sum_{i=1}^N c_i \widehat{O}_i^{\mu_i}\otimes \widehat{O}_{i}^{\mu_i'} \,,
\end{equation}
where  $\widehat{O}_i^{\mu_i}$ are operators acting on $m/2$ qubits which can be written as a tensor product of Pauli operators. Here, we assume the summation  in Eq.~\eqref{eq:Oml} includes two possible cases as schematically shown in Fig.~\ref{fig:f3}(b): First, when  $\widehat{O}_i^{\mu_i}$ ($\widehat{O}_i^{\mu_i'}$) acts on the first (last) $m/2$ qubits of a given $S_{k}$, and second, when $\widehat{O}_i^{\mu_i}$ ($\widehat{O}_i^{\mu_i'}$) acts on the last (first) $m/2$ qubits of a given $S_{k}$ ($ S_{k+1}$). This type of cost function includes any ultralocal  cost function (i.e., where the $\widehat{O}_i$ are one-body) as in~\eqref{eq:localcostfunction}, and also VQE Hamiltonians with up to $m/2$ neighbor interactions. Then, the following theorem holds.
\begin{theorem}\label{thm2}
Consider a trainable parameter  $\theta^\nu$ in a block $W$ of the ansatz in Fig.~\ref{fig:f3}.  Let $\Var[\partial_\nu C]$ be the variance of the partial derivative of an $m$-local cost function $C$ (with $O$ given by~\eqref{eq:Oml}) with respect to $\theta^\nu$. $W_{\text{A}}$, $W_{\text{B}}$ of~\eqref{eq:newWaWb}, and each block in $V(\vec{\theta})$ form a local $2$-design, then $\Var[\partial_\nu C]$ is lower bounded by 
\begin{equation}
    G_{n}(L,l)\leq\Var[\partial_\nu C]\,,\label{eq:varMaint2}
\end{equation}
with 
\begin{align}\label{eq:varMaint22}
    G_{n}(L,l)&=
    \frac{2^{m(l+1)-1}}{(2^{2m}-1)^2(2^{m}+1)^{L+l}}\notag\\
    &\hspace{8pt}\times\sum_{ i\in i_\LC} \sum_{\substack{(k,k')\in k_{\LC_{\text{B}}} \\ k' \geq k}} c_i^2  \epsilon(\rho_{k,k'}) \epsilon(\widehat{O}_i)\,,
\end{align}
where $i_\LC$ is the set of $i$ indices whose associated operators $\widehat{O}_i$ act on qubits in the forward light-cone $\LC$ of $W$, and $k_{\LC_{\text{B}}}$ is the set of $k$ indices whose associated subsystems $S_k$ are in the backward light-cone $\LC_{\text{B}}$ of $W$. Here we defined the function $\epsilon(M) = D_{\HS}\left(M,\Tr(M)\id / d_M\right)$ where $D_{\HS}$ is the Hilbert-Schmidt distance and $d_M$ is the dimension of the matrix $M$. In addition, $\rho_{k,k'}$ is the partial trace of the input state $\rho$ down to the subsystems $S_k S_{k+1}... S_{k'}$.
\end{theorem}

Let us make a few remarks. First, note that the $\epsilon(\widehat{O}_i)$ in the lower bound indicates that training $V(\vec{\theta})$ is easier when $\widehat{O}_i$ is far from the identity. Second, the presence of $\epsilon(\rho_{k,k'})$ in $G_{n}(L,l)$ implies that we have no guarantee on the trainability of a parameter $\theta^\nu$ in $W$ if $\rho$ is maximally mixed on the qubits in the backwards light-cone.

From Theorem~\ref{thm2} we derive the following corollary for $m$-local cost functions, which guarantees the trainability of the ansatz for shallow circuits.
\begin{corollary}\label{cor2}
Consider the function $F_n(L,l)$. Let $O$ be an operator of the form~\eqref{eq:Oml}, as in Theorem~\ref{thm2}. If at least one term $c_i^2  \epsilon(\rho_{k,k'}) \epsilon(\widehat{O}_i)$  in the sum in \eqref{eq:varMaint22} vanishes no faster than $\Omega(1/\poly(n))$, and if the number of layers $L$ is $\OC(\log(n))$, then 
\small
\begin{equation}
    G_n(L,l) \in \Omega\left(\frac{1}{\poly(n)}\right)\,.\label{eq:varGlobalCor2}
\end{equation}
\normalsize
On the other hand, if at  least one term $c_i^2  \epsilon(\rho_{k,k'}) \epsilon(\widehat{O}_i)$ in the sum in~\eqref{eq:varMaint22} vanishes no faster than  $\Omega\left(1/2^{\poly(\log(n))}\right)$, and if the number of layers is $\OC(\poly(\log(n)))$, then 
\begin{equation}
    G_n(L,l) \in \Omega\left(\frac{1}{2^{\poly(\log(n))}}\right)\,.\label{eq:varGlobalCor3}
\end{equation}
\end{corollary}
Hence, when $L$ is $\OC(\poly(\log(n)))$  there is a transition region where the lower bound vanishes faster than polynomially, but slower than exponentially.

We finally justify the assumption of each block being a local $2$-design from the fact that shallow circuit depths lead to such local $2$-designs. Namely, it has been shown that one-dimensional $2$-designs have efficient quantum circuit descriptions, requiring  $\OC(m^2)$ gates to be exactly implemented~\cite{dankert2009exact}, or $\OC(m)$ to be approximately implemented~\cite{brandao2016local,harrow2018approximate}. Hence, an $L$-layered ansatz in which each block forms a $2$-design can be exactly implemented with a depth $D\in\OC(m^2 L)$, and approximately implemented with $D\in\OC(mL)$.  For the case of  two-dimensional connectivity, it has been shown that approximate $2$-designs require a  circuit depth of $\OC(\sqrt{m})$ to be implemented~\cite{harrow2018approximate}. Therefore, in this case the depth of the layered ansatz is  $D\in\OC(\sqrt{m}L)$. The latter shows that increasing the dimensionality of the circuit reduces the circuit depth needed to make each block a $2$-design.   

Moreover, it has been shown that the Alternating Layered Ansatz of Fig.~\ref{fig:f3} will form an approximate one-dimensional $2$-design on $n$ qubits if the number of layers is $\OC(n)$~\cite{harrow2018approximate}. Hence, for deep circuits, our ansatz behaves like a random circuit and we recover the barren plateau result of~\cite{mcclean2018barren} for both local and global cost functions.

\subsection*{Numerical Simulations}
\label{sec:Numerics}

As an important example to illustrate the cost-function-dependent barren plateau phenomenon, we consider quantum autoencoders~\cite{Romero,wan2017quantum,lamata2018quantum,pepper2019experimental, verdon2018universal}. In particular, the pioneering VQA proposed in Ref.~\cite{Romero} has received significant literature attention, due to its importance to quantum machine learning and quantum data compression. Let us briefly explain the algorithm in Ref.~\cite{Romero}. 

%\subsection*{Quantum autoencoder}
 
Consider a bipartite quantum system $AB$ composed of $n_{\text{A}}$ and $n_{\text{B}}$ qubits, respectively, and let $\{p_\mu, \ket{\psi_\mu}\}$ be an ensemble of pure states on $AB$. The goal of the quantum autoencoder is to train a gate sequence $V(\thv)$ to compress this ensemble into the $A$ subsystem, such that one can recover each state $\ket{\psi_\mu}$ with high fidelity from the information in subsystem $A$. One can think of $B$ as the ``trash'' since it is discarded after the action of $V(\thv)$.

To quantify the degree of  data compression, Ref.~\cite{Romero} proposed a cost function of the form:
\begin{align}\label{eq:CG-AE1}
    C_{\text{G}}' &=1- \Tr[\dya{\vec{0}}\rho_{\text{B}}^{\text{out}}] \\
\label{eq:CG-AE2}   & = \Tr[O_{\text{G}}' V(\thv)\rho_{\text{AB}}^{\text{in}}V(\thv)\ad]\,,
\end{align}
where $\rho_{\text{AB}}^{\text{in}}=\sum_\mu p_\mu \dya{\psi_\mu}$ is the ensemble-average input state, $\rho_{\text{B}}^{\text{out}}=\sum_{\mu}p_\mu\Tr_{\text{A}}[\dya{\psi_\mu'}]$ is the ensemble-average trash state, and $\ket{\psi_\mu '} =V(\thv ) \ket{\psi_\mu}$. Equation~\eqref{eq:CG-AE2} makes it clear that $C_{\text{G}}'$ has the form in \eqref{eq:cost_function}, and $O_{\text{G}}' = \id_{\text{AB}} - \id_{\text{A}} \ot \dya{\vec{0}}$ is a global observable of the form in \eqref{eq-global-O}. Hence, according to Corollary~\ref{cor1}, $C_{\text{G}}'$ exhibits a barren plateau for large $n_{\text{B}}$. (Specifically, Corollary~\ref{cor1} applies in this context when $n_{\text{A}} < n_{\text{B}}$). As a result, large-scale data compression, where one is interested in discarding large numbers of qubits, will not be possible with $C_{\text{G}}'$.

\begin{figure}[t]
	\centering
	\includegraphics[width=.6\columnwidth]{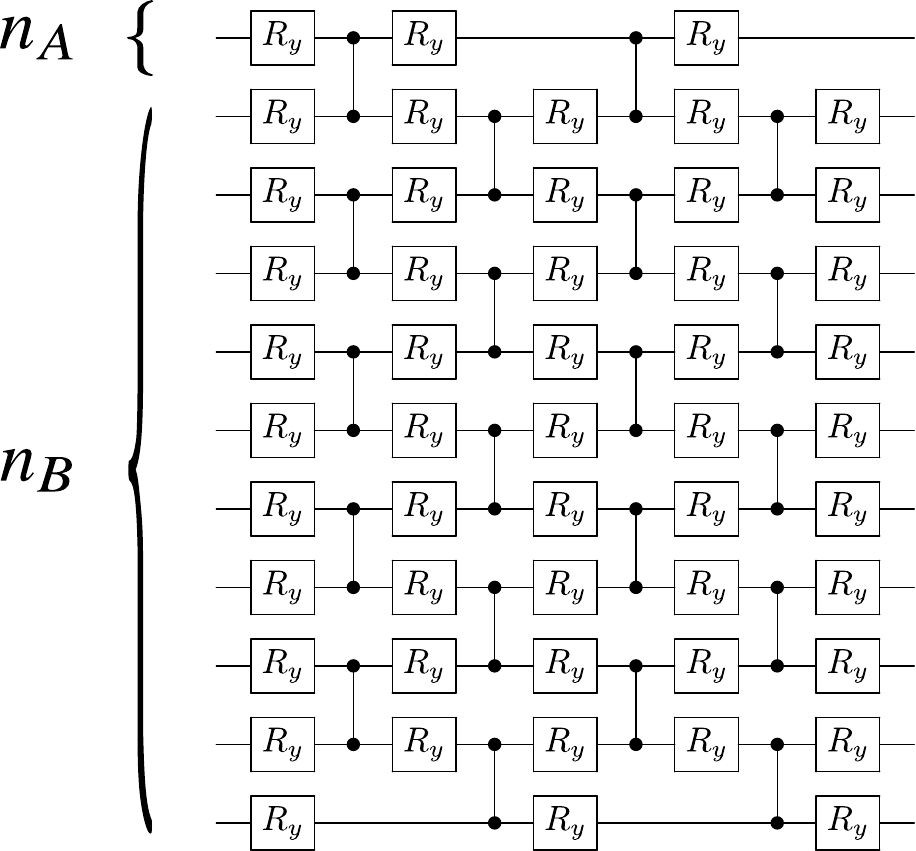}
	\caption{Alternating Layered Ansatz for $V(\vec{\theta})$ employed in our numerical simulations. Each layer is composed of control-$Z$ gates acting on alternating pairs of neighboring qubits which are preceded and followed by single qubit rotations around the $y$-axis, $R_y (\theta_i)= e^{-i\theta_i \sigma_y/2}$.  Shown is the case of two layers, $n_{\text{A}}=1$, and $n_{\text{B}}=10$ qubits. The number of variational parameters and gates scales linearly with $n_{\text{B}}$: for the case shown there are $71$ gates and $51$ parameters. While each block in this ansatz will not form an exact local $2$-design, and hence does not fall under our theorems,  one can still obtain a cost-function-dependent barren plateau. 
	\label{fig:f4}}
\end{figure}

To address this issue, we propose the following local cost function
\begin{align}\label{eq:CL-AE}
    C_{\text{L}}' &= 1- \frac{1}{n_{\text{B}}}\sum_{j=1}^{n_{\text{B}}}\Tr\left[\left(\dya{0}_j\otimes\id_{\overline{j}}\right)\rho_{\text{B}}^{\text{out}}\right] \\
    \label{eq:CL-AE2}   & = \Tr[O_{\text{L}}' V(\thv)\rho_{\text{AB}}^{\text{in}}V(\thv)\ad]\,,
\end{align}
where $O_{\text{L}}' = \id_{\text{AB}} - \frac{1}{n_{\text{B}}}\sum_{j=1}^{n_{\text{B}}} \id_{\text{A}}\ot \dya{0}_j\ot \id_{\overline{j}}$, and $\id_{\overline{j}}$ is the identity on all qubits in $B$ except the $j$-th qubit. As shown in the Supplementary Note 9, $C_{\text{L}}'$ satisfies $C_{\text{L}}'\leq C_{\text{G}}' \leq n_{\text{B}} C_{\text{L}}'$, which implies that $C_{\text{L}}'$ is faithful (vanishing under the same conditions as $C_{\text{G}}'$). Furthermore, note that $O_{\text{L}}'$ has the form in \eqref{eq:Oml}. Hence Corollary~\ref{cor2} implies that $C_{\text{L}}'$ does not exhibit a barren plateau for shallow ansatzes.

Here we simulate the autoencoder algorithm to solve a simple problem where $n_{\text{A}}=1$, and where the input state ensemble $\{p_\mu,\ket{\psi_\mu}\}$ is given by
\begin{align} 
\ket{\psi_1}&= \ket{0}_{\text{A}}\otimes\ket{0,0,0,\ldots,0}_{\text{B}}\,, \quad \text{with} \quad p_1=2/3\,, \label{eq:autoenc1}\\ \ket{\psi_2}&=\ket{1}_{\text{A}}\otimes \ket{1,1,0,\ldots,0}_{\text{B}}\,, \quad \text{with} \quad p_2=1/3\,. \label{eq:autoenc2}
\end{align}
In order to analyze the cost-function-dependent barren plateau phenomenon, the dimension of subsystem $B$ is gradually increased as $n_{\text{B}}=10,15,\ldots,100$.

\begin{figure*}
    \centering
    \includegraphics[width=\textwidth]{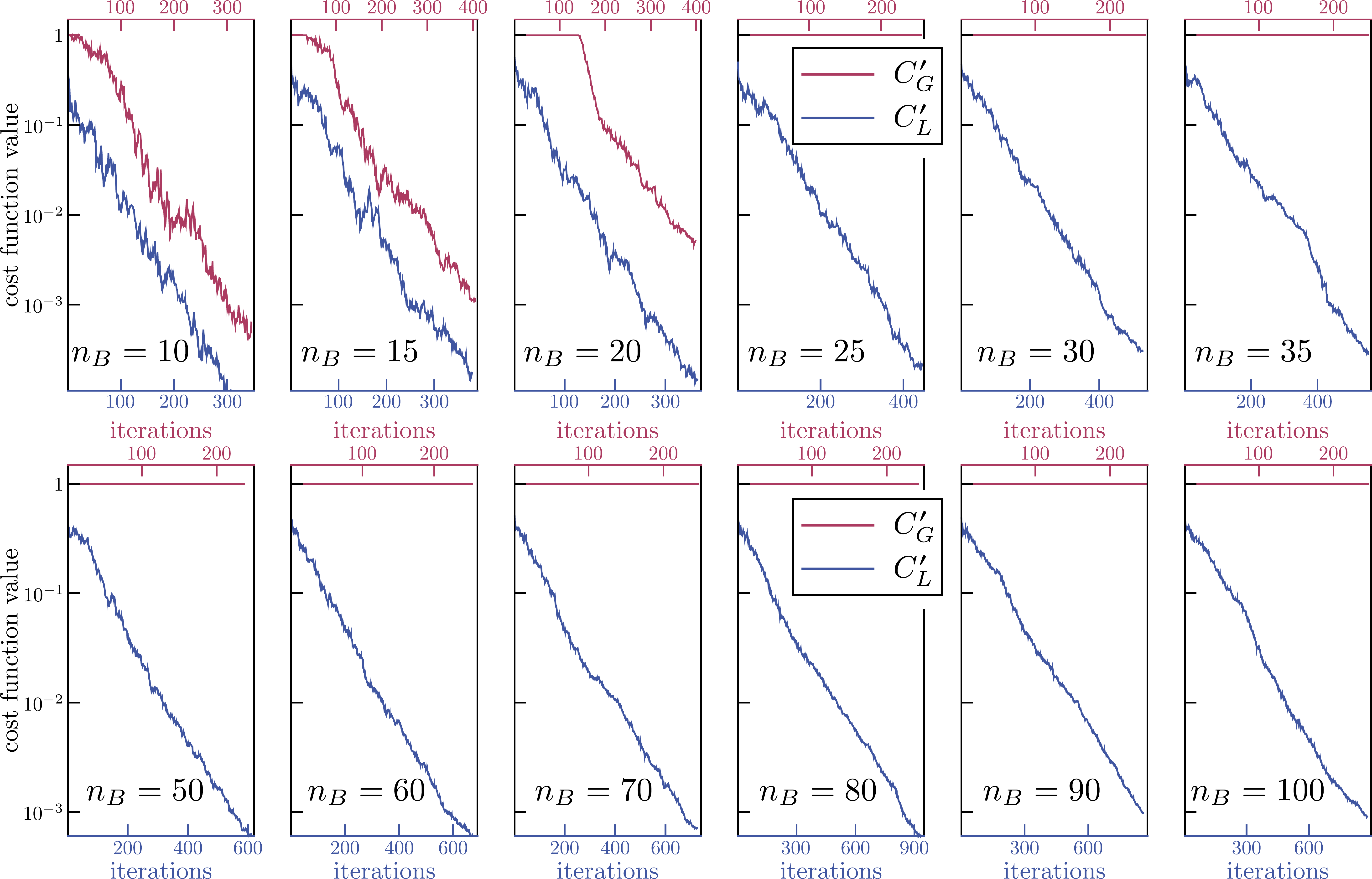}
    \caption{Cost versus number of iterations for the quantum autoencoder problem defined by Eqs.~\eqref{eq:autoenc1}--\eqref{eq:autoenc2}. In all cases we employed two layers of the ansatz shown in Fig.~\ref{fig:f4}, and we set $n_{\text{A}}=1$, while increasing $n_{\text{B}}=10,15,\ldots,100$. The top (bottom) axis corresponds to the global cost function $C_{\text{G}}'$ of Eq.~\eqref{eq:CG-AE2} (local cost function $C_{\text{L}}'$ of~\eqref{eq:CL-AE}). As can be seen, $C_{\text{G}}'$ can be trained up to $n_{\text{B}}=20$ qubits, while $C_{\text{L}}'$ trained in all cases. These results indicate that global cost function presents a barren plateau even for a shallow depth ansatz, and this can be avoided by employing a local cost function.}
    \label{fig:AutoencoderData}
\end{figure*}

\subsection*{Numerical results}

In our heuristics, the gate sequence $V(\thv)$ is given by two layers of the ansatz in Fig.~\ref{fig:f4}, so that  the number of gates and parameters in $V(\thv)$ increases linearly with $n_{\text{B}}$. Note that this ansatz is a simplified version of the ansatz in Fig.~\ref{fig:f3}, as we can only generate unitaries with real coefficients. All parameters in $V(\vec{\theta})$  were randomly initialized and as detailed in the Methods section, we employ a gradient-free training algorithm that gradually increases the number of shots per cost-function evaluation.

Analysis of the $n$-dependence. Figure~\ref{fig:AutoencoderData} shows representative results of our numerical implementations of the quantum autoencoder in Ref.~\cite{Romero} obtained by training $V(\thv)$ with the global and local cost functions respectively given by~\eqref{eq:CG-AE2} and~\eqref{eq:CL-AE}.  Specifically, while we train with finite sampling, in the figures we show the exact cost-function values versus the number of iterations. Here, the top (bottom) axis corresponds to the number of iterations performed while training with $C_{\text{G}}'$ ($C_{\text{L}}'$).
For $n_{\text{B}}=10$ and $15$, Fig.~\ref{fig:AutoencoderData} shows that we are able to train $V(\thv)$ for both cost functions. For $n_{\text{B}}=20$, the global cost function initially presents a plateau in which the optimizing algorithm is not able to determine a minimizing direction. However, as the number of shots per function evaluation increases, one can eventually minimize $C_{\text{G}}'$.  Such result indicates the presence of a barren plateau where the gradient takes small values which can only be detected when the number of shots becomes sufficiently large. In this particular example, one is able to start training at around 140 iterations.

When $n_{\text{B}}> 20$ we are unable to train the global cost function, while always being able to train our proposed local cost function. Note that the number of iterations is different for $C_{\text{G}}'$ and  $C_{\text{L}}'$, as for the global cost function case we reach the maximum number of shots in fewer iterations. These results indicate that the  global cost function of~\eqref{eq:CG-AE2} exhibits a barren plateau where the gradient of the cost function vanishes exponentially with the number of qubits, and which arises even for constant depth ansatzes. We remark that in principle one can always find a minimizing direction when training $C_{\text{G}}'$, although this would require a number of shots that increases exponentially with $n_{\text{B}}$.  Moreover, one can see in Fig.~\ref{fig:AutoencoderData} that randomly initializing the parameters always leads to $C_{\text{G}}'\approx 1$ due to the narrow gorge phenomenon (see Prop.~\ref{prop1}), i.e., where the probability of being near the global minimum vanishes exponentially with $n_{\text{B}}$. 

On the other hand, Fig.~\ref{fig:AutoencoderData} shows that the barren plateau is avoided when employing a local cost function since we can train $C_{\text{L}}'$ for all considered values of $n_{\text{B}}$. Moreover, as seen in Fig.~\ref{fig:AutoencoderData},  $C_{\text{L}}'$  can be trained with a small number of shots per cost-function evaluation  (as small as $10$ shots per evaluation).

Analysis of the $L$-dependence.  The power of Theorem~\ref{thm2} is that it gives the scaling in terms of $L$. While one can substitute a function of $n$ for $L$ as we did in Corollary~\ref{cor2}, one can also directly study the scaling with $L$ (for fixed $n$).  Figure~\ref{fig:L_dependence} shows the dependence on $L$ when training $C_{\text{L}}'$ for the autoencoder example with $n_{\text{A}} =1$ and $n_{\text{B}}=10$. As one can see, the training becomes more difficult as $L$ increases. Specifically, as shown in the inset it appears to become exponentially more difficult, as the number of shots needed to achieve a fixed cost value grows exponentially with $L$. This is consistent with (and hence verifies) our bound on the variance in Theorem~\ref{thm2}, which vanishes exponentially in $L$, although we remark that this behavior can saturate for very large $L$~\cite{mcclean2018barren}.

\begin{figure}[t]
    \centering
    \includegraphics[width=.95\columnwidth]{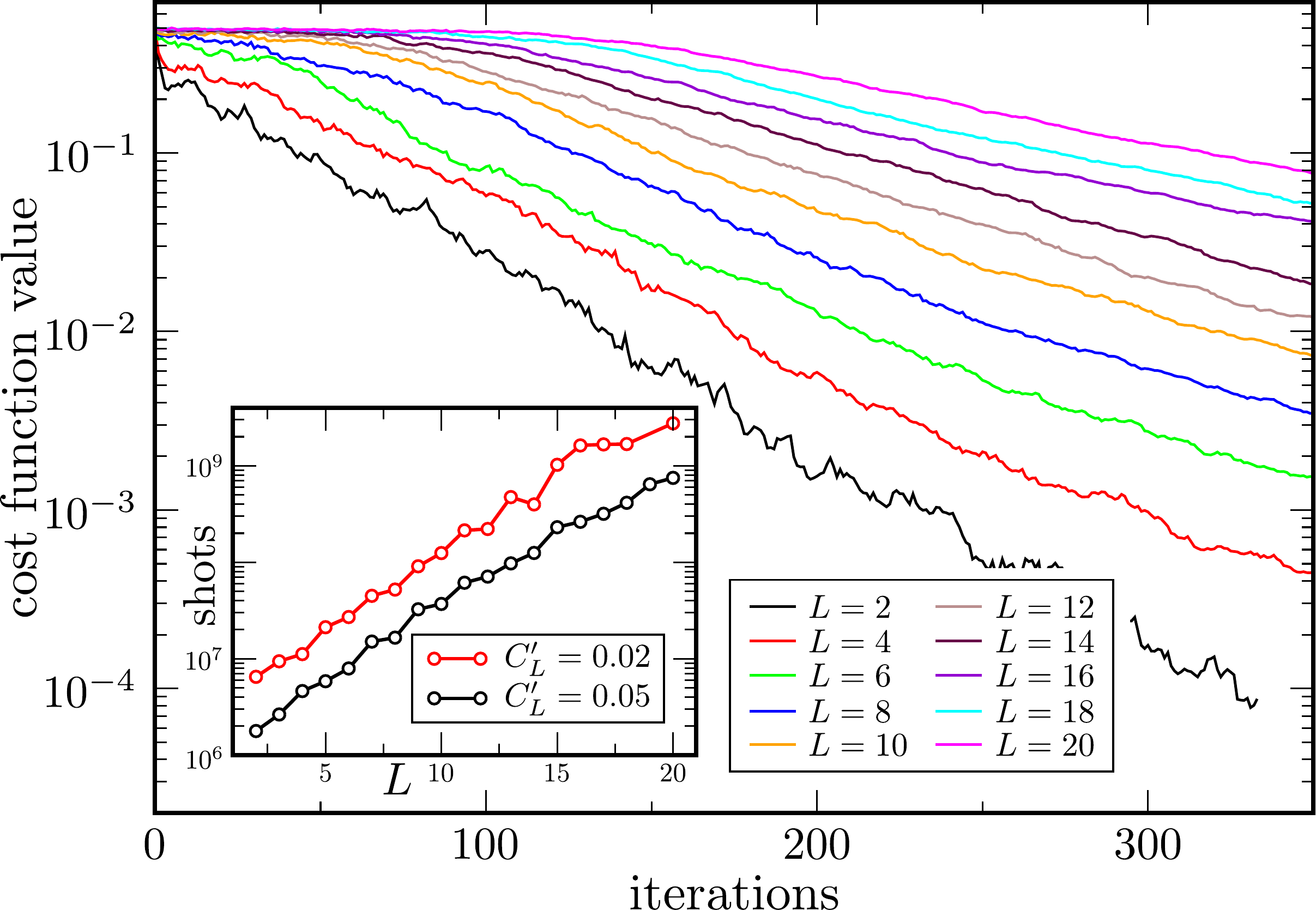}
    \caption{Local cost $C_{\text{L}}'$ versus number of iterations for the quantum autoencoder problem in Eqs.~\eqref{eq:autoenc1}--\eqref{eq:autoenc2} with $n_{\text{B}}=10$. Each  curve  corresponds to a  different number of layers $L$ in the ansatz of Fig.~\ref{fig:f4} with $L=2,\ldots,20$. Curves were averaged over $9$ instances of the autoencoder. As the number of layers increases, the optimization becomes harder. Inset: Number of shots needed to reach cost values of $C_{\text{L}}'=0.02$ and $C_{\text{L}}'=0.05$ versus number of layers $L$.  As $L$ increases the number of shots needed to reach the indicated cost values appears to increase exponentially.} 
    \label{fig:L_dependence}
\end{figure}

 In summary, even though the ansatz employed in our heuristics is beyond the scope of our theorems, we still find  cost-function-dependent barren plateaus, indicating that the cost-function dependent barren plateau phenomenon might be more general and go beyond our analytical results.

\section*{DISCUSSION}
\label{sec:Discussion}

 While scaling results have been obtained for classical neural networks~\cite{pennington2017geometry}, very few such results exist for the trainability of parametrized quantum circuits, and more generally for quantum neural networks. Hence, rigorous scaling results are urgently needed for VQAs, which many researchers believe will provide the path to quantum advantage with near-term quantum computers. One of the few such results is the barren plateau theorem of Ref.~\cite{mcclean2018barren}, which holds for VQAs with deep, hardware-efficient ansatzes.

 In this work, we proved that the barren plateau phenomenon extends to VQAs  with randomly initialized shallow Alternating Layered Ansatzes. The key to extending this phenomenon to shallow circuits was to consider the locality of the operator~$O$ that defines the cost function~$C$. Theorem~\ref{thm1} presented a universal upper bound on the variance of the gradient for global cost functions, i.e., when $O$ is a global operator. Corollary~\ref{cor1} stated the asymptotic scaling of this upper bound for shallow ansatzes as being exponentially decaying in $n$, indicating a barren plateau.  Conversely, Theorem~\ref{thm2} presented a universal lower bound on the variance of the gradient for local cost functions, i.e., when $O$ is a sum of local operators. Corollary~\ref{cor2} notes that for shallow ansatzes this lower bound decays polynomially in $n$. Taken together, these two results show that barren plateaus are cost-function-dependent, and they establish a connection between locality and trainability.

In the context of chemistry or materials science, our present work can inform researchers about which transformation to use when mapping a fermionic Hamiltonian to a spin Hamiltonian~\cite{tranter2018comparison}, i.e., Jordan-Wigner versus Bravyi-Kitaev~\cite{bravyi2002fermionic}. Namely, the Bravyi-Kitaev transformation often leads to more local Pauli terms, and hence (from Corollary~\ref{cor2}) to a more trainable cost function. This fact was recently numerically confirmed~\cite{uvarov2020variational}.

Moreover, the fact that Corollary~\ref{cor2} is valid for arbitrary input quantum states may be useful when constructing variational ansatzes. For example, one could propose a growing ansatz method where one appends $\log(n)$ layers of the hardware-efficient ansatz to a previously trained (hence fixed) circuit. This could then lead to a layer-by-layer training strategy where the previously trained circuit can correspond to multiple layers of the same hardware-efficient ansatz.

We remark that our definition of a global operator (local operator) is one that is both non-local (local) and many body (few body). Therefore, the barren plateau phenomenon could be due to the many-bodiness of the operator rather than the non-locality of the operator; we leave the resolution of this question to future work. On the other hand, our Theorem~\ref{thm1} rules out the possibility that barren plateaus could be due to cardinality, i.e., the number of terms in $O$ when decomposed as a sum of Pauli products~\cite{biamonte2019universal}. Namely, case (ii) of this theorem implies barren plateaus for $O$ of essentially arbitrary cardinality, and hence cardinality is not the key variable at work here.

We illustrated these ideas for two example VQAs. In Fig.~\ref{fig:f2}, we considered a simple state-preparation example, which allowed us to delve deeper into the cost landscape and uncover another phenomenon that we called a narrow gorge, stated precisely in Prop.~\ref{prop1}. In Fig.~\ref{fig:AutoencoderData}, we studied the more important example of quantum autoencoders, which have generated significant interest in the quantum machine learning community. Our numerics showed the effects of barren plateaus: for more than $20$  qubits we were unable to minimize the global cost function introduced in~\cite{Romero}. To address this, we introduced a local cost function for quantum autoencoders, which we were able to minimize for system sizes of up to $100$ qubits.

There are several directions in which our results could be generalized in future work. Naturally, we hope to extend the narrow gorge phenomenon in Prop.~\ref{prop1} to more general VQAs. In addition, we hope in the future to unify our theorems~\ref{thm1} and~\ref{thm2} into a single result that bounds the variance as a function of a parameter that quantifies the locality of $O$. This would further solidify the connection between locality and trainability. Moreover, our numerics suggest that our theorems (which are stated for exact $2$-designs) might be extendable in some form to ansatzes composed of simpler blocks, like approximate $2$-designs~\cite{brandao2016local}.

We emphasize that while our theorems are stated for a hardware-efficient ansatz and for costs that are of the form~\eqref{eq:cost_function}, it remains an interesting open question as to whether other ansatzes, cost function,  and architectures exhibit similar scaling behavior as that stated in our theorems. For instance, we have recently shown~\cite{sharma2020trainability} that our results can be extended to a more general type of Quantum Neural Network called dissipative quantum neural networks~\cite{beer2020training}. Another potential example of interest could be the  unitary-coupled cluster (UCC) ansatz in chemistry~\cite{bartlett2007coupled}, which is intended for use in the $\OC(\poly(n))$ depth regime~\cite{lee2018generalized}. Therefore it is important to study the key mathematical features of an ansatz that might allow one to go from trainability for $\OC(\log n)$ depth (which we guarantee here for local cost functions) to trainability for $\OC(\poly n)$ depth.

Finally, we remark that  some strategies have been developed to mitigate the effects of barren plateaus~\cite{volkoff2020large,verdon2019learning,grant2019initialization,skolik2020layerwise}. While these methods are promising and have been shown to work in certain cases, they  are still heuristic methods with no provable guarantees that they can work in generic scenarios. Hence, we believe that more work needs to be done to better understand how to prevent, avoid, or mitigate the effects of barren plateaus.

\section*{METHODS}
\label{sec:Methods}

In this section, we provide additional details for the results in the main text, as well as a sketch of the proofs for our main theorems. We note that the proof of Theorem~\ref{thm2} comes before that of Theorem~\ref{thm1} since the latter builds on the former. More detailed proofs of our theorems are given in the Supplementary Information.

\subsection*{Variance of the cost function partial derivative} \label{sebsec:Ansatz2}

Let us first discuss the formulas we employed to compute $\Var[\partial_\nu C]$. Let us first note that without loss of generality, any block $W_{kl}(\vec{\theta}_{kl})$ in the Alternating Layered Ansatz  can be written as a product of $\zeta_{kl}$ independent gates from a gate alphabet $\mathcal{A}=\{G_{\mu}(\theta)\}$ as 
\begin{equation}
W_{kl}(\vec{\theta}_{kl})=G_{\zeta_{kl}}(\theta_{kl}^{\zeta_{kl}})\ldots G_{\nu}(\theta^\nu_{kl})\ldots G_{1}(\theta^{1}_{kl})\,,     \label{eq;W-gates}
\end{equation}
where each $\theta_{kl}^\nu$ is a continuous parameter.
Here, $G_\nu(\theta^\nu_{kl})= R_{\nu}(\theta^\nu_{kl}) Q_{\nu}$ where $Q_{\nu}$ is an unparametrized gate and $R_{\nu}(\theta^\nu_{kl}) = e^{-i\theta^\nu_{kl} \sigma_\nu /2 }$ with $\sigma_\nu$ a Pauli operator.  Note that $W_{kL}$ denotes a block in the last layer of $V(\thv)$. 

For the proofs of our results, it is helpful to conceptually break up the ansatz as follows. Consider a block $W_{kl}(\vec{\theta}_{kl})$ in the $l$-th layer of the ansatz. For simplicity we henceforth use  $W$ to refer to a given $W_{kl}(\vec{\theta}_{kl})$. Let $S_w$  denote the $m$-qubit subsystem that contains the qubits $W$ acts on, and let $S_{\overline{w}}$ be the $(n-m)$ subsystem on which $W$ acts trivially. Similarly, let $\HC_w$ and $\HC_{\overline{w}}$ denote the Hilbert spaces associated with $S_w$ and $S_{\overline{w}}$, respectively.  Then, as shown in Fig.~\ref{fig:f3}(a), $V(\vec{\theta})$ can be expressed as
\begin{equation}
    V(\vec{\theta})=V_{\text{R}} (\id_{\overline{w}}\otimes W)V_{\text{L}}\,. \label{eq:Ansatz-W}
\end{equation}
Here, $\id_{\overline{w}}$ is the identity on $\HC_{\overline{w}}$, and $V_{\text{R}}$ contains the gates in the (forward) light-cone $\LC$ of $W$, i.e., all gates with at least one input qubit causally connected to the output qubits of $W$. The latter allows us to define $S_\LC$ as the subsystem of all qubits in $\LC$. 

Let us here recall that the Alternating Layered Ansatz can be implemented with either a 1D or 2D square connectivity as schematically depicted in  Fig.~\ref{fig:f3}(c). We remark that the following results are valid for both cases as the light-cone structure will be the same. Moreover, the notation employed in our proofs applies to both the  1D and 2D cases. Hence, there is no need to refer to the connectivity dimension in what follows. 

Let us now  assume that $\theta^\nu$ is a parameter inside a given block $W$, we obtain from~\eqref{eq:cost_function}, \eqref{eq;W-gates}, and~\eqref{eq:Ansatz-W}
\begin{align}
\partial_\nu C=&\frac{i}{2} \Tr\bigg[(\id_{\overline{w}}\otimes  W_{\text{B}})V_{\text{L}}\rho  V_{\text{L}}\ad (\id_{\overline{w}}\otimes  W_{\text{B}}\ad)\label{eq:grad-C}
\\&\times[\id_{\overline{w}}\otimes\sigma_\nu,(\id_{\overline{w}}\otimes   W_{\text{A}}\ad)V_{\text{R}}\ad O V_{\text{R}}(\id_{\overline{w}}\otimes  W_{\text{A}})]\bigg]\nonumber\,,
\end{align}
with 
\begin{equation}
    W_{\text{B}}=\prod_{\mu =1}^{\nu-1} G_{\mu}(\theta^{\mu})\,, \quad \text{and} \quad  W_{\text{A}}=\prod_{\mu =\nu}^{\zeta} G_{\mu}(\theta^{\mu})\,. \label{eq:WAWB}
\end{equation}

Finally, from~\eqref{eq:grad-C} we can derive a general formula for the variance:
\small
\begin{equation}
 \Var[\partial_\nu C]=\frac{2^{m-1}\Tr[\sigma_\nu^2]}{(2^{2m}-1)^2}\sum_{\substack{\vec{p}\vec{q}\\\vec{p}'\vec{q}'}}\langle\Delta\Omega_{\vec{q}\vec{p}}^{\vec{q}'\vec{p}'}\rangle_{V_{\text{R}}}\langle\Delta\Psi_{\vec{p}\vec{q}}^{\vec{p}'\vec{q}'}\rangle_{V_{\text{L}}}
\label{eq:Var-C}
\end{equation}
\normalsize
which holds if $W_{\text{A}}$ and $W_{\text{B}}$ form independent $2$-designs. Here, the summation runs over all bitstrings $\vec{p}$, $\vec{q}$, $\vec{p}'$, $\vec{q}'$ of length $2^{n-m}$. In addition, we defined 
\begin{align}
        \Delta\Omega_{\vec{q}\vec{p}}^{\vec{q}'\vec{p}'} &= \Tr[ \Omega_{\vec{q}\vec{p}}\Omega_{\vec{q}'\vec{p}'}] -\frac{\Tr[ \Omega_{\vec{q}\vec{p}}]\Tr[\Omega_{\vec{q}'\vec{p}'}]}{2^{m}}\,, \label{eq:Delta-Psi} \\
    \Delta\Psi_{\vec{p}\vec{q}}^{\vec{p}'\vec{q}'} &=  \Tr[\Psi_{\vec{p}\vec{q}}\Psi_{\vec{p}'\vec{q}'}]-\frac{\Tr[\Psi_{\vec{p}\vec{q}}]\Tr[\Psi_{\vec{p}'\vec{q}'}]}{2^{m}}\,,
\end{align}
 where $\Tr_{\overline{w}}$ indicates the trace over subsystem $S_{\overline{w}}$, and  $\Omega_{\vec{q}\vec{p}}$ and $\Psi_{\vec{q}\vec{p}}$ are operators on $\HC_w$ defined as 
\begin{align}
    \Omega_{\vec{q}\vec{p}}&=\Tr_{\overline{w}}\left[(\ketbra{\vec{p}}{\vec{q}}\otimes\id_{w})V_{\text{R}}\ad OV_{\text{R}}\right]\,,\label{eq:omegapq}\\
    \Psi_{\vec{p}\vec{q}}&=\Tr_{\overline{w}}\left[(\ketbra{\vec{q}}{\vec{p}}\otimes\id_{w})V_{\text{L}} \rho V_{\text{L}}\ad\right]\,.
\end{align}
We derive Eq.~\eqref{eq:Var-C} in the Supplementary Note 4.

\subsection*{Computing averages over $V$}
\label{sec:sec-haar}

Here we introduce the main tools employed to compute quantities of the form  $\langle\dots\rangle_V$. These tools are used throughout the proofs of our main results. 

Let us first remark that if the blocks in $V(\vec{\theta})$ are independent, then any  average over $V$ can be computed by averaging over the individual blocks, i.e.,   $\langle\dots\rangle_V=\langle\dots\rangle_{W_{11},\ldots, W_{kl},\ldots}=\langle\dots\rangle_{V_{\text{L}},W,V_{\text{R}}}$. For simplicity let us first consider the expectation value over a single block $W$ in the ansatz. In principle $\langle\dots\rangle_{W}$ can be approximated by varying the parameters in $W$ and sampling over the resulting $2^m\times 2^m$ unitaries. However, if $W$ forms a $t$-design, this procedure can be simplified as it is known that sampling over its distribution yields the same properties as sampling random unitaries from the unitary group with respect to the unique normalized Haar measure. 

Explicitly, the Haar measure is a uniquely defined left and right-invariant measure over the unitary group $d\mu(W)$, such that for any unitary matrix  $A\in U(2^m)$ and for any function $f(W)$ we have
\begin{equation}
\begin{aligned} 
    \int_{U(2^m)}d\mu(W) f(W)&=\int d\mu(W) f(A W)\\
    &=\int d\mu(W) f(W A)\,,
\end{aligned}
\end{equation}
where the integration domain is assumed to be $U(2^m)$ throughout this work. Consider a finite set $\{W_y\}_{y\in Y}$  (of size $|Y|$) of unitaries $W_y$, and let  $P_{(t,t)}(W)$ be an arbitrary polynomial of degree at most $t$ in the matrix elements of $W$ and at most $t$ in those of  $W\ad$. Then, this finite set is a $t$-design if~\cite{dankert2009exact} 
\begin{align}
    \langle P_{(t,t)}(W) \rangle_w&=\frac{1}{|Y|}\cdot \sum_{y\in Y}P_{(t,t)}(W_y)\nonumber\\
    &=\int d\mu(W) P_{(t,t)}(W)\,. \label{eq:t-design}
\end{align}

From the general form of $C$ in Eq.~\eqref{eq:cost_function} we can see the cost function is a polynomial of degree at most $2$ in the matrix elements of each block  $W_{kl}$ in $V(\vec{\theta})$, and at most $2$ in those of  $(W_{{kl}})\ad$. Then, if a given block $W$ forms a $2$-design, one can employ the following elementwise formula of the Weingarten calculus~\cite{collins2006integration,puchala2017symbolic} to explicitly evaluate averages over $W$ up to the second moment:
{\small
\begin{equation}
\begin{aligned}
    \int d\mu(W)w_{\vec{i}\vec{j}}w_{\vec{i}'\vec{j}'}^*&=\frac{\delta_{\vec{i}\vec{i}'}\delta_{\vec{j}\vec{j}'}}{2^m}        \label{eq:Haar2moment}\\
\!\!\int d\mu(W)w_{\vec{i}_1\vec{j}_1}w_{\vec{i}_2\vec{j}_2}w_{\vec{i}_1'\vec{j}_1'}^{*}w_{\vec{i}_2'\vec{j}_2'}^{*}&=\frac{1}{2^{2m}-1}\left(\Delta_1-\frac{\Delta_2}{2^m}\right)
\end{aligned}
\end{equation}}
where $w_{\vec{i}\vec{j}}$ are the matrix elements of $W$, and
\begin{equation}
\begin{aligned}
\Delta_1&=\delta_{\vec{i}_1\vec{i}_1'}\delta_{\vec{i}_2\vec{i}_2'}\delta_{\vec{j}_1\vec{j}_1'}\delta_{\vec{j}_2\vec{j}_2'}+\delta_{\vec{i}_1\vec{i}_2'}\delta_{\vec{i}_2\vec{i}_1'}\delta_{\vec{j}_1\vec{j}_2'}\delta_{\vec{j}_2\vec{j}_1'}\,,\\
\Delta_2&=\delta_{\vec{i}_1\vec{i}_1'}\delta_{\vec{i}_2\vec{i}_2'}\delta_{\vec{j}_1\vec{j}_2'}\delta_{\vec{j}_2\vec{j}_1'}+\delta_{\vec{i}_1\vec{i}_2'}\delta_{\vec{i}_2\vec{i}_1'}\delta_{\vec{j}_1\vec{j}_1'}\delta_{\vec{j}_2\vec{j}_2'}\,.    \label{eq:Haar2moment2}
\end{aligned}
\end{equation}

\begin{figure}[t]
    \centering
    \includegraphics[width=.6\columnwidth]{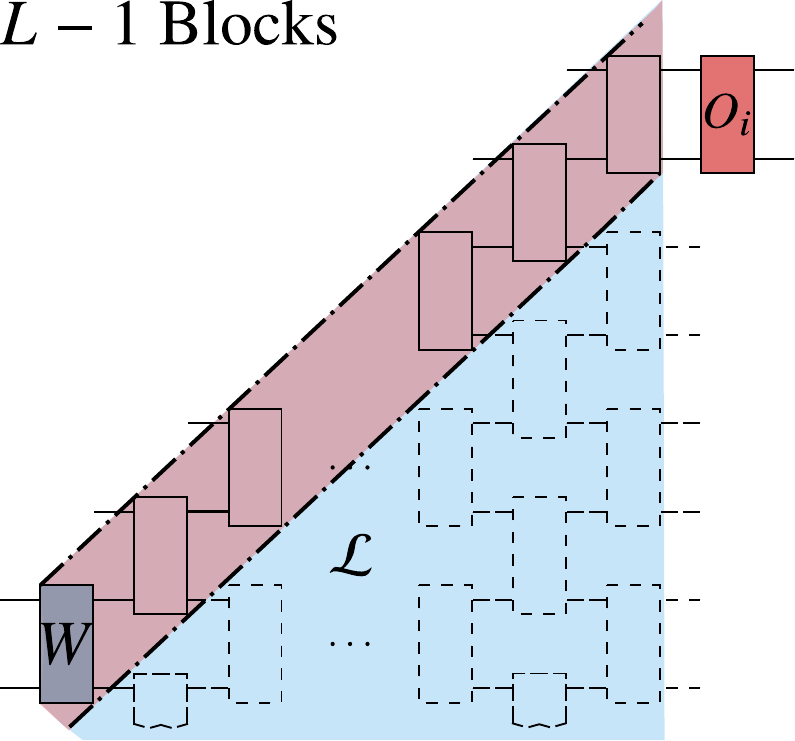}
    \caption{The block $W$ is in the first layer of $V(\thv)$, and the  operator $O_i$ acts on the topmost $m$ qubits in the forward light-cone $\mathcal{L}$ of $W$. Dashed thick lines indicate the backward light-cone of $O_i$.   All but $L-1$ blocks simplify to identity in $\Omega_{\vec{q}\vec{p}}$ of Eq.~\eqref{eq:omegapq}.}
    \label{fig:new}
\end{figure}

\subsection*{Intuition behind the main results}

The goal of this section is to provide some intuition for our main results. Specifically, we show here how the scaling of the cost function variance can be related to the number of blocks we have to integrate to compute $\langle\cdots\rangle_{V_{\text{R}},V_{\text{L}}}$, the locality of the cost functions, and with the number of layers in the ansatz. 

First, we recall from Eq.~\eqref{eq:Haar2moment} that integrating over a block leads to a coefficient of the order $1/2^{2m}$. Hence, we see that the more blocks one integrates over, the worse the scaling can be.

We now generalize the warm-up example. Let $V(\thv)$ be a single  layer of the alternating ansatz of Fig.~\ref{fig:f3}, i.e., $V(\thv)$ is a tensor product of $m$-qubit blocks $W_{k}:=W_{k1}$, with $k=1,\ldots,\xi$ (and with $\xi=n/m$), so that $\theta^\nu$ is in the  block $W_{k'}$. 
In the Supplementary Note 5 we generalize this scenario to the when the input state is not $\ket{\vec{0}}$, but instead is an arbitrary state $\rho$.  

From~\eqref{eq:Var-C}, the partial derivative of the global cost function in~\eqref{eq:Cost-TD} can be expressed as
\small
\begin{align}
 \Var[\partial_\nu C_{\text{G}}]=\upsilon\!\prod_{k\neq k'}\!\left\langle\!\Tr\left[\dya{0}^{\otimes m}W_{k}\dya{0}^{\otimes m}W_{k}\ad\right]^2\right\rangle_{\!W_{k}}\label{eq:globalint}
\end{align}
\normalsize
where $\upsilon=\frac{(2^m-1)^2\Tr[\sigma_\nu^2]}{2^{2m}(2^{m+1}-1)^2}$. From~\eqref{eq:globalint} we have that in order to compute~\eqref{eq:globalint} one needs to integrate over $\xi-1$ blocks. Then, since each integration leads to a coefficient $1/2^{2m}$ the variance will  scale as $\OC(1/(2^{2m})^{\xi-1}=\OC(1/2^{2n})$. Hence, the scaling of the variance gets worse for each block we integrate (such that  the block acts on qubits we are measuring).

On the other hand, for a local cost let us consider a single term in~\eqref{eq:localcostfunction22} where $j\in S_{\tilde{k}}$, so that  
\small
\begin{align}
 \Var[\partial_\nu C_{\text{L}}]\!\propto\!\frac{\upsilon}{n^2}\!\left\langle\!\Tr\left[(\dya{0}_j\otimes\id_{\overline{j}})W_{\tilde{k}}\dya{0}^{\otimes m}W_{\tilde{k}}\ad\right]^2\right\rangle_{\!W_{\tilde{k}}}\,.\label{eq:localint}
\end{align}
\normalsize
Here, in contrast to the global case, we only have to integrate over a single block irrespective of the total number of qubits. Hence, we now find that the variance scales as $\OC(1/n^2)$, where we remark that the scaling is essentially given by the prefactor $1/n^2$ in~\eqref{eq:localcostfunction22}.

Let us now briefly provide some intuition as to why the scaling of local cost gradients becomes exponentially vanishing with the number of layers as in Theorem~\ref{thm2}. Consider the case when  $V(\thv)$ contains $L$ layers of the ansatz in Fig.~\ref{fig:f3}. Moreover, as shown in Fig.~\ref{fig:new}, let  $W$ be in the first layer,  and let  $O_i$ act on the $m$ topmost qubits of $\mathcal{L}$.  As schematically depicted in Fig.~\ref{fig:new}, we now have to integrating over $L-1$ blocks. Then, as proved in the Supplementary Note 5, integrating over a block leads to a coefficient  $2^{m/2}/(2^m+1)$. Hence,  after integrating $L-1$ times,  we obtain a coefficient $2^{m(L-1)/2}/(2^m+1)^{L-1}$ which vanishes no faster than $\Omega\left(1/\poly(n)\right)$ if $mL\in \OC(\log(n))$.

As we discuss below, for more general scenarios the computation of $\Var[\partial_\nu C]$ becomes more complex.

\begin{figure*}[htp!]
    \centering
    \includegraphics[width=.9\linewidth]{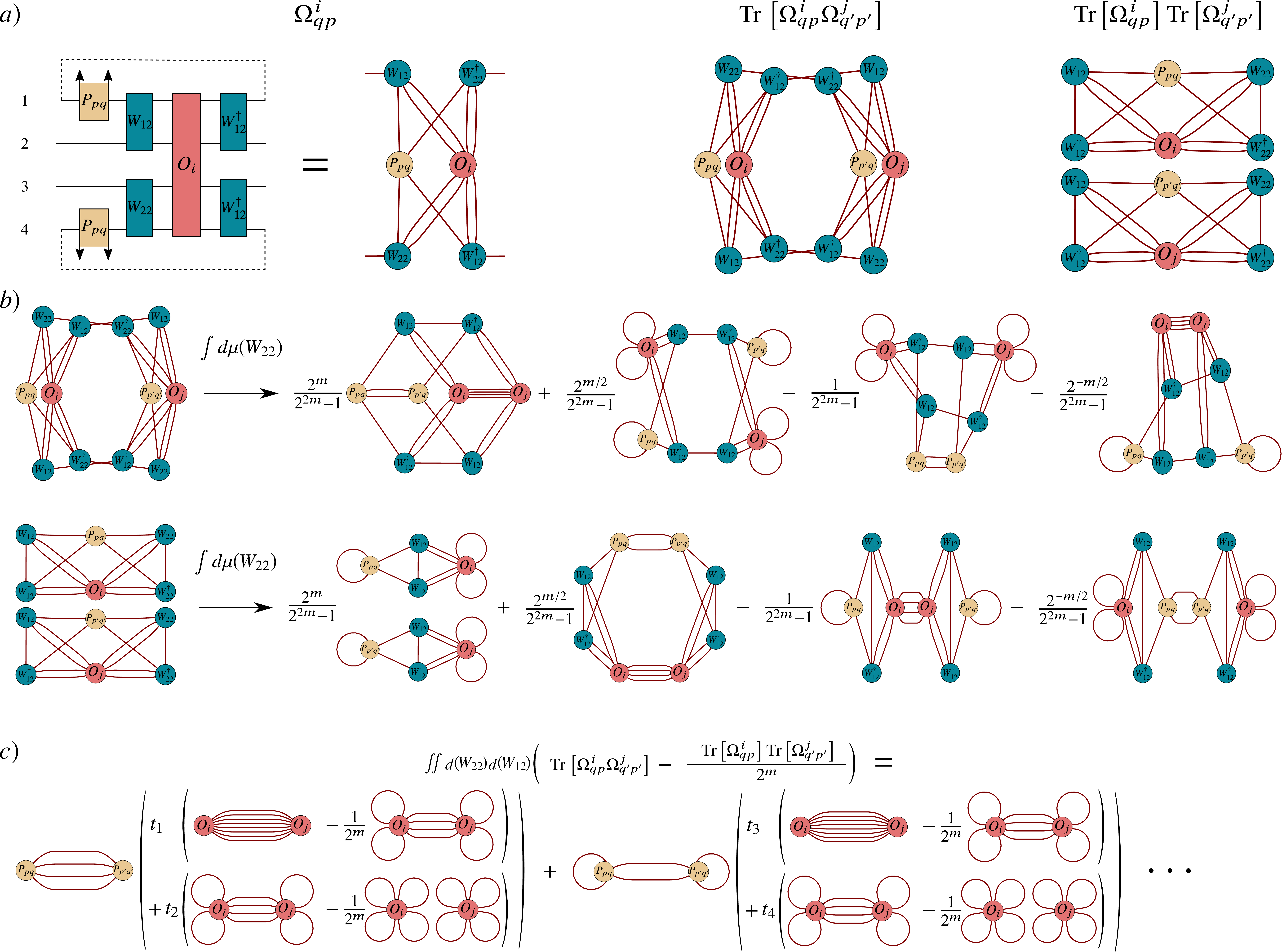}
    \caption{Tensor-network representations of the terms relevant to $\Var[\partial_\nu C]$. a)  Representation of $\Omega_{\vec{q}\vec{p}}^i$ of Eq.~\eqref{eq:omegapq} (left), where the superscript indicates that $O$ is  replaced by $O_i$. In this illustration, we show the case of $n=2m$ qubits, and  we denote $P_{\vec{p}\vec{q}}=\ketbra{\vec{q}}{\vec{p}}$. We also show the representation of $\Tr[ \Omega_{\vec{q}\vec{p}}^i\Omega_{\vec{q}'\vec{p}'}^j]$ (middle) and of  $\Tr[ \Omega_{\vec{q}\vec{p}}^i]\Tr[\Omega_{\vec{q}'\vec{p}'}^j]$ (right). (b) By means of the Weingarten calculus we can algorithmically integrate over each block in the ansatz. After each integration one obtains four new tensors according to Eq.~\eqref{eq:Haar2moment}. Here we show the tensor obtained after the integrations $\int d\mu(W) \Tr[ \Omega_{\vec{q}\vec{p}}^i\Omega_{\vec{q}'\vec{p}'}^j]$ and  $\int d\mu(W) \Tr[ \Omega_{\vec{q}\vec{p}}^i]\Tr[\Omega_{\vec{q}'\vec{p}'}^j]$, which are needed to compute   $\langle\Delta\Omega_{\vec{q}\vec{p}}^{\vec{q}'\vec{p}'}\rangle_{V_{\text{R}}}$ as in Eq.~\eqref{eq:Var-C}. c) As shown in the Supplementary Note 1, the result of the integration is a sum of the form~\eqref{eq:varVLC}, where the deltas over $\vec{p}$,  $\vec{q}$, $\vec{p}'$, and $\vec{q}'$ arise from the contractions between $P_{\vec{p}\vec{q}}$ and $P_{\vec{p}'\vec{q}'}$.}
    \label{fig:f5}
\end{figure*}

\subsection*{Sketch of the proof of the main theorems}

Here we present a sketch of the proof of Theorem~\ref{thm1} and Theorem~\ref{thm2}. We refer the reader to the Supplementary Information for a detailed version of the proofs.

As mentioned in the previous subsection, if each block in $V(\thv)$ forms a local $2$-design, then we can explicitly calculate expectation values $\langle\dots\rangle_{W}$ via~\eqref{eq:Haar2moment}. Hence, to compute  $\langle\Delta\Omega_{\vec{q}\vec{p}}^{\vec{p}'\vec{q}'}\rangle_{V_{\text{R}}}$, and $\langle\Delta\Psi_{\vec{p}\vec{q}}^{\vec{p}'\vec{q}'}\rangle_{V_{\text{L}}}$ in~\eqref{eq:Var-C}, one needs to algorithmically integrate over each block using the Weingarten calculus. 
In order to make such computation tractable, we employ the tensor network representation of quantum circuits.

For the sake of clarity, we recall that any two qubit gate can be expressed as $   U=\sum_{ijkl}U_{ijkl}\ketbra{ij}{kl}$, where $U_{ijkl}$ is a  $2\!\times\!2\!\times\!2\!\times\!2$ tensor. Similarly, any block in the ansatz can be considered as a $2^{\frac{m}{2}}\!\times2^{\frac{m}{2}}\!\times2^{\frac{m}{2}}\!\times2^{\frac{m}{2}}$ tensor. 
As schematically shown in Fig.~\ref{fig:f5}(a), one can use the circuit description of  $\Omega_{\vec{q}\vec{p}}^i$ and $\Psi_{\vec{p}\vec{q}}$ to derive the tensor network representation of terms such as $\Tr[ \Omega_{\vec{q}\vec{p}}^i\Omega_{\vec{q}'\vec{p}'}^j]$. Here, $\Omega_{\vec{q}\vec{p}}^i$ is obtained from~\eqref{eq:omegapq} by simply replacing $O$ with $O_i$. 

In Fig.~\ref{fig:f5}(b)  we depict an example where we employ the tensor network representation of $\Omega_{\vec{q}\vec{p}}^i$  to compute the average of $\Tr[ \Omega_{\vec{q}\vec{p}}^i\Omega_{\vec{q}'\vec{p}'}^j]$, and $\Tr[ \Omega_{\vec{q}\vec{p}}^i]\Tr[\Omega_{\vec{q}'\vec{p}'}^j]$. As expected, after each integration one obtains a sum of four new tensor networks according to Eq.~\eqref{eq:Haar2moment}.

\subsection*{Proof of Theorem~\ref{thm2} }
Let us first consider an $m$-local cost function $C$ where $O$ is given by~\eqref{eq:Oml}, and where $\widehat{O}_{i}$ acts non trivially in a given subsystem $S_k$ of $\mathcal{S}$. In particular, when $\widehat{O}_i$ is of this form the proof is simplified, although the more general proof is presented in the Supplementary Note 6. If $S_{k}\not\subset S_\LC$ we find $\Omega_{\vec{q}\vec{p}}^i\propto \id_w$, and hence 
\begin{equation}
    \Tr[ \Omega_{\vec{q}\vec{p}}^i\Omega_{\vec{q}'\vec{p}'}^j] -\frac{\Tr[ \Omega_{\vec{q}\vec{p}}^i]\Tr[\Omega_{\vec{q}'\vec{p}'}^j]}{2^{m}}=0\,.
\end{equation}
The latter implies that we only have to consider the operators $\widehat{O}_i$ which act on qubits inside of the forward light-cone $\LC$ of $W$.

Then, as shown  in the Supplementary Information 
\begin{equation}
    \left\langle\Tr[ \Omega_{\vec{q}\vec{p}}^i\Omega_{\vec{q}'\vec{p}'}^i] -\frac{\Tr[ \Omega_{\vec{q}\vec{p}}^i]\Tr[\Omega_{\vec{q}'\vec{p}'}^i]}{2^{m}}\right\rangle_{V_{\text{R}}}\propto 
    \epsilon(\widehat{O}_i)  \,.
    \label{eq:Tn_Sum-local}
\end{equation}  
Here we remark that the proportionality factor contains terms of the form $\delta_{(\vec{p},\vec{q})_{S_{\overline{w}}^+}}\delta_{(\vec{p}',\vec{q}')_{ S_{\overline{w}}^+}}\delta_{(\vec{p},\vec{q}')_{S_{\overline{w}}^-}}\delta_{(\vec{p}',\vec{q})_{S_{\overline{w}}^-}}$ (where $S_{\overline{w}}^+\cup S_{\overline{w}}^-=S_{\overline{w}}$), which arises from the different tensor contractions of $P_{\vec{p}\vec{q}}=\ketbra{\vec{q}}{\vec{p}}$ in Fig.~\ref{fig:f5}(c). It is then straightforward to show that 
\begin{align}
    \sum_{\substack{\vec{p}\vec{q}\\\vec{p}'\vec{q}'}}& \delta_{(\vec{p},\vec{q})_{S_{\overline{w}}^+}}\delta_{(\vec{p}',\vec{q}')_{ S_{\overline{w}}^+}}\delta_{(\vec{p},\vec{q}')_{S_{\overline{w}}^-}}\delta_{(\vec{p}',\vec{q})_{S_{\overline{w}}^-}}
    \left\langle\Delta\Psi_{\vec{p}\vec{q}}^{\vec{p}'\vec{q}'}\right\rangle_{V_{\text{L}}}\nonumber\\
    &= \left\langle D_{\HS}\left(\tilde{\rho}^{-}, \Tr_w[\tilde{\rho}^{-}]\otimes\frac{\id}{2^m}\right)\right\rangle_{V_{\text{L}}}\,, \label{eq:DHS}
\end{align}
where we define $\tilde{\rho}^{-}$ as the reduced states of $\tilde{\rho}=V_{\text{L}}\rho V_{\text{L}}\ad$ in the Hilbert spaces associated with subsystems $S_w\cup S_{\overline{w}}^-$. Here we recall that $D_{HS}$ is the Hilbert-Schmidt distance $D_{HS}\left(\rho,\sigma\right)=\Tr[(\rho-\sigma)^2]$.

By employing properties of $D_{HS}$ one can show (see Supplementary Note 6)
\begin{equation}
    D_{\HS}\left(\tilde{\rho}^{-}, \Tr_w[\tilde{\rho}^{-}]\otimes\frac{\id}{2^m}\right)\geq \frac{D_{\HS}\left(\widetilde{\rho}_w,\frac{\id}{2^m}\right)}{2^{m(L-l+2)/2}}\,,
\end{equation}
where $\tilde{\rho}_{w}=\Tr_{S_{\overline{w}}^-}[\tilde{\rho}^{-}]$. We can then leverage the tensor network representation of quantum circuits to algorithmically integrate over each block in $V_{\text{L}} $ and compute $\langle D_{\HS}\left(\widetilde{\rho}_w,\frac{\id}{2^m}\right)\rangle_{V_{\text{L}}}$. One finds
\begin{align}
    \left\langle D_{\HS}\left(\tilde{\rho}_w,\frac{\id}{2^m}\right)\right\rangle_{V_{\text{L}}}&=\sum_{\substack{(k,k')\in k_{\LC_{\text{B}}} \\ k' \geq k}}   t_{k,k'} \epsilon(\rho_{k,k'})\,,
\end{align}
with $t_{k,k'}\geq \frac{2^{ml}}{(2^m+1)^{2l}}$ $\forall k,k'$, and $\epsilon(\rho_{k,k'})$ defined in Theorem~\ref{thm2}. Combining these results leads to Theorem~\ref{thm2}. Moreover, as detailed in the Supplementary information, Theorem~\ref{thm2} is also valid when $O$ is of the form~\eqref{eq:Oml}.

\subsection*{Proof of Theorem~\ref{thm1} }

Let us now provide a sketch of the proof of Theorem~\ref{thm1}, case (i). Here we denote for simplicity $\widehat{O}_{k}:=\widehat{O}_{1k}$. We leave the proof of case (ii) for the Supplementary Note 7.  In this case there are now operators $O_i$ which act outside of the forward light-cone $\LC$ of $W$. Hence, it is convenient to include in $V_{\text{R}}$ not only all the gates in $\LC$ but also all the blocks in the final layer of $V(\thv)$ (i.e., all blocks $W_{kL}$, with $k=1,\ldots\xi$). We can  define $S_{\LCb}$ as the compliment of $S_\LC$, i.e., as the subsystem of all qubits which are not in $\LC$ (with associated Hilbert-space $\HC_{\LCb}$). Then,  we have 
$V_{\text{R}}=V_\LC\otimes V_{\LCb}$ and $\ketbra{\vec{q}}{\vec{p}}=\ketbra{\vec{q}}{\vec{p}}_\LC\otimes \ketbra{\vec{q}}{\vec{p}}_{\LCb}$,
where we define $V_{\LCb}:=\bigotimes_{k\in k_{\LCb}} W_{kL}$,  $\ketbra{\vec{q}}{\vec{p}}_\LC:=\bigotimes_{k\in k_{\LC}}\ketbra{\vec{q}}{\vec{p}}_{k}$, and $ \ketbra{\vec{q}}{\vec{p}}_{\LCb}:=\bigotimes_{k'\in k_{\LCb}}\ketbra{\vec{q}}{\vec{p}}_{k'}$.
Here, we define  $k_{\LC}:= \{k : S_{k}\subseteq S_{\LC}\}$ and $k_{\LCb}:= \{k : S_{k}\subseteq S_{\LCb}\}$, which are the set of indices whose associated qubits are inside and outside $\LC$, respectively. We also write $O=c_0\id+c_1\hat{O}_\LC\otimes\hat{O}_{\LCb}$,
where we define $\hat{O}_\LC := \bigotimes_{k\in k_\LC} \widehat{O}_k$ and $\hat{O}_{\LCb}:=\bigotimes_{k'\in k_{\LCb}} \widehat{O}_{k'}$.

Using the fact that the blocks in $V(\vec{\theta})$ are independent we can now compute $\langle\Delta\Omega_{\vec{q}\vec{p}}^{\vec{q}'\vec{p}'}\rangle_{V_{\text{R}}}=\langle\Delta\Omega_{\vec{q}\vec{p}}^{\vec{q}'\vec{p}'}\rangle_{V_{\LCb},V_\LC}$. Then, from the definition of $\Omega_{\vec{p}\vec{q}}$ in Eq.~\eqref{eq:omegapq} and the fact that one can always express 
\begin{align}
\left\langle\Delta\Omega_{\vec{q}\vec{p}}^{\vec{q}'\vec{p}'}\right\rangle_{V_{\text{R}}}&= \left\langle\Tr[ \Omega_{\vec{q}\vec{p}}^\LC\Omega_{\vec{q}'\vec{p}'}^\LC] -\frac{\Tr[ \Omega_{\vec{q}\vec{p}}^\LC]\Tr[\Omega_{\vec{q}'\vec{p}'}^\LC]}{2^{m}}\right\rangle_{V_\LC}\nonumber\\
&\times\left(\prod_{k\in k_{\LCb}}\left\langle\Omega_k\right\rangle_{W_{kL}}\right)\,,
\label{eq:OmegaLC}
\end{align}
with
\begin{align}
   \Omega_{\vec{q}\vec{p}}^\LC&= \Tr_{\LC\cap\overline{w}}\left[(\ketbra{\vec{p}}{\vec{q}}\otimes\id_w)V_\LC O^\LC V_\LC\right]\nonumber\\
   \Omega_k&=\Tr\left[\ketbra{\vec{p}}{\vec{q}}_{k}W_{kL}\ad \widehat{O}_{k} W_{kL} \right]\Tr\left[\ketbra{\vec{p}'}{\vec{q}'}_{k}W_{kL}\ad \widehat{O}_{k} W_{kL} \right]\nonumber
\end{align}
and where  $\Tr_{\LC\cap\overline{w}}$ indicates the  partial trace over the Hilbert-space associated with the qubits in $S_\LC\cap S_{\overline{w}}$.
As detailed in the Supplementary Information we can use Eq.~\eqref{eq:Haar2moment} to show that 
{\small
\begin{align}
\left\langle\Omega_k\right\rangle_{W_{kL}}\leq \frac{r_k^2\left(\delta_{(\vec{p},\vec{q})_{S_{k}}}\delta_{(\vec{p'},\vec{q'})_{S_{k}}}+\delta_{(\vec{p},\vec{q'})_{S_{k}}}\delta_{(\vec{p'},\vec{q})_{S_{k}}}\right)}{2^{2m}-1}\,.\label{eq:varWKL}
\end{align}}

On the other hand, as shown in the Supplementary Note 7 (and as schematically depicted in Fig.~\ref{fig:f5}(c)),  when computing the expectation value $\langle\dots\rangle_{V_\LC}$ in~\eqref{eq:OmegaLC}, one obtains
\begin{align}
   \left\langle\Tr[ \Omega_{\vec{q}\vec{p}}^\LC\Omega_{\vec{q}'\vec{p}'}^\LC] -\frac{\Tr[ \Omega_{\vec{q}\vec{p}}^\LC]\Tr[\Omega_{\vec{q}'\vec{p}'}^\LC]}{2^{m}}\right\rangle_{V_\LC}&=\sum_\tau t_\tau^{ij} \Delta O^{\LC}_\tau\delta_\tau\,,
   \label{eq:varVLC}
\end{align}
where we defined $\delta_\tau=\delta_{(\vec{p},\vec{q})_{S_{\overline{\tau}}}}\delta_{(\vec{p}',\vec{q}')_{S_{\overline{\tau}}}}\delta_{(\vec{p},\vec{q}')_{S_\tau}}\delta_{(\vec{p}',\vec{q})_{S_\tau}}$,  $t_\tau\in \mathbb{R}$, $S_\tau\cup S_{\overline{\tau}}=S_\LC\cap S_{\overline{w}}$ (with $S_\tau\neq\emptyset)$, and 
\begin{equation}
\begin{aligned}
\Delta O^{\LC}_\tau=&\Tr_{x_\tau y_\tau}\left[\Tr_{z_\tau}\left[O_i\right]\Tr_{z_\tau}\left[O_j\right]\right]\\
&-\frac{\Tr_{x_\tau}\left[\Tr_{y_\tau z_\tau}\left[O_i\right]\Tr_{y_\tau z_\tau}\left[O_j\right]\right]}{2^m}\,.
\end{aligned}
\end{equation}
Here we use the notation $\Tr_{x_\tau}$ to indicate the trace over the Hilbert space associated with subsystem $S_{x_\tau}$, such that  $S_{x_\tau}\cup S_{y_\tau}\cup S_{z_\tau}=S_\LC$.  As shown in the Supplementary Note 7, combining the deltas in Eqs.~\eqref{eq:varWKL}, and~\eqref{eq:varVLC} with $ \left\langle\Delta\Psi_{\vec{p}\vec{q}}^{\vec{p}'\vec{q}'}\right\rangle_{V_{\text{L}}}$ leads to  Hilbert-Schmidt distances between two quantum states as in~\eqref{eq:DHS}. One can then use the following bounds  $D_{\HS}\left(\rho_1,\rho_2\right)\leq 2$, $\Delta O_\tau^\LC\leq \prod_{k\in k_\LC} r_k^2$, and $\sum_\tau t_\tau \leq 2$, along with some additional simple algebra explained in the Supplementary Information to obtain the upper bound in Theorem~\ref{thm1}.

\subsection*{Ansatz and optimization method}

Here we describe the gradient-free optimization method used in our heuristics. First, we note that all the parameters in the ansatz are randomly initialized. Then, at each iteration, one solves the following sub-space search problem: $\min_{\vec{s}\in \mathbb{R}^d} C(\vec{\theta}+ \vec{A}\cdot\vec{s})$, where $\vec{A}$ is a randomly generated isometry, and $\vec{s}=(s_1,\ldots,s_d)$ is a vector of coefficients to be optimized over. We used $d=10$ in our simulations. Moreover, the training algorithm gradually increases the number of shots per cost-function evaluation. Initially, $C$ is evaluated with $10$ shots, and once the optimization reaches a plateau, the number of shots is increased by a factor of $3/2$. This process is repeated until a termination condition on the value of $C$ is achieved, or until we reach the maximum value of $10^5$ shots per function evaluation.  While this is a simple variable-shot approach, we remark that a more advanced variable-shot optimizer can be found in Ref.~\cite{kubler2020adaptive}. 

Finally, let us remark that while we employ a sub-space search algorithm, in the presence of barren plateaus all optimization methods will (on average) fail unless the algorithm has a precision (i.e., number of shots) that grows exponentially with $n$. The latter is due to the fact that an exponentially vanishing gradient implies that on average the cost function landscape will essentially be flat, with the slope of the order of $\OC(1/2^n)$. Hence, unless one has a precision that can detect such small changes in the cost value, one will not be able to determine a cost minimization direction with gradient-based, or even with black-box optimizers such as the Nelder-Mead method~\cite{nelder1965simplex}.

\section*{ACKNOWLEDGMENTS}
We thank Jacob Biamonte, Elizabeth Crosson, Burak Sahinoglu, Rolando Somma, Guillaume Verdon and Kunal Sharma for helpful conversations. All authors were supported by the Laboratory Directed Research and Development (LDRD) program of Los Alamos National Laboratory (LANL) under project numbers 20180628ECR (for MC), 20190065DR (for AS, LC, and PJC), and 20200677PRD1 (for TV). MC and AS were also supported by the Center for Nonlinear Studies at LANL. PJC acknowledges initial support from the LANL ASC Beyond Moore's Law project. This work was also supported by the U.S. Department of Energy (DOE), Office of Science, Office of Advanced Scientific Computing Research, under the Quantum Computing Application Teams program.

\section*{AUTHOR CONTRIBUTIONS}
The project was conceived by MC, LC, and PJC. The manuscript was written by MC, AS, TV, LC, and PJC. TV proved Proposition 1.  MC and AS proved Proposition 2 and Theorems 1--2. MC, AS, TV, and PJC proved Corollaries 1--2. MC, AS, TV, LC, and PJC analyzed the quantum autoencoder.  For the numerical results, TV performed the simulation in Fig.~2, and LC performed the simulation in Fig.~5. 

\section*{DATA AVAILABILITY}
Data generated and analyzed during current study are available from the corresponding author upon reasonable request.

\section*{COMPETING INTERESTS}
The authors declare no competing interests.

%\footnote{See Supplementary Information which includes Refs.~\cite{paley_zygmund_1932,Fukuda_2019,nielsen_chuang}.}
%\nocite{SM}

%\bibliography{ref.bib}

%merlin.mbs apsrev4-1.bst 2010-07-25 4.21a (PWD, AO, DPC) hacked
%Control: key (0)
%Control: author (0) dotless jnrlst
%Control: editor formatted (1) identically to author
%Control: production of article title (0) allowed
%Control: page (1) range
%Control: year (0) verbatim
%Control: production of eprint (0) enabled
%

\newpage

\onecolumngrid

\setcounter{section}{0}
\setcounter{proposition}{0}
\setcounter{figure}{0}
\setcounter{corollary}{0}
\renewcommand{\figurename}{SUP FIG.}

\section*{\Large{Supplementary Information for ``Cost Function Dependent Barren Plateaus in Shallow Parametrized Quantum Circuits'' }}

In this Supplementary Information, we present detailed proofs of the propositions, theorems, and corollaries presented in the manuscript ``Cost Function Dependent Barren Plateaus in Shallow Parametrized Quantum Circuits''. In Supplementary Note~\ref{secSM:Haar} we first introduce  several lemmas which will be useful in the derivation of our main results.  Then, in Supplementary Note~\ref{secSM:Prop1}, and Supplementary Note~\ref{secSM:Prop2} we respectively provide proofs for Propositions~1 and~2 of the main text. We then derive the general equations for the variance of the cost function partial derivative in Supplementary Note~\ref{secSM:Vari}. In Supplementary Note~\ref{secSM:tensorproduct} we explicitly evaluate the variance of the cost function derivative for  the special case when $V(\thv)$ is given by a single layer of the Alternating Layered Ansatz.

In Supplementary Note~\ref{sec:proofTheorem2}, and Supplementary Note~\ref{sec:proofTheorem1}, we provide our proofs to Theorem 2 and Theorem 1, respectively. Where we remark that the proof of Theorem~2 comes before that of Theorem~1 since the latter builds on the former.  Then, in Supplementary Note~\ref{sec:proofcor}, we prove Corollaries 1, and 2. In Supplementary Note~\ref{sec:faithfulness}, we demonstrate that the local cost function for the quantum autoencoder is faithful.

\renewcommand{\thesection}{\arabic{section}}
\renewcommand{\figurename}{Supplementary Figure }
\titleformat{\section}{\large\bfseries}{}{0pt}{Supplementary Note \thesection:\quad}
\titleformat{\figure}{\large\bfseries}{}{0pt}{Supplementary Figure \thesection:\quad}

\section{Preliminaries}
\label{secSM:Haar}
In this section, we present properties that allow for analytic calculation of integrals of polynomial functions over the unitary group with respect to the unique normalized Haar measure. For more details on this topic, we refer the reader to Ref.~\cite{collins2006integration,puchala2017symbolic}. In addition, to make the Supplementary Information more self-contained, we reiterate here the definition of a $t$-design.
Consider a finite set $\{W_y\}_{y\in Y}$  (of size $|Y|$) of unitaries $W_y$ on a $d$-dimensional Hilbert space, and let  $P_{(t,t)}(W)$ be an arbitrary polynomial of degree at most $t$ in the matrix elements of $W$ and at most $t$ in those of  $W\ad$. Then, we say that this finite set is a $t$-deisgn if~\cite{dankert2009exact}
\begin{equation}
    \frac{1}{|Y|}\cdot \sum_{y\in Y}P_{(t,t)}(W_y)=\int_{U(d)} d\mu(W) P_{(t,t)}(W)\,, \label{eq:SMt-design}
\end{equation}
where in the right-hand side $U(d)$ denotes the unitary group of degree $d$. Equation~\eqref{eq:SMt-design} implies  that averaging $P_{(t,t)}(W)$ over the $t$-design is indistinguishable from integrating over $U(d)$ with respect to the Haar distribution.

Given $W\in U(d)$ the following expressions are valid for the first two moments~\cite{collins2006integration,puchala2017symbolic}
\begin{align}
\label{eq:delta0}
    \int_{U(d)} d\mu(W)w_{\vec{i},\vec{j}}w_{\vec{p},\vec{k}}^*&=\frac{\delta_{\vec{i},\vec{p}}\delta_{\vec{j},\vec{k}}}{d}\,,   \\
\int_{U(d)} d\mu(W)w_{\vec{i}_1,\vec{j}_1}w_{\vec{i}_2,\vec{j}_2}w_{\vec{i}_1',\vec{j}_1'}^{*}w_{\vec{i}_2',\vec{j}_2'}^{*}&=\frac{1}{d^2-1}\left(\delta_{\vec{i}_1,\vec{i}_1'}\delta_{\vec{i}_2,\vec{i}_2'}\delta_{\vec{j}_1,\vec{j}_1'}\delta_{\vec{j}_2,\vec{j}_2'}+\delta_{\vec{i}_1,\vec{i}_2'}\delta_{\vec{i}_2,\vec{i}_1'}\delta_{\vec{j}_1,\vec{j}_2'}\delta_{\vec{j}_2,\vec{j}_1'}\right)\nonumber \\
&-\frac{1}{d(d^2-1)}\left(\delta_{\vec{i}_1,\vec{i}_1'}\delta_{\vec{i}_2,\vec{i}_2'}\delta_{\vec{j}_1,\vec{j}_2'}\delta_{\vec{j}_2,\vec{j}_1'}+\delta_{\vec{i}_1,\vec{i}_2'}\delta_{\vec{i}_2,\vec{i}_1'}\delta_{\vec{j}_1,\vec{j}_1'}\delta_{\vec{j}_2,\vec{j}_2'}\right)\,.
    \label{eq:delta}
\end{align}
All throughout this section the integration domain will be implied to be $U(d)$, and  unless otherwise specified we consider $W$ to be an operator acting on a Hilbert space $\HC_w$ of dimension $d$.  When $d=2^{m}$, as occurs for a Hilbert space of $m$ qubit, we adopt the symbol $\vec{i} = (i_1, \dots i_m)$ to denote a bitstring of length $m$ such that $i_1,i_2,\dotsc,i_{m}\in\{0,1\}$. Moreover, given two bitstrings $\vec{i}$ and $\vec{j}$ we define their concatenation as $\vec{i}\cdot \vec{j}=(i_1, \dots i_n,j_1, \dots j_n)$.

Operators in the computational basis of $m$ qubits can be written as
\begin{align*}
    \begin{split}
     W=\sum_{\vec{i},\vec{j}}w_{\vec{i},\vec{j}}\ketbra{\vec{i}}{\vec{j}}\,,\quad W\ad=\sum_{\vec{i}',\vec{j}'}w_{\vec{i}',\vec{j}'}^{*}\ketbra{\vec{j}'}{\vec{i}'}\,,\quad
     A=\sum_{\vec{k},\vec{l}}a_{\vec{k},\vec{l}}\ketbra{\vec{k}}{\vec{l}}\,,\quad B=\sum_{\vec{q},\vec{p}}b_{\vec{q},\vec{p}}\ketbra{\vec{q}}{\vec{p}}\, .
    \end{split}
\end{align*}

From the previous expressions we can derive the  following lemmas.

\begin{lemma}
\label{lemma1}
Let $\lbrace W_{y} \rbrace_{y\in Y}\subset U(d)$ form a unitary $t$-design with $t\ge 1$, and let $A ,B: \HC_w\to \HC_w$ be arbitrary linear operators. Then
\begin{equation}
    {1\over \vert Y\vert}\sum_{y\in Y}\Tr\left[W_{y}AW_{y}\ad B\right]=\int d\mu(W)\Tr\left[WAW\ad B\right]=\frac{\Tr\left[A\right]\Tr\left[B\right]}{d}\,.
    \label{eq:lemma1}
\end{equation}
\end{lemma}

\begin{proof}
The first equality follows from the definition of a $t$-design. Note that $\Tr\left[WAW\ad B\right]$ can be written as 
\begin{align*}
   \Tr\left[WAW\ad B\right]
   =\sum_{\vec{i}_1,\vec{j}_1,\vec{i}_1',\vec{j}_1'}a_{\vec{j}_1,\vec{j}_1'}b_{\vec{i}_1',\vec{i}_1}w_{\vec{i}_1,\vec{j}_1}w_{\vec{i}_1',\vec{j}_1'}^*\,.
\end{align*}
Then, from  (\ref{eq:delta0}), we have 
\begin{align*}
    \int d\mu(W)\Tr\left[WAW\ad B\right]=\frac{1}{d}\sum_{\vec{i}_1,\vec{j}_1}a_{\vec{j}_1,\vec{j}_1}b_{\vec{i}_1,\vec{i}_1}=\frac{\Tr\left[A\right]\Tr\left[B\right]}{d}\,. 
\end{align*}
\end{proof}

\begin{lemma}
\label{lemma2}
Let $\lbrace W_{y} \rbrace_{y\in Y}\subset U(d)$ form a unitary $t$-design with $t\ge 2$ and let $A,B,C, D:\HC_w \to \HC_w$ be arbitrary linear operators. Then 
\begin{equation}
\begin{split}
    {1\over \vert Y\vert}\sum_{y\in Y}\Tr[W_{y}AW_{y}\ad BW_{y}CW_{y}\ad D]=& \int d\mu(W)\Tr[WAW\ad BWCW\ad D]\\=&\frac{1}{d^2-1}\left(\Tr[A]\Tr[C]\Tr[BD]+\Tr[AC]\Tr[B]\Tr[D]\right)\\
    &-\frac{1}{d(d^2-1)}\left(\Tr[AC]\Tr[BD]+\Tr[A]\Tr[B]\Tr[C]\Tr[D]\right)\,.
    \end{split}
    \label{eq:lemma2}
\end{equation}
\end{lemma}
\begin{proof}
The first equality follows from that fact that $\Tr[W_{y}AW_{y}\ad BW_{y}CW_{y}\ad D]\in P_{(2,2)}(W_{y})$. By writting
\begin{align*}
    \Tr[WAW\ad BWCW\ad D]=\sum_{\substack{\vec{i}_1,\vec{j}_1,\vec{i}_1',\vec{j}_1'\\\vec{i}_2,\vec{j}_2,\vec{i}_2',\vec{j}_2'}}a_{\vec{j}_1,\vec{j}_1'}b_{\vec{i}_1',\vec{i}_2}c_{\vec{j}_2,\vec{j}_2'}d_{\vec{i}_2',\vec{i}_1}w_{\vec{i}_1,\vec{j}_1}w_{\vec{i}_2,\vec{j}_2}w_{\vec{i}_1',\vec{j}_1'}^{*}w_{\vec{i}_2',\vec{j}_2'}^{*}\,,
\end{align*}
we can use \eqref{eq:delta} to obtain
\begin{align}
\begin{split}
    \int d\mu(W)\Tr[WAW\ad BWCW\ad D]=&\frac{1}{d^2-1}\sum_{\vec{i}_1,\vec{j}_1,\vec{i}_2,\vec{j}_2}\left(a_{\vec{j}_1,\vec{j}_1}b_{\vec{i}_1,\vec{i}_2}c_{\vec{j}_2,\vec{j}_2}d_{\vec{i}_2,\vec{i}_1}+a_{\vec{j}_1,\vec{j}_2}b_{\vec{i}_2,\vec{i}_2}c_{\vec{j}_2,\vec{j}_1}d_{\vec{i}_1,\vec{i}_1}\right)\\
    &-\frac{1}{d(d^2-1)}\sum_{\vec{i}_1,\vec{j}_1,\vec{i}_2,\vec{j}_2}\left(a_{\vec{j}_1,\vec{j}_2}b_{\vec{i}_1,\vec{i}_2}c_{\vec{j}_2,\vec{j}_1}d_{\vec{i}_2,\vec{i}_1}+a_{\vec{j}_1,\vec{j}_1}b_{\vec{i}_2,\vec{i}_2}c_{\vec{j}_2,\vec{j}_2}d_{\vec{i}_1,\vec{i}_1}\right)\\
    =&\frac{1}{d^2-1}\left(\Tr[A]\Tr[C]\Tr[BD]+\Tr[AC]\Tr[B]\Tr[D]\right)\\
    &-\frac{1}{d(d^2-1)}\left(\Tr[AC]\Tr[BD]+\Tr[A]\Tr[B]\Tr[C]\Tr[D]\right)\,.
    \end{split}
\end{align}
\end{proof}

\begin{lemma}
\label{lemma3}
Let $\lbrace W_{y} \rbrace_{y\in Y}\subset U(d)$ form a unitary $t$-design with $t\ge 2$ and let $A,B,C, D: \HC_w\to \HC_w$ be arbitrary linear operators. Then   
\begin{equation}
\begin{split}
{1\over \vert Y\vert}\sum_{y\in Y}\Tr[W_{y}AW_{y}\ad B]\Tr[W_{y}CW_{y}\ad D]
&=\int d\mu(W)\Tr[WAW\ad B]\Tr[WCW\ad D]\nonumber \\
    =&\frac{1}{d^2-1}\left(\Tr[A]\Tr[B]\Tr[C]\Tr[D]+\Tr[AC]\Tr[BD]\right)\\
    &-\frac{1}{d(d^2-1)}\left(\Tr[AC]\Tr[B]\Tr[D]+\Tr[A]\Tr[C]\Tr[BD]\right)\,.
    \end{split}
    \label{eq:lemma3}
\end{equation}
\end{lemma}

\begin{proof}
The first equality follows from a reasoning similar to the one used in Lemma \ref{lemma2}. By expressing 
\begin{align*}
    \Tr\left[WAW\ad B\right]\Tr\left[WCW\ad D\right]=\sum_{\vec{\alpha},\vec{\beta}}\Tr\left[WAW\ad B|\vec{\alpha}\rangle\langle\vec{\beta}|WCW\ad D|\vec{\beta}\rangle\langle\vec{\alpha}|\right]\,,
\end{align*}
we can employ  \eqref{eq:lemma2} to obtain
\begin{align*}
    \int d\mu(W)\Tr[WAW\ad B]\Tr[WCW\ad D]=&\sum_{\vec{\alpha},\vec{\beta}}\int d\mu(W)\Tr\left[WAW\ad B|\vec{\alpha}\rangle\langle\vec{\beta}|WCW\ad D|\vec{\beta}\rangle\langle\vec{\alpha}|\right]\\
    =&\frac{1}{d^2-1}\sum_{\vec{\alpha},\vec{\beta}}\left(\Tr[A]\Tr[C]\langle\vec{\alpha}|B|\vec{\alpha}\rangle\langle\vec{\beta}|D|\vec{\beta}\rangle+\Tr[AC]\langle\vec{\beta}|B|\vec{\alpha}\rangle\langle\vec{\alpha}|D|\vec{\beta}\rangle\right)\\
    &-\frac{1}{d(d^2-1)}\sum_{\vec{\alpha},\vec{\beta}}\left(\Tr[AC]\langle\vec{\alpha}|B|\vec{\alpha}\rangle\langle\vec{\beta}|D|\vec{\beta}\rangle+\Tr[A]\Tr[C]\langle\vec{\beta}|B|\vec{\alpha}\rangle\langle\vec{\alpha}|D|\vec{\beta}\rangle\right)\\
    =&\frac{1}{d^2-1}\left(\Tr[A]\Tr[B]\Tr[C]\Tr[D]+\Tr[AC]\Tr[BD]\right)\\
    &-\frac{1}{d(d^2-1)}\left(\Tr[AC]\Tr[B]\Tr[D]+\Tr[A]\Tr[C]\Tr[BD]\right)\,.
\end{align*}
\end{proof}

\begin{lemma}
\label{lemma5}
Let $\HC=\HC_{\overline{w}}\otimes \HC_w$ be a bipartite Hilbert space of dimension $d=d_{\overline{w}}d_w$, and let $\lbrace W_{y}\rbrace_{y\in Y}$ be a unitary $t$-design with $t\ge 1$ such that $W_{y}\in U(d_w)$ for all $y\in Y$. Then for arbitrary linear operators $A,B:\HC\to \HC$, we have  
\begin{equation}
\int d\mu(W)(\id_{\overline{w}}\otimes W)A( \id_{\overline{w}}\otimes W\ad)B= \frac{\Tr_w\left[A\right]\otimes\id_w}{d_w}B\,,
    \label{eq:lemma5(1)}
\end{equation}
and 
\begin{equation}
    \int d\mu(W)\Tr\left[(\id_{\overline{w}}\otimes W)A( \id_{\overline{w}}\otimes W\ad)B\right] =\frac{1}{d_w}\Tr\left[\Tr_w\left[A\right]\Tr_w\left[B\right]\right]\,.
    \label{eq:lemma5(2)}
\end{equation}
\end{lemma}

Here we use the notation $\id_w$ to indicate the identity operator on subsystem $\HC_w$, and we employ $\Tr_w$ to indicate the partial trace over $\HC_w$.

\begin{proof}
First, note that
\begin{align*}
    (\id_{\overline{w}}\otimes W)A( \id_{\overline{w}}\otimes W\ad)B = \sum_{\vec{i},\vec{j},\vec{i}',\vec{j}'} w_{\vec{i},\vec{j}}w^*_{\vec{i}',\vec{j}'}(\id_{\overline{w}}\otimes \ketbra{\vec{i}}{\vec{j}})A(\id_{\overline{w}}\otimes \ketbra{\vec{j}'}{\vec{i}'})B\,,
\end{align*}
by using  (\ref{eq:delta0}) we have
\begin{align*}
    \int d\mu(W)(\id_{\overline{w}}\otimes W)A( \id_{\overline{w}}\otimes W\ad)B =& \sum_{\vec{i},\vec{j},\vec{i'},\vec{j'}} \int d\mu(W) w_{\vec{i},\vec{j}}w^*_{\vec{i'},\vec{j'}}(\id_{\overline{w}}\otimes \ketbra{\vec{i}}{\vec{j}})A(\id_{\overline{w}}\otimes \ketbra{\vec{j'}}{\vec{i'}})B\\
    =& \frac{1}{d_w} \sum_{\vec{i},\vec{j}} (\id_{\overline{w}}\otimes \ketbra{\vec{i}}{\vec{j}})A(\id_{\overline{w}}\otimes\ketbra{\vec{j}}{\vec{i}})B\\
    =&\frac{1}{d_w}\left(\Tr_w\left[A\right]\otimes\id_w\right)B\,.
\end{align*}
Finally we can also obtain 
\begin{align*}
    \int d\mu(W)\Tr\left[(\id_{\overline{w}}\otimes W)A( \id_{\overline{w}}\otimes W\ad)B\right] =\frac{1}{d_w}\Tr\left[\Tr_w\left[A\right]\Tr_w\left[B\right]\right]\,.
\end{align*}
\end{proof}

\begin{lemma}
\label{lemma4}
If $\HC=\HC_{\overline{w}}\otimes \HC_w$ is a bipartite Hilbert space of dimension $d=d_{\overline{w}}d_w$ ($d=2^n$, and $\overline{d}=2^{n'}$), and if $A,B:\HC \rightarrow \HC$ are arbitrary linear operators, then for any linear operator $W:\HC_w\rightarrow \HC_w$ we have 
\begin{align}
\Tr\left[(\id_{\overline{w}}\otimes W)A(\id_{\overline{w}}\otimes W\ad) B\right]=\sum_{\vec{p},\vec{q}}\Tr\left[W A_{\vec{q}\vec{p}}W\ad B_{\vec{p}\vec{q}}\right]\,,
    \label{eq:lemma4}
\end{align}
where the summation runs over all bitstrings of length $n'$, and where 
\begin{equation}
    A_{\vec{q}\vec{p}}=\Tr_{\overline{w}}\left[(\ketbra{\vec{p}}{\vec{q}}\otimes\id_w)A\right]\,, \quad B_{\vec{p}\vec{q}}=\Tr_{\overline{w}}\left[(\ketbra{\vec{q}}{\vec{p}}\otimes\id_w)B\right]\, .
\end{equation}
\end{lemma}

\begin{proof}
By expanding the operators in the computational basis
\begin{align}
     \id_{\overline{w}}\otimes W&=\sum_{\vec{p},\vec{i},\vec{j}}w_{\vec{i},\vec{j}}\ketbra{\vec{p}}{\vec{p}}\otimes\ketbra{\vec{i}}{\vec{j}}\,,\quad\quad  \id_{\overline{w}}\otimes W\ad=\sum_{\vec{q},\vec{i}',\vec{j}'}w_{\vec{i}',\vec{j}'}^{*}\ketbra{\vec{q}}{\vec{q}}\otimes\ketbra{\vec{j}'}{\vec{i}'}\,,\quad\\
     A&=\sum_{\substack{\vec{k}_1,\vec{k}_2\\\vec{l}_1,\vec{l}_2}}a_{\vec{k}_1\cdot \vec{k}_2,\vec{l}_1\cdot\vec{l}_2}\ketbra{\vec{k}_1}{\vec{l}_1}\otimes\ketbra{\vec{k}_2}{\vec{l}_2}\,,\quad B=\sum_{\substack{\vec{p}_1,\vec{p}_2\\\vec{q}_1,\vec{q}_2}}b_{\vec{p}_1\cdot \vec{p}_2,\vec{q}_1\cdot\vec{q}_2}\ketbra{\vec{p}_1}{\vec{q}_1}\otimes\ketbra{\vec{p}_2}{\vec{q}_2}\,,
\end{align}
we have
\begin{equation}
    \Tr\left[(\id_{\overline{w}}\otimes W)A(\id_{\overline{w}}\otimes W\ad) B\right]=\sum_{\substack{\vec{i},\vec{j},\vec{i}',\vec{j}'\\\vec{p},\vec{q}}} w_{\vec{i},\vec{j}} a_{\vec{p}\cdot \vec{j},\vec{q}\cdot\vec{i'}} w_{\vec{i}',\vec{j}'} b_{\vec{q}\cdot \vec{j'},\vec{p}\cdot\vec{i}}\\
    =\sum_{\vec{p},\vec{q}}\Tr\left[W A_{\vec{q}\vec{p}}W\ad B_{\vec{p}\vec{q}}\right]\,.
\end{equation}
\end{proof}

\begin{lemma}
\label{lemma3.5}
Let $\HC=\HC_1\otimes \HC_2\otimes\HC_3\otimes\HC_4$  be a Hilbert space of dimension $d\cdot d \cdot d \cdot d=d^4$, and let $W:\HC_1\otimes \HC_2\to \HC_1\otimes \HC_2$ be a $t$-design with $t\geq 2$. For arbitrary linear operators $A,A':\HC\to \HC$ we  define
\begin{align}
    \Omega_{1}&= ( W \otimes\id_3\otimes\id_4)A( W\ad\otimes\id_3\otimes\id_4)(\ketbra{\vec{p}}{\vec{q}}\otimes\id_2\otimes\id_3\otimes\id_4)\,, \,\,\, \Omega_{2}= ( W \otimes\id_3\otimes\id_4)A( W\ad\otimes\id_3\otimes\id_4)\,, \\ 
    \Omega_{1}'&= (W \otimes\id_3\otimes\id_4)A'( W\ad\otimes\id_3\otimes\id_4)(\ketbra{\vec{p}'}{\vec{q}'}\otimes\id_2\otimes\id_3\otimes\id_4)\,, \,\,\, \Omega_{2}'= ( W \otimes\id_3\otimes\id_4)A'( W\ad\otimes\id_3\otimes\id_4)\,,
\end{align}
where $\vec{p}$, $\vec{q}$, $\vec{p}'$, and $\vec{q}'$ are bitstrings of length $d$. 

Let us first consider the quantity
\begin{equation}
   \Delta \Omega^{(1)}_{jkl}=\Tr_{kl}\left[\Tr_{1j}[\Omega_1] \Tr_{1j}[\Omega'_1]\right]-\frac{\Tr_{\text{L}}[\Tr_{1jk}[\Omega_1]_{1jk}\Tr[\Omega'_1]]}{d_k}\,, \label{eq:lemma6Proj}
\end{equation}
where  $\Tr_j$ indicates the partial trace over $\HC_j$, and where $\Tr_{jk} :=  \Tr_{j}\Tr_{k}$. Here the indexes $j,k$, and $l$ define Hilbert spaces such that  $\HC_j\otimes\HC_k\otimes\HC_{\text{L}}=\HC_2\otimes\HC_3\otimes\HC_4$. In addition, let us define $\HC_j=\HC_0 := \{\vec{0}\}$, in which case for any operator $O$ we also define  the partial trace over $\HC_0$ as  $\Tr_{0}\left[O\right]:=O$, and  $\Tr_{k0}\left[O\right]:=\Tr_k\left[O\right]$. From the previous definition, we then have  that $\HC_k$ can be chosen from the set  $\{\HC_2,\HC_3,\HC_2\otimes\HC_3,\HC_3\otimes\HC_4\}$, and for any such choice 
the following equality always holds
\begin{align}
\begin{aligned}
\int d\mu(W) \Delta\Omega^{(1)}_{jkl}=&\frac{t_1^{(1)} \delta_{\vec{p}\vec{q}'}\delta_{\vec{p}'\vec{q}}}{d^4-1}\left(\Tr_{123x}[\Tr_y[A]\Tr_y[A']]-\frac{\Tr_{3x}[\Tr_{12y}[A]\Tr_{12y}[A']]}{d^2}\right)\\
 &+\frac{t_2^{(1)} \delta_{\vec{p}\vec{q}}\delta_{\vec{p}'\vec{q}'}}{d^4-1}\left(\Tr_{3x}[\Tr_{12y}[A]\Tr_{12y}[A']]-\frac{\Tr_x[\Tr_{123y}[A]\Tr_{123y}[A']]}{d'}\right)\\
 &-\frac{t_3^{(1)} \delta_{\vec{p}\vec{q}'}\delta_{\vec{p}'\vec{q}}}{d^4-1}\left(\Tr_{12x}[\Tr_{3y}[A]\Tr_{3y}[A']]-\frac{\Tr_x[\Tr_{123y}[A]\Tr_{123y}[A']]}{d^2}\right)\\
 &-\frac{t_4^{(1)} \delta_{\vec{p}\vec{q}}\delta_{\vec{p}'\vec{q}'}}{d^4-1}\left(\Tr_{123x}[\Tr_{y}[A]\Tr_{y}[A']]-\frac{\Tr_{12x}[\Tr_{3y}[A]\Tr_{3y}[A']]}{d'}\right) \label{eq:lemma3.5}
\end{aligned}
\end{align}
where $d'=d^2$ if $\HC_4\leq\HC_k$, and  $d'=d$ otherwise. Here we employ the notation $A\leq B$ to indicate that $A$ is a subspace of $B$. Note that $\HC_y=\HC_4$, $\HC_x=\HC_0$ if $\HC_4\leq \HC_j$, and $\HC_y=\HC_0$, $\HC_x=\HC_4$ otherwise. In addition, we have
{\small
\begin{align}
  t_1^{(1)} &=
    \begin{cases}
      d & \text{if $\HC_2\leq\HC_j$}\\
      0 & \text{if $\HC_k=\HC_2$}\\
      d^2 & \text{otherwise}
    \end{cases}\,, \quad
    t_2^{(1)} =
    \begin{cases}
      d^2 & \text{if $\HC_2\leq\HC_j$}\\
      0 & \text{if $\HC_k=\HC_2$}\\
      d & \text{otherwise}
    \end{cases}\,, \quad
        t_4^{(1)} =
    \begin{cases}
      1 & \text{if $\HC_2\leq\HC_j$, and $\HC_3\in\HC_k$}\\
      0 & \text{if $\HC_k=\HC_1$}\\
      \frac{1}{d} & \text{otherwise}
    \end{cases} \\
    t_3^{(1)} &=
    \begin{cases}
      d & \text{if $\HC_k=\HC_3$, and $\HC_2\leq\HC_{\text{L}}$}\\
      1-d^2 & \text{if $\HC_k=\HC_2$}\\
      \frac{1}{d} & \text{if $\HC_k=\HC_2\otimes\HC_3$, or if $\HC_k=\HC_3\otimes\HC_4$, and $\HC_j=\HC_2$}\\
      1 & \text{otherwise}
    \end{cases}\,, 
\end{align}
}
where we can verify that for all $12$ possible cases that $\sum_{k=1}^4 |\frac{t_k^{(1)}}{(d^4-1)}|\leq 1$. 

Second, let us consider the quantity
\begin{equation}
   \Delta \Omega^{(2)}_{jkl}=\Tr_{kl}\left[\Tr_{1j}[\Omega_2] \Tr_{1j}[\Omega'_2]\right]-\frac{\Tr_{\text{L}}[\Tr_{12j}[\Omega_2]_{1jk}\Tr[\Omega'_2]]}{d_k}\,. \label{eq:lemma6NoProj}
\end{equation}
 Then, for $\HC_k\in\{\HC_2,\HC_3,\HC_2\otimes\HC_3,\HC_3\otimes\HC_4\}$, 
the following equality always holds
\begin{align}
\begin{aligned}
\int d\mu(W) \Delta\Omega^{(2)}_{jkl}=&\frac{t_1^{(2)}}{d^2+1} \left(\Tr_{123x}[\Tr_y[A]\Tr_y[A']]-\frac{\Tr_{3x}[\Tr_{12y}[A]\Tr_{12y}[A']]}{d^2}\right) \\
&+\frac{t_2^{(2)}}{d^2+1} \left(\Tr_{3x}[\Tr_{12y}[A]\Tr_{12y}[A']]-\frac{\Tr_x[\Tr_{123y}[A]\Tr_{123y}[A']]}{d'}\right)\\
 &-\frac{t_3^{(2)}}{d^2+1}\left(\Tr_{12x}[\Tr_{3y}[A]\Tr_{3y}[A']]-\frac{\Tr_x[\Tr_{123y}[A]\Tr_{123y}[A']]}{d^2}\right)\label{eq:lemma3.5.2}\,,
\end{aligned}
\end{align}
where $d'=d^2$ if $\HC_4\leq\HC_k$, and  $d'=d$ otherwise. Similarly, $\HC_y=\HC_4$, $\HC_x=\HC_0$ if $\HC_4\leq \HC_j$, and $\HC_y=\HC_0$, $\HC_x=\HC_4$ otherwise. In addition, we now have
{\small
\begin{align}
  t_1^{(2)} &=
    \begin{cases}
      0 & \text{if $\HC_2\leq\HC_j$}\\
      d & \text{otherwise}
    \end{cases}\,, \quad
    t_2^{(2)} =
    \begin{cases}
      0 & \text{if $\HC_2\leq\HC_j$, or $\HC_k=\HC_2$}\\
      \frac{d^2+1}{d} & \text{otherwise}
    \end{cases}\,, \quad
        t_3^{(2)} =
    \begin{cases}
      1 & \text{if $\HC_2\leq\HC_{\text{L}}$, and $\HC_k=\HC_3$}\\
      \frac{1}{d} & \text{if $\HC_{\text{L}}=\HC_2$, and $\HC_k=\HC_3\otimes\HC_4$}\\
      0 & \text{otherwise}
    \end{cases} \,,
\end{align}
}
where we can verify that for all $9$ nontrivial cases $\sum_{k=1}^3 |\frac{t_k^{(2)}}{d^2+1}|\leq 1$.

\end{lemma}

\begin{proof}
The proof of Eqs.~\eqref{eq:lemma3.5}, and~\eqref{eq:lemma3.5.2} can be obtain by explicitly integrating each term via~\eqref{eq:delta}. 
\end{proof}

In particular,  let us consider from Eq.~\eqref{eq:lemma3.5} the following case which will be relevant for our proofs. 
{\footnotesize
\begin{align}
\begin{aligned}
\int d\mu(W)\left(\Tr_{34}\left[\Tr_{12}[\Omega_1] \Tr_{12}[\Omega_1']\right]-\frac{\Tr[\Omega_1]\Tr[\Omega_1']}{d^2}\right)=&\frac{d^2 \delta_{\vec{p}_2\vec{q}_2'}\delta_{\vec{p}_2'\vec{q}_2}}{d^4-1}\left(\Tr_{234}[\Tr_1[A]\Tr_1[A']]-\frac{\Tr_{4}[\Tr_{123}[A]\Tr_{123}[A']]}{d^2}\right)\\
 &+\frac{d \delta_{\vec{p}_2\vec{q}_2}\delta_{\vec{p}_2'\vec{q}_2'}}{d^4-1}\left(\Tr_{4}[\Tr_{123}[A]\Tr_{123}[A']]-\frac{\Tr[A]\Tr[A']}{d}\right)\\
 &-\frac{\delta_{\vec{p}_2\vec{q}_2'}\delta_{\vec{p}_2'\vec{q}_2}}{d(d^4-1)}\left(\Tr_{23}[\Tr_{14}[A]\Tr_{14}[A']]-\frac{\Tr[A]\Tr[A']}{d^2}\right)\\
 &-\frac{\delta_{\vec{p}_2\vec{q}_2}\delta_{\vec{p}_2'\vec{q}_2'}}{d(d^4-1)}\left(\Tr_{234}[\Tr_{1}[A]\Tr_{1}[A']]-\frac{\Tr_{23}[\Tr_{14}[A]\Tr_{14}[A']]}{d}\right) \,.\label{eq:lemma6eq1}
\end{aligned}
\end{align}
}

\begin{lemma}
\label{lemma6}
Let $\HC=\HC_1\otimes \HC_2\otimes \HC_3$ be a tripartite Hilbert space of dimension $d=d_1d_2d_3$, and let $O_1=O_{\text{A}}\otimes\id_2\otimes\id_3$, and $O_2=\id_1\otimes O_{\text{B}}\otimes\id_3$ be linear operators on $S$. Here $\id_i$ indicates the identity over subsystem $S_i$, so that $O_1$ and $O_2$ have no overlapping support. Then for any linear operators $O_{\text{A}}:S_{1}\rightarrow S_{1}$, and $O_{\text{B}}:S_{2}\rightarrow S_{2}$ we have 
\begin{align}
\Tr_{jk}\left[\Tr_i\left[O_1\right]\Tr_i\left[O_2\right]\right]-\frac{\Tr_{k}\left[\Tr_{ij}\left[O_1\right]\Tr_{ij}\left[O_2\right]\right]}{d_j}
=0\,, 
    \label{eq:lemma6}
\end{align}
where $\Tr_i$ indicates the partial trace over $\HC_i$, $\Tr_{ij} :=  \Tr_{i}\Tr_{j}$, and where we defined $\HC_i\otimes\HC_j\otimes \HC_k=\HC$ such that $\HC_j=\HC_1,\HC_2$ or $\HC_3$. Moreover, one can always choose $\HC_i=\HC_0 := \{\vec{0}\}$ (or $\HC_k=\HC_0$), in which case we define the partial trace over $\HC_0$ as  $\Tr_{0}\left[O\right]:=O$, and  $\Tr_{j0}\left[O\right]:=\Tr_j\left[O\right]$.
\end{lemma}

\begin{proof}
Let us show how this equality holds for the specific case when $\HC_i=\{\vec{0}\}$, $\HC_j=\HC_2$, and $\HC_k=\HC_1\otimes \HC_3$, and let us remark that all remaining  cases follow similarly We have
\begin{equation}
 \Tr_0\left[O_1\right]\Tr_0\left[O_2\right]=O_{\text{A}}\otimes O_{\text{B}}\otimes\id_3\,,    \quad   \Tr_{2}\left[O_1\right]=d_2O_{\text{A}}\otimes\id_3\,,\quad \Tr_2\left[O_2\right]=\Tr_2\left[O_{\text{B}}\right]\id_1\otimes \id_3\,,
\end{equation}
and Eq.~\eqref{eq:lemma6} becomes
\begin{align}
\Tr_{123}\left[O_1O_2\right]-\frac{\Tr_{13}\left[\Tr_{2}\left[O_1\right]\Tr_{2}\left[O_2\right]\right]}{d_2}&=\Tr_{123}\left[O_{\text{A}}\otimes O_{\text{B}}\otimes\id_3\right]-\frac{d_2\Tr_2\left[O_{\text{B}}\right] \Tr_{13}\left[O_{\text{A}}\otimes\id_3\right]}{d_2}\\
&=d_3\Tr_1\left[O_{\text{A}}\right]\Tr_2\left[O_{\text{B}}\right]-\frac{d_2d_3\Tr_1\left[O_{\text{A}}\right]\Tr_2\left[O_{\text{B}}\right]}{d_2}\\
&=0\,.
\end{align}

\end{proof}

\section{Proof of Proposition~1}\label{secSM:Prop1}

Here we provide a proof for Proposition~1, which we recall for convenience:

\begin{proposition}\label{prop1SM}
Let $\theta^{j}$ be uniformly distributed on $[-\pi,\pi]$ $\forall j$. For any $\delta \in (0,1)$, the probability that $C_{\text{G}}\le \delta$ satisfies
\begin{equation}
\text{\normalfont Pr}\lbrace C_{\text{G}}\le \delta \rbrace \le (1-\delta)^{-1}\left({1\over 2} \right)^{n}.
\end{equation}
For any $\delta \in [{1\over 2},1]$, the probability that $C_{\text{L}}\le \delta$ satisfies
\begin{equation}
\text{\normalfont Pr}\lbrace C_{\text{L}}\le \delta \rbrace \ge { (2\delta -1)^{2}\over {1\over 2n} + (2\delta -1)^{2} } \underset{n\to\infty}{\longrightarrow}1\,.
\end{equation}
\label{prop:oneSM}
\end{proposition}

Proposition~1 formalizes the narrow gorge phenomenon shown in Fig.~2 of the main text for the warm-up example. Specifically, it bounds the volume of parameter space that is trainable in the sense that the cost deviates from its maximum value of one. As a consequence, one finds that the probability that the global cost function is less than 1 vanishes exponentially as $n\rightarrow \infty$, whereas for the local cost function there are neighbohoods of the minimum that have constant probability as $n\rightarrow \infty$. In the following proof, it is helpful to define the trainable region for the global and local cost as $A_{\delta}^{G}:=\lbrace \theta :C_G(\theta)\leq \delta \rbrace$ and $A_{\delta}^{L}:=\lbrace \theta :C_{\text{L}}(\theta)\leq \delta \rbrace$.

\begin{proof}

Let us first consider the global cost function
\begin{equation}
C_G = \Tr[O_G V(\thv)\dya{\vec{0}}V(\thv)\ad]\,, \label{eq:Cost-TDSM} 
\end{equation}
with  $O_G=\id - \dya{\vec{0}}$, and with $V(\vec{\theta})=\bigotimes_{j=1}^{n}e^{-i \theta^j \sigma_{x}^{(j)} /2} $. The probability of the region $A_{\delta}^{G}$ can be bounded from above by defining the random variable $X=\prod_{j=1}^{n}\cos^{2}{\theta_{j}\over 2}=1-C_{\text{G}}$ and noting that $(1-\delta)P\lbrace X\ge 1-\delta\rbrace \le E(X)$. Then
\begin{align}
\Pr (A_{\delta}^{G}) &= P\lbrace \prod_{j=1}^{n}\cos^{2}{\theta^{j}\over 2} \ge 1-\delta \rbrace  \nonumber \\
&\le (1-\delta)^{-1}E\left( \prod_{j=1}^{n}\cos^{2}{\theta^{j}\over 2} \right) \nonumber \\
&= (1-\delta)^{-1} \left({1\over 2} \right)^{n}
\end{align}
Therefore, for any $\delta \in (0,1)$, the probability that $C_{\text{G}}\le \delta$ goes to zero exponentially with $n$.

Now consider the local cost function $C_{\text{L}}$ given by
\begin{align}
C_{\text{L}}=\Tr\left[O_{\text{L}}V(\thv)\dya{\vec{0}} V(\thv)\ad\right],
\label{eq:localcostfunctionsm}\quad \text{with}\quad O_{\text{L}}=\id - \frac{1}{n}\sum_{j=1}^n \dya{0}_j\otimes\id_{\overline{j}}\,,
\end{align}
and where $\id_{\overline{j}}$ is the identity on all qubits except qubit $j$. To show the dependence of the result on the range of local cost function values, we write a parametrized local cost function $C_{\text{L}}(\bm{\theta};\lambda):= 1-{1\over \lambda n}\sum_{j=1}^{n}\cos^{2}{\theta_{j}\over 2}$ and define $A_{\delta}^{L}:=    \lbrace \bm{\theta}: C_{\text{L}}(\bm{\theta};\lambda)\leq \delta\rbrace$. $C_{\text{L}}$ in the main text is obtained for $\lambda = 1$. The parameter $\lambda$ is introduced to show that the range of validity of the inequality is dependent on the cost function. For $\delta$ in the interval $\delta\in (1-(2\lambda)^{-1},1]$,
\begin{align}
\Pr\left(A_{\delta}^{L}\right) = \Pr\left(\left\lbrace \bm{\theta}:{1\over \lambda n}\sum_{j=1}^{n}\cos^{2}{\theta^j\over 2} \ge 1-\delta \right\rbrace \right) 
\geq {(1-2\lambda(1-\delta))^{2}\over {1\over 2n}+(1-2\lambda(1-\delta))^{2}}
\label{eqn:pz}
\end{align}
which tends to 1 as $n\rightarrow \infty$. The inequality follows from the Paley-Zygmund inequality \cite{paley_zygmund_1932} in the form
\begin{equation}
P(X\ge r E(X))\geq {(1-r)^{2}E(X)^{2}\over \text{Var}X + (1-r)^{2}E(X)^{2}}
\end{equation}
which holds for random variables $X\ge 0$ and scalar $r$ with $0\leq r \leq 1$. In particular, (\ref{eqn:pz}) follows by taking $X={1\over n\lambda} \sum_{j=1}^{n}\cos^{2}{\theta^{j}\over 2}$ and $r=2\lambda(1-\delta)$, so that $r$ varies from $1$ to $0$ as $\delta$ varies from $1-(2\lambda)^{-1}$ to $1$.
\end{proof}

\section{Proof of Proposition 2}
\label{secSM:Prop2}

Let us first recall that in the main text we analyze cost functions $C$ which can be expressed as the expectation value of a given observable $O$ as 
\begin{align}\label{eq:cost_function2}
    C=\Tr\left[ O V(\vec{\theta}) \rho V\ad(\vec{\theta}) \right]\,,
\end{align}
where  $V(\vec{\theta})$ is a parametrized quantum gate sequence and $\rho$ is a general input mixed quantum state on $n$ qubits.
Here we provide a proof for Proposition~2, which we recall for convenience:
\begin{proposition}
\label{prop2sm}
The average of the partial derivative of any cost function of the form~\eqref{eq:cost_function2} with respect to a parameter $\theta^\nu$ in a block $W$ of the ansatz $V(\thv)$ is
\begin{equation}
    \langle\partial_\nu C\rangle_V=0\,,
\end{equation}
provided that either $W_{\text{A}}$ or $W_{\text{B}}$ form a $1$-design.
\end{proposition}

 As discussed in the main text, we employ an Alternating Layered Ansatz, where each layer is composed of $m$-qubits gates or ``blocks''. In particular, each block $W_{kl}(\vec{\theta}_{kl})$ in $V(\vec{\theta})$ can be written as a product of $\zeta_{kl}$ independent gates from a gate alphabet $\mathcal{A}=\{G_{\mu}(\theta)\}$ as 
\begin{equation}
W_{kl}(\vec{\theta}_{kl})=G_{\zeta_{kl}}(\theta_{kl}^{\zeta_{kl}})\ldots G_{\nu}(\theta^\nu_{kl})\ldots G_{1}(\theta^{1}_{kl})\,,     \label{eq;W-gatessm}
\end{equation}
where $\theta_{kl}^\nu$ are continuous parameter, and where $G_\nu(\theta^\nu_{kl})= R_{\nu}(\theta^\nu_{kl}) Q_{\nu}$ with $Q_{\nu}$ an unparametrized gate, and $R_{\nu}(\theta^\nu_{kl}) = e^{-i\theta^\nu_{kl} \sigma_\nu /2 }$ such that $\sigma_\nu$ is a Pauli operator.

Consider now a block $W_{kl}(\vec{\theta}_{kl})$ in the $l$-th layer of the ansatz. For the rest of this Supplementary Information we simply use the notation $W$ when referring to this particular block. First, let $S_w$  denote the $m$-qubit subsystem that contains the qubits $W$ acts on, and let $S_{\overline{w}}$ be the $(n-m)$ subsystem on which $W$ acts trivially, with $\HC_w$, and $\HC_{\overline{w}}$ their respective associated Hilbert spaces. Then, consider a given trainable parameter $\theta^\nu$ in $W$, such that we can express $W=W_{\text{B}} W_{\text{A}}$, with
\begin{equation}\label{eq:WA-WB}
W_{\text{B}}=\prod_{\mu =1}^{\nu-1} G_{\mu}(\theta^{\mu})\,,\quad \text{and}\quad  W_{\text{A}}=\prod_{\mu =\nu}^{\zeta} G_{\mu}(\theta^{\mu})\,.
\end{equation}

Then, we recall that we have defined the forward light-cone $\LC$ of $W$ as all gates with at least one input qubit causally connected to the output qubits of $W$. We can then define $S_\LC$ as the subsystem of all qubits in $\LC$, and $\HC_\LC$ as its associated Hilbert space. Without loss of generality, the trainable gate sequence can be expressed as
\begin{equation}
    V(\vec{\theta})=V_{\text{R}} (\id_{\overline{w}}\otimes W)V_{\text{L}}\,, \label{eq:Ansatz-WSM}
\end{equation}
where  $\id_{\overline{w}}$ indicates the identity in $\HC_{\overline{w}}$, and where we assume without loss of generality that  $V_{\text{R}}$ contains the gates in $\LC$ and all the blocks $W_{kL}$ in the last layer of $V(\thv)$.

\begin{proof}
 The partial derivative of $W$ with respect to the angle $\theta_\nu$ is given by
\begin{align}
    \partial_\nu W = W_{\text{A}}\left(-\frac{i}{2}\sigma_\nu\right)W_{\text{B}}\,, \label{eq:parderW}
\end{align}
where here $\sigma_\nu$ is an operator  $\sigma_\nu:\HC_w\rightarrow\HC_w$, which acts non-trivially on a qubit given qubit $j$ in $\HC_w$, i.e., $\sigma_\nu := (\sigma_\nu)_j\otimes\id_{\overline{j}}$. Hence, by means of Eqs.~\eqref{eq:parderW} and~\eqref{eq:Ansatz-WSM} we have 
\begin{align*}
    \partial_\nu C =&\Tr\left[O\left(\partial_\nu V(\vec{\theta})\right)\rho V\ad(\vec{\theta})+V(\vec{\theta})\rho\left(\partial_\nu V\ad(\vec{\theta})\right)\right]\\
    =&\Tr\left[OV_{\text{R}}\left(\id_{\overline{w}}\otimes W_{\text{A}}\right)\left(\id_{\overline{w}}\otimes\left(-\frac{i}{2}\sigma_\nu\right)\right)\left(\id_{\overline{w}}\otimes W_{\text{B}}\right)V_{\text{L}}\rho V_{\text{L}}\ad\left(\id_{\overline{w}}\otimes W_{\text{B}}\ad W_{\text{A}}\ad\right)V_{\text{R}}\ad\right]\\
    &+\Tr\left[OV_{\text{R}}\left(\id_{\overline{w}}\otimes W_{\text{A}}W_{\text{B}}\right)V_{\text{L}}\rho V_{\text{L}}\ad\left(\id_{\overline{w}}\otimes W_{\text{B}}\ad\right)\left(\id_{\overline{w}}\otimes\left(+\frac{i}{2}\sigma_\nu\right)\right)\left(\id_{\overline{w}}\otimes W_{\text{A}}\ad\right)V_{\text{R}}\ad\right]\,.
\end{align*}
Which can be simplified as
\begin{equation}
    \partial_\nu C=\frac{i}{2} \Tr\left[\left(\id_{\overline{w}}\otimes W_{\text{B}}\right)V_{\text{L}} \rho V_{\text{L}}\ad\left(\id_{\overline{w}}\otimes W_{\text{B}}\ad\right)
    \left[\id_{\overline{w}}\otimes \sigma_\nu,\left(\id_{\overline{w}}\otimes W_{\text{A}}\ad\right)V_{\text{R}}\ad O V_{\text{R}} \left(\id_{\overline{w}}\otimes W_{\text{A}}\right)\right]\right]\,,
    \label{eq:gradientWA}
\end{equation}
or equivalently, as
\begin{equation}
    \partial_\nu C=-\frac{i}{2} \Tr\left[\left(\id_{\overline{w}}\otimes W_{\text{A}}\ad\right)V_{\text{R}}\ad O V_{\text{R}} \left(\id_{\overline{w}}\otimes W_{\text{A}}\right)\left[\id_{\overline{w}}\otimes \sigma_\nu,\left(\id_{\overline{w}}\otimes W_{\text{B}}\right)V_{\text{L}} \rho V_{\text{L}}\ad\left(\id_{\overline{w}}\otimes W_{\text{B}}\ad\right)\right]\right]\,.
    \label{eq:gradientWB}
\end{equation}

In order to compute the expectation value $\langle \partial_\nu C\rangle_V$ we need to consider three different scenarios: (1) when only $W_{\text{A}}$ is a $1$-design; (2) when only $W_{\text{B}}$ is a $1$-design; and (3) when both $W_{\text{A}}$ and $W_{\text{B}}$ form $1$-designs. 
We first consider the case when  $W_{\text{A}}$ is a $1$-design. Since $W_{\text{A}}$, $W_{\text{B}}$, $V_{\text{R}}$ and $V_{\text{L}}$ are independent, we can compute the expectation value over the ansatz as $\langle \partial_\nu C\rangle_V=\langle \langle\partial_\nu C\rangle_{W_{\text{A}}}\rangle_{V_{\text{L}},W_{\text{B}},V_{\text{R}}}$.  From  \eqref{eq:gradientWA} and the definition of a $1$-design in~\eqref{eq:SMt-design}, we can compute
\begin{align}
   \langle \partial_\nu C\rangle_{W_{\text{A}}}=& 
   -\frac{i}{2}   \Tr\left[\left(\id_{\overline{w}}\otimes W_{\text{B}}\right)V_{\text{L}}\rho V_{\text{L}}\ad\left(\id_{\overline{w}}\otimes W_{\text{B}}\ad\right)\left[\id_{\overline{w}}\otimes \sigma_\nu,\int d\mu(W_{\text{A}})\left(\id_{\overline{w}}\otimes W_{\text{A}}\right)V_{\text{R}}\ad O V_{\text{R}} \left(\id_{\overline{w}}\otimes W_{\text{A}}\ad\right)\right]\right]\nonumber\\
    =&-\frac{i}{2}   \Tr\left[\left(\id_{\overline{w}}\otimes W_{\text{B}}\right)V_{\text{L}}\rho V_{\text{L}}\ad\left(\id_{\overline{w}}\otimes W_{\text{B}}\ad\right)
    \left[\id_{\overline{w}}\otimes \sigma_\nu,\frac{1}{2^m}\Tr_w[V_{\text{R}}\ad O V_{\text{R}}]\otimes\id_w\right]\right]\nonumber\\
    =&0\,, 
    \label{eq:SM_gradWAzero}
\end{align}
where in the second equality we used Lemma~\ref{lemma5}. 

On the other hand, if  $W_{\text{B}}$ is a $1$-design, we can now employ  \eqref{eq:gradientWB} to get
\begin{align}
   \langle \partial_\nu C\rangle_{W_{\text{B}}}=& 
   -\frac{i}{2}  \Tr\left[\left(\id_{\overline{w}}\otimes W_{\text{A}}\ad\right) V_{\text{R}}\ad O V_{\text{R}}\left(\id_{\overline{w}}\otimes W_{\text{A}}\right)\left[\id_{\overline{w}}\otimes \sigma_\nu,\int d\mu(W_{\text{B}})\left(\id_{\overline{w}}\otimes W_{\text{B}}\right)V_{\text{L}}\rho V_{\text{L}}\ad\left(\id_{\overline{w}}\otimes W_{\text{B}}\ad\right)\right]\right]\nonumber\\
   =&-\frac{i}{2}  \Tr\left[\left(\id_{\overline{w}}\otimes W_{\text{A}}\ad\right) V_{\text{R}}\ad O V_{\text{R}}\left(\id_{\overline{w}}\otimes W_{\text{A}}\right)\left[\id_{\overline{w}}\otimes \sigma_\nu, \frac{1}{2^m}\Tr_w[V_{\text{L}}\rho V_{\text{L}}\ad]\otimes\id_w\right]\right]\nonumber\\
   =&0\,,
   \label{eq:SM_gradWBzero}
  \end{align}
which follows from the same argument used to derive \eqref{eq:SM_gradWAzero}. Finally, from Eqs.~\eqref{eq:SM_gradWAzero} and \eqref{eq:SM_gradWBzero}, we have $\langle \partial_\nu C\rangle_{V}=0$~~($\forall V_{\text{L}}, V_{\text{R}}$) if $W_{\text{A}}$ and $W_{\text{B}}$ are both $1$-designs.

\end{proof}

\section{Variance of the cost function partial derivative}
\label{secSM:Vari}

In this section, we derive the formula for the variance of the cost function gradient. Namely,  
\begin{align}
 \langle(\partial_\nu C)^2\rangle_{V}=\frac{2^{m-1}\Tr[\sigma_\nu^2]}{(2^{2m}-1)^2}\sum_{\substack{\vec{p}\vec{q}\\\vec{p}'\vec{q}'}}\left\langle\Delta\Omega_{\vec{p}\vec{q}}^{\vec{p}'\vec{q}'}\right\rangle_{V_{\text{R}}}\left\langle\Delta\Psi_{\vec{p}\vec{q}}^{\vec{p}'\vec{q}'}\right\rangle_{V_{\text{L}}}\,,\label{eq:Var-Csm}
\end{align}
with 
\begin{align}
    \Delta\Omega_{\vec{p}\vec{q}}^{\vec{p}'\vec{q}'} &= \Tr[ \Omega_{\vec{q}\vec{p}}\Omega_{\vec{q}'\vec{p}'}] -\frac{\Tr[ \Omega_{\vec{q}\vec{p}}]\Tr[\Omega_{\vec{q}'\vec{p}'}]}{2^{m}}\,, \label{eq:DeltaOmega} \\
    \Delta\Psi_{\vec{p}\vec{q}}^{\vec{p}'\vec{q}'} &=  \Tr[\Psi_{\vec{p}\vec{q}}\Psi_{\vec{p}'\vec{q}'}]-\frac{\Tr[\Psi_{\vec{p}\vec{q}}]\Tr[\Psi_{\vec{p}'\vec{q}'}]}{2^{m}}\,,  \label{eq:DeltaPsi}
\end{align}
and where
\begin{align}
    \Omega_{\vec{q}\vec{p}}&=\Tr_{\overline{w}}\left[(\ketbra{\vec{p}}{\vec{q}}\otimes\id_w)V_{\text{R}}\ad OV_{\text{R}}\right]\,,\label{eq:eq-Omegapq}\\ 
    \Psi_{\vec{p}\vec{q}}&=\Tr_{\overline{w}}\left[(\ketbra{\vec{q}}{\vec{p}}\otimes\id_w)V_{\text{L}} \rho V_{\text{L}}\ad\right]\,\label{eq:eq-Psipq}.
\end{align}

\begin{proof}

As shown in the previous section, $\langle\partial_\nu C\rangle_V=0$ when either $W_{\text{B}}$ or  $W_{\text{A}}$ of~\eqref{eq:WA-WB} are $1$-designs. Hence, we can compute the variance of $\partial_\nu C$ as $\Var[\partial_\nu C]=\langle(\partial_\nu C)^2\rangle_V-\langle\partial_\nu C\rangle_V^2=\langle(\partial_\nu C)^2\rangle_V$. From~\eqref{eq:gradientWB} and Lemma~\ref{lemma4} we have
\begin{align}
    (\partial_\nu C)^2=-\frac{1}{4}\sum_{\substack{\vec{p},\vec{q}\\\vec{p}',\vec{q}'}}\Tr\left[W_{\text{A}} \Omega_{\vec{q}\vec{p}}W_{\text{A}}\ad \Gamma_{\vec{p}\vec{q}}\right]\Tr\left[W_{\text{A}} \Omega_{\vec{q}'\vec{p}'}W_{\text{A}}\ad \Gamma_{\vec{p}'\vec{q}'}\right]\,,\label{eq:SM-var1}
\end{align}
with 
\begin{align}
    \Gamma_{\vec{p}\vec{q}}&=\Tr_{\overline{w}}\left[(\ketbra{\vec{q}}{\vec{p}}\otimes\id_w)\left[\id_{\overline{w}}\otimes\sigma_\nu,(\id_{\overline{w}}\otimes   W_{\text{B}})V_{\text{L}}\rho V_{\text{L}}\ad(\id_{\overline{w}}\otimes  W_{\text{B}}\ad)\right]\right]\nonumber\\
    &=\Tr_{\overline{w}}\left[[\id_{\overline{w}}\otimes\sigma_\nu,(\id_{\overline{w}}\otimes   W_{\text{B}})(\ketbra{\vec{q}}{\vec{p}}\otimes\id_w)V_{\text{L}}\rho V_{\text{L}}\ad(\id_{\overline{w}}\otimes  W_{\text{B}}\ad)]\right]\nonumber\\
    &=\left[\sigma_\nu, W_{\text{B}}\Tr_{\overline{w}}[(\ketbra{\vec{q}}{\vec{p}}\otimes\id_w)V_{\text{L}}\rho V_{\text{L}}\ad]  W_{\text{B}}\ad\right]\nonumber\\
    &=\left[\sigma_\nu, W_{\text{B}}\Psi_{\vec{p}\vec{q}}  W_{\text{B}}\ad\right]\,.
\end{align}

As previously mentioned, if  $W_{\text{A}}$, $W_{\text{B}}$, $V_{\text{R}}$ and $V_{\text{L}}$ are independent, the expectation value of  \eqref{eq:SM-var1} can be computed as $\langle(\partial_\nu C)^2\rangle_V=\langle(\partial_\nu C)^2\rangle_{V_{\text{L}},W_{\text{B}},W_{\text{A}},V_{\text{R}}}$. In addition, if $W_{\text{A}}$ is a $2$-design, we get from Lemma~\ref{lemma3} 
\begin{align}
\langle(\partial_\nu C)^2\rangle_{W_{\text{A}}}&=-\frac{1}{4}\sum_{\substack{\vec{p},\vec{q}\\\vec{p}',\vec{q}'}}\int d\mu(W_{\text{A}})\Tr\left[W_{\text{A}} \Omega_{\vec{q}\vec{p}}W_{\text{A}}\ad \Gamma_{\vec{p}\vec{q}}\right]\Tr\left[W_{\text{A}} \Omega_{\vec{q}'\vec{p}'}W_{\text{A}}\ad \Gamma_{\vec{p}'\vec{q}'}\right]\nonumber\\
    &=-\frac{1}{4}\sum_{\substack{\vec{p},\vec{q}\\\vec{p}',\vec{q}'}}\bigg(\frac{1}{2^{2m}-1}\left(\Tr[ \Omega_{\vec{q}\vec{p}}]\Tr[\Gamma_{\vec{p}\vec{q}}]\Tr[\Omega_{\vec{q}'\vec{p}'}]\Tr[\Gamma_{\vec{p}'\vec{q}'}]+\Tr[ \Omega_{\vec{q}\vec{p}}\Omega_{\vec{q}'\vec{p}'}]\Tr[\Gamma_{\vec{p}\vec{q}}\Gamma_{\vec{p}'\vec{q}'}]\right) \nonumber\\   &\quad-\frac{1}{2^{m}(2^{2m}-1)}\left(\Tr[ \Omega_{\vec{q}\vec{p}}\Omega_{\vec{q}'\vec{p}'}]\Tr[\Gamma_{\vec{p}\vec{q}}]\Tr[\Gamma_{\vec{p}'\vec{q}'}]+\Tr[ \Omega_{\vec{q}\vec{p}}]\Tr[\Omega_{\vec{q}'\vec{p}'}]\Tr[\Gamma_{\vec{p}\vec{q}}\Gamma_{\vec{p}'\vec{q}'}]\right)\bigg)\nonumber\\
    &=-\frac{1}{4(2^{2m}-1)}\sum_{\substack{\vec{p},\vec{q}\\\vec{p}',\vec{q}'}}\bigg(\Tr[ \Omega_{\vec{q}\vec{p}}\Omega_{\vec{q}'\vec{p}'}] -\frac{1}{2^{m}}\Tr[ \Omega_{\vec{q}\vec{p}}]\Tr[\Omega_{\vec{q}'\vec{p}'}]\bigg)\Tr[\Gamma_{\vec{p}\vec{q}}\Gamma_{\vec{p}'\vec{q}'}]\,,\label{eq:SM-varWA}
\end{align}
where in the third equality we used the fact that the trace of a commutator is zero: $\Tr[\Gamma_{\vec{p}\vec{q}}] =0$. 

If $W_{\text{B}}$ is also a $2$-design, then from~\eqref{eq:SM-varWA} we need to compute the following expectation value 
\begin{align}
\langle(\partial_\nu C)^2\rangle_{W_{\text{B}},W_{\text{A}}}&=-\frac{1}{4(2^{2m}-1)}\sum_{\substack{\vec{p},\vec{q}\\\vec{p}',\vec{q}'}}\bigg(\Tr[ \Omega_{\vec{q}\vec{p}}\Omega_{\vec{q}'\vec{p}'}] -\frac{1}{2^{m}}\Tr[ \Omega_{\vec{q}\vec{p}}]\Tr[\Omega_{\vec{q}'\vec{p}'}]\bigg)\int d\mu(W_{\text{B}})\Tr[\Gamma_{\vec{p}\vec{q}}\Gamma_{\vec{p}'\vec{q}'}]\,.\label{eq:labelrev}
\end{align}
Let us first note that 
\begin{align}
    \Gamma_{\vec{p}\vec{q}}\Gamma_{\vec{p}'\vec{q}'}&=\left[\sigma_\nu, W_{\text{B}}\Psi_{\vec{p}\vec{q}}  W_{\text{B}}\ad\right]\left[\sigma_\nu, W_{\text{B}}\Psi_{\vec{p}'\vec{q}'}  W_{\text{B}}\ad\right]\nonumber\\
    &=2\left(\sigma_\nu W_{\text{B}}\Psi_{\vec{p}\vec{q}}  W_{\text{B}}\ad\sigma_\nu W_{\text{B}}\Psi_{\vec{p}'\vec{q}'} W_{\text{B}}\ad\right)-2\left(W_{\text{B}}\sigma_\nu^2W_{\text{B}}\ad\Psi_{\vec{p}\vec{q}}\Psi_{\vec{p}'\vec{q}'}  \right)\,.
    \label{eq:SMVarWB}
\end{align}
This result can be used along with Lemmas~\ref{lemma1} and~\ref{lemma2} to compute the integral in  \eqref{eq:SMVarWB} as
\begin{align}
    \int d\mu(W_{\text{B}})\Tr[\Gamma_{\vec{p}\vec{q}}\Gamma_{\vec{p}'\vec{q}'}]=&2\int d\mu(W_{\text{B}})\Tr\left[\sigma_\nu W_{\text{B}}\Psi_{\vec{p}\vec{q}}  W_{\text{B}}\ad\sigma_\nu W_{\text{B}}\Psi_{\vec{p}'\vec{q}'} W_{\text{B}}\ad\right]\nonumber\\\
    &-2\int d\mu(W_{\text{B}})\Tr\left[W_{\text{B}}\sigma_\nu^2W_{\text{B}}\ad\Psi_{\vec{p}\vec{q}}\Psi_{\vec{p}'\vec{q}'}  \right]\nonumber\\
    =&-\frac{2^{m+1}}{2^{2m}-1}\Tr[\sigma_\nu^2]\left(\Tr[\Psi_{\vec{p}\vec{q}}\Psi_{\vec{p}'\vec{q}'}]-\frac{1}{2^{m}}\Tr[\Psi_{\vec{p}\vec{q}}]\Tr[\Psi_{\vec{p}'\vec{q}'}]\right)\,,
    \label{eq:SMVarInt}
\end{align}
where we used the fact that $\sigma_\nu$ is a Pauli operator, and hence its trace is equal to zero.

Then, combining Eqs.~\eqref{eq:labelrev} and~\eqref{eq:SMVarInt}, we  obtain
\begin{align}
 \langle(\partial_\nu C)^2\rangle_{V}=\frac{2^{m-1}\Tr[\sigma_\nu^2]}{(2^{2m}-1)^2}\sum_{\substack{\vec{p}\vec{q}\\\vec{p}'\vec{q}'}}\left\langle\Delta\Omega_{\vec{p}\vec{q}}^{\vec{p}'\vec{q}'}\right\rangle_{V_{\text{R}}}\left\langle\Delta\Psi_{\vec{p}\vec{q}}^{\vec{p}'\vec{q}'}\right\rangle_{V_{\text{L}}}\,.
\label{eq:SMVarWB2}
\end{align}

\end{proof}

\section{Variance of the cost function partial derivative for a  single layer of  the Alternating Layered Ansatz}
\label{secSM:tensorproduct}

In this section we explicitly evaluate Eqs.~\eqref{eq:Var-Csm}--\eqref{eq:DeltaPsi} for the special case when $V(\thv)$ is composed of a single layer  of  the Alternating Layered Ansatz.  This case is a generalization of the warm-up example of the main text, and constitutes the first step towards our main theorems. In particular, we remark that the tools employed here are the same as the ones used to derive our main result.

\subsection{Variance of global cost function partial derivative}
Let us first recall that the global cost function is
$C_G=1-\Tr\left[OV\rho V\ad\right]$, where $O=\bigotimes_{k=1}^{\xi}\widehat{O}_k$, and where $V(\thv)$ is given by a single layer of the Alternating Layered Ansatz, i.e., $V(\thv)=\bigotimes_{k=1}^{\xi}W_{k1}\left(\vec{\theta}_k\right)$. Moreover, we recall that  we assume without loss of generality that $V_{\text{R}}$ contains the gates in $\LC$ and all the blocks $W_{k1}$ in the last (and in this case only) layer of $V(\thv)$. The latter means that here $V_{\text{L}}=\id$, and $V_{\text{R}}=\id_h\otimes\left(\bigotimes_{k\neq h}W_{k1}\left(\vec{\theta}_k\right)\right)$. In addition, for simplicity, we have defined $W_{k}:= W_{k1}$. Here, $\xi$ is the total number of blocks so that $n=\xi m$, and we assume that the angle $\theta_\nu$ we want to train is in the $h$-th block $W_{h}$. 

\subsubsection{Expectation value over $V_{\text{R}}$}

First, let us compute $\Omega_{\vec{qp}}$. From \eqref{eq:eq-Omegapq}, we obtain
\begin{align}
    \Omega_{\vec{qp}}= \widehat{O}_{h} \prod_{k\neq h}^{\xi}\Tr\left[W_k\ad \widehat{O}_k W_k \ketbra{\vec{p}_k}{\vec{q}_k}\right]\,.
\end{align}
Replacing this result in Eq.~\eqref{eq:DeltaOmega} and employing  Lemma~\ref{lemma3} $(\xi-1)$-times results in
\begin{align*}
    \left\langle\Delta\Omega_{\vec{p}\vec{q}}^{\vec{p}'\vec{q}'}\right\rangle_{V_{\text{R}}}
    =&\frac{1}{(2^{2m}-1)^{\xi-1}}\left(\Tr\left[\widehat{O}_h^2\right]-\frac{1}{2^m}\Tr\left[\widehat{O}_h\right]^2\right)\\
    &\times \prod_{k\neq h}^{\xi} \left(\delta_{(\vec{p},\vec{q})_{S_k}}\delta_{(\vec{p}' ,\vec{q}')_{S_k}}\left(\Tr\left[\widehat{O}_k\right]^2-\frac{1}{2^m}\Tr\left[\widehat{O}_k^2\right]\right)
    +\delta_{(\vec{p'}, \vec{q})_{S_k}}\delta_{(\vec{p}, \vec{q'})_{S_k}}\left(\Tr\left[\widehat{O}_k^2\right]-\frac{1}{2^m}\Tr\left[\widehat{O}_k\right]^2\right)\right)\nonumber\,.
\end{align*}
Now, let us consider $\widehat{O}_k$ $(k=1,2,\cdots, \xi)$ to be rank-$1$ projector, i.e. $\Tr\left[\widehat{O}_k\right]=\Tr\left[\widehat{O}_k^2\right]=\rank\left[\widehat{O}_k\right]=1$. 
Therefore, we can obtain 
\begin{align}\label{eqSM:deltaOmegafinal}
   \left\langle\Delta\Omega_{\vec{p}\vec{q}}^{\vec{p}'\vec{q}'}\right\rangle_{V_{\text{R}}} &= \frac{1}{(2^{2m}-1)^{\xi-1}}\left(1-\frac{1}{2^m}\right)^{\xi} \prod_{k\neq h}^{\xi}\left(\delta_{(\vec{p},\vec{q})_{S_k}}\delta_{(\vec{p}' ,\vec{q}')_{S_k}}
    +\delta_{(\vec{p'}, \vec{q})_{S_k}}\delta_{(\vec{p}, \vec{q'})_{S_k}}\right)\,.
\end{align}

\subsubsection{Expectation value over $V_{\text{L}}$}

Next, let us consider $\Psi_{\vec{pq}}$ in \eqref{eq:eq-Psipq}. Here, we can set $V_{\text{L}}=\id$, which leads to  $\Psi_{\vec{pq}}=\Tr_{\overline{h}}\left[\left(\ketbra{\vec{p}}{\vec{q}}\otimes\id_{h}\right)\rho\right]$.  
Note that any quantum state $\rho$ can be always written as 
\begin{align}\label{eqSM:rho}
  \rho = \sum_\lambda p_\lambda \dya{\psi_\lambda}\,, \quad \text{with}\quad   \ket{\psi_\lambda} = \sum_{\vec{\alpha}^{\lambda}}c_{\vec{\alpha}^{\lambda}}\ket{\vec{\alpha}_{1}^{\lambda}}\otimes\cdots\otimes\ket{\vec{\alpha}_{h}^{\lambda}}\otimes\cdots \otimes \ket{\vec{\alpha}_{\xi}^{\lambda}}\,,
\end{align}
where $\vec{\alpha}:=\vec{\alpha_1}\cdots\vec{\alpha_\xi}$, and where $\vec{\alpha_i}$ are bitstrings of length $m$. Hence we find
\begin{align}\label{eqsm:psi2}
    \Psi_{\vec{pq}} = \sum_{\lambda} p_\lambda  \sum_{\vec{\alpha}^{\lambda},\vec{\alpha'}^{\lambda}}c_{\vec{\alpha}^{\lambda}}c^*_{\vec{\alpha'}^{\lambda}}\left(\prod_{k\neq h}^{\xi}\delta_{(\vec{q},\vec{\alpha}^{\lambda})_{k}}\delta_{(\vec{p},\vec{\alpha'}^{\lambda})_{S_k}}\right)\ketbra{\vec{\alpha}_{h}^{\lambda}}{\vec{\alpha'}_h^{\lambda}}\,.
\end{align}
Then,  since $V_{\text{L}}=\id$, we have $ \left\langle\Delta\Psi_{\vec{p}\vec{q}}^{\vec{p}'\vec{q}'}\right\rangle_{V_{\text{L}}}=\Delta\Psi_{\vec{p}\vec{q}}^{\vec{p}'\vec{q}'}$, and we can use~\eqref{eqsm:psi2} to  get 
\begin{equation}
\label{eqSM:psiFinal}
\begin{split}
  \Delta\Psi_{\vec{p}\vec{q}}^{\vec{p}'\vec{q}'} =& \sum_{\lambda,\lambda'}p_\lambda p_{\lambda'}\!\!\!\!\!\sum_{\substack{\vec{\alpha}^{\lambda},\vec{\alpha'}^{\lambda}\\\vec{\beta}^{\lambda'},\vec{\beta'}^{\lambda'}}}\!\!\!c_{\vec{\alpha}^{\lambda}}c^*_{\vec{\alpha'}^{\lambda}}c_{\vec{\beta}^{\lambda'}}c^*_{\vec{\beta'}^{\lambda'}}\left(\prod_{k\neq h}^{\xi}\delta_{(\vec{q},\vec{\alpha}^{\lambda})_{S_k}}\delta_{(\vec{p},\vec{\alpha'}^{\lambda})_{S_k}}\delta_{(\vec{q'},\vec{\beta}^{\lambda'})_{S_k}}\delta_{(\vec{p'},\vec{\beta'}^{\lambda'})_{S_k}}\right)\\
  &\times\left(\delta_{(\vec{\alpha'}^{\lambda},\vec{\beta}^{\lambda'})_{S_h}}\delta_{(\vec{\alpha}^{\lambda},\vec{\beta'}^{\lambda'})_{S_h}}-\frac{1}{2^m}\delta_{(\vec{\alpha}^{\lambda},\vec{\alpha'}^{\lambda})_{S_h}}\delta_{(\vec{\beta}^{\lambda'},\vec{\beta'}^{\lambda'})_{S_h}}\right)\,.
  \end{split}
\end{equation}
Finally, from~\eqref{eqSM:deltaOmegafinal}, \eqref{eqSM:psiFinal}, and the fact that $\Tr\left[\sigma_\nu^2\right]=2^m$,  we obtain {\small
\begin{align*}
    \Var\left[\partial_\nu C_G\right] =&\frac{2^{2m-1}}{(2^{2m}-1)^2(2^{2m}-1)^{\xi-1}}\left(1-\frac{1}{2^m}\right)^{\xi}
    \sum_{\lambda,\lambda'}p_\lambda p_{\lambda'} \sum_{\substack{\vec{\alpha}^{\lambda},\vec{\alpha'}^{\lambda}\\\vec{\beta}^{\lambda'},\vec{\beta'}^{\lambda'}}}c_{\vec{\alpha}^{\lambda}}c^*_{\vec{\alpha'}^{\lambda}}c_{\vec{\beta}^{\lambda'}}c^*_{\vec{\beta'}^{\lambda'}}\\
    &\times\left(\delta_{(\vec{\alpha'}^{\lambda},\vec{\beta}^{\lambda'})_{S_h}}\delta_{(\vec{\alpha}^{\lambda},\vec{\beta'}^{\lambda'})_{S_h}}-\frac{1}{2^m}\delta_{(\vec{\alpha}^{\lambda},\vec{\alpha'}^{\lambda})_{S_h}}\delta_{(\vec{\beta}^{\lambda'},\vec{\beta'}^{\lambda'})_{S_h}}\right)\prod_{k\neq h} \left(\delta_{(\vec{\alpha}^{\lambda},\vec{\alpha'}^{\lambda})_{S_k}}\delta_{(\vec{\beta}^{\lambda'},\vec{\beta'}^{\lambda'})_{S_k}}+\delta_{(\vec{\alpha}^{\lambda},\vec{\beta'}^{\lambda'})_{S_k}}\delta_{(\vec{\alpha'}^{\lambda},\vec{\beta}^{\lambda'})_{S_k}}\right)\,,
\end{align*}}
where we used the fact that 
{\footnotesize
\begin{align*}
\sum_{\substack{\vec{p}\vec{q}\\\vec{p}'\vec{q}'}}&\left(\prod_{k\neq h}^{\xi}\left(\delta_{(\vec{p}, \vec{q})_{S_k}}\delta_{(\vec{p}', \vec{q}')_{S_k}}
    +\delta_{(\vec{p}',\vec{q})_{S_k}}\delta_{(\vec{p}, \vec{q'})_{S_k}}\right)\delta_{(\vec{q},\vec{\alpha}^{\lambda})_{S_k}}\delta_{(\vec{p},\vec{\alpha'}^{\lambda})_{S_k}}\delta_{(\vec{q'},\vec{\beta}^{\lambda'})_{S_k}}\delta_{(\vec{p'},\vec{\beta'}^{\lambda'})_{S_k}}\right)\\
    &=\prod_{k\neq h} \left(\delta_{(\vec{\alpha}^{\lambda},\vec{\alpha'}^{\lambda})_{S_k}}\delta_{(\vec{\beta}^{\lambda'},\vec{\beta'}^{\lambda'})_{S_k}}+\delta_{(\vec{\alpha}^{\lambda},\vec{\beta'}^{\lambda'})_{S_k}}\delta_{(\vec{\alpha'}^{\lambda},\vec{\beta}^{\lambda'})_{S_k}}\right)\nonumber\,.
\end{align*} }
Let us define {\footnotesize
\begin{align*}
    J:=&\sum_{\lambda,\lambda'}p_\lambda p_{\lambda'} \sum_{\substack{\vec{\alpha}^{\lambda},\vec{\alpha'}^{\lambda}\\\vec{\beta}^{\lambda'},\vec{\beta'}^{\lambda'}}}c_{\vec{\alpha}^{\lambda}}c^*_{\vec{\alpha'}^{\lambda}}c_{\vec{\beta}^{\lambda'}}c^*_{\vec{\beta'}^{\lambda'}}\\
    &\times\left(\delta_{(\vec{\alpha'}^{\lambda},\vec{\beta}^{\lambda'})_{S_h}}\delta_{(\vec{\alpha}^{\lambda},\vec{\beta'}^{\lambda'})_{S_h}}-\frac{1}{2^m}\delta_{(\vec{\alpha}^{\lambda},\vec{\alpha'}^{\lambda})_{S_h}}\delta_{(\vec{\beta}^{\lambda'},\vec{\beta'}^{\lambda'})_{S_h}}\right)\prod_{k\neq h} \left(\delta_{(\vec{\alpha}^{\lambda},\vec{\alpha'}^{\lambda})_{S_k}}\delta_{(\vec{\beta}^{\lambda'},\vec{\beta'}^{\lambda'})_{S_k}}+\delta_{(\vec{\alpha}^{\lambda},\vec{\beta'}^{\lambda'})_{S_k}}\delta_{(\vec{\alpha'}^{\lambda},\vec{\beta}^{\lambda'})_{S_k}}\right)\,,
\end{align*}}
\!\!\!which,  has the form of $J = 1-\frac{1}{2^m}+\sum_{l=1}^{2^{\xi-1}-1}\left(\AC_{\text{L}}-\frac{1}{2^m}\BC_{\text{L}}\right)$,
where $\AC_{\text{L}},~\BC_{\text{L}}$ are the purities of reduced states of $\rho$. The latter can be understood in the following way. 
Let us define the set $\SC=\{S_1,S_2,\cdots, S_\xi\}$, whose elements represents each block. Then, suppose that we want to consider the partial trace of $\rho$ over several subsystems $\HC_{\overline{\vec{K}}}$, where $\overline{\vec{K}}\subset\mathcal{S}$. The reduced state lives in the composite Hilbert space $\HC_{\vec{K}}$, where $\vec{K}=\SC\setminus \overline{\vec{K}}$, and we can write $\vec{\alpha}:= \vec{\alpha}_{\vec{K}}\cdot \vec{\alpha}_{\overline{\vec{K}}}$. Let us define the reduced state as $\rho_{\vec{K}} := \Tr_{\overline{\vec{K}}}\left[\rho\right]$, and its purity can be explicitly written as
\begin{align*}
    \Tr\left[\rho_{\vec{K}}^2\right] = \sum_{\lambda,\lambda'}p_{\lambda} p_{\lambda'} \sum_{\vec{\alpha}^{\lambda},\vec{\beta}^{\lambda'}} c_{\vec{\alpha}^{\lambda}}c^*_{\vec{\alpha}_{\overline{\vec{K}}}^{\lambda}\cdot\vec{\beta}_{\vec{K}}^{\lambda'}}c_{\vec{\beta}^{\lambda'}}c^*_{\vec{\alpha}_{\vec{K}}^{\lambda}\cdot\vec{\beta}_{\overline{\vec{K}}}^{\lambda'}}\,,
\end{align*}
which appears in the expression of $J$.
Since $\Tr\left[\rho_{\vec{K}}^2\right]\leq1$, we have $J<2^{\xi-1}$; therefor, we can write 
\begin{align}\label{eq:SMGlobal}
    \Var\left[\partial_\nu C_G\right] < \frac{2^{2m-2}}{(2^{2m}-1)2^{2m\xi-\xi}(1+2^{-m})^{\xi}}<\frac{1}{2^{\left(2-\frac{1}{m}\right)n}}~~ (\forall m\in\mathbb{N})\,,
\end{align}
where in the last inequality we used the fact that $n=m\xi$. Equation~\eqref{eq:SMGlobal} shows that for the single layer of the alternating layered ansatz the global cost functions presents a barren plateau.

\subsection{Variance of the local cost function partial derivative}
Here we consider the case when $V(\thv)$ is given by a single layer of the Alternating Layered Ansatz, and when the local cost function is
\begin{align*}
    C_{\text{L}}=1-\frac{1}{n}\sum_{j=1}^n\Tr\left[\left(\dya{0}_j\otimes\id_{\overline{j}}\right)V\rho V\ad\right]\,.
\end{align*}
with  $O_{\text{L}} = \frac{1}{n}\sum_{j=1}^n\dya{0}_j\otimes\id_{\overline{j}}$, and
where $j$ denotes the $j$-th qubit. Moreover, let us assume that we are training a parameter $\theta_\nu$ in $h$-th block. 
This case is simpler than the one previously considered for the global cost function as  the gradient is non vanishing only when we measure a qubit $j$ in $S_{h}$. We can then redefine $O_{\text{L}}^h: \HC_h \rightarrow \HC_h$ as $\widehat{O}_{\text{L}}^h=\frac{1}{n}\sum_{j=1}^m\dya{0}_j\otimes\id_{\overline{j}}$,
where each $j$ is such that $j\in S_h$.
Then, from~\eqref{eq:SMVarWB2} we have
\begin{equation}
    \Tr\left[ \widehat{O}_{\text{L}}^h\right] = \frac{m2^{m-1}}{n} \,,\quad
    \Tr\left[\left(\widehat{O}_{\text{L}}^h\right)^2\right]=
    \frac{m(m+1)2^{m-2}}{n^2} \,.
    \label{eq:localoperatortrace}
\end{equation}

From \eqref{eq:eq-Omegapq}, $\Omega_{\vec{qp}}$ can be written as $\Omega_{\vec{qp}}=\delta_{\vec{pq}}W_hO_{\text{L}}^hW_h\ad$, which leads to 
\begin{align}
    \Delta\Omega_{\vec{p}\vec{q}}^{\vec{p}'\vec{q}'} 
    =\frac{m2^{m-2}}{n^2}\delta_{\vec{pq}}\delta_{\vec{p'q'}}\,.
\end{align}
Next, from \eqref{eq:eq-Psipq} we have  $\Psi_{\vec{pq}} = \Tr_{\overline{h}}\left[\left(\ketbra{\vec{q}}{\vec{p}}\otimes\id_{h}\right)\rho\right]$,
and it is straightforward to see that 
\begin{align}
    \sum_{\vec{pq}}\delta_{\vec{pq}}\Psi_{\vec{pq}}=\sum_{\vec{pq}}\delta_{\vec{pq}}\Tr_{\overline{h}}\left[\left(\ketbra{\vec{q}}{\vec{p}}\otimes\id_{h}\right)\rho\right]=\rho_h\,,
\end{align}
where we defined $\rho_h:=\Tr_{\overline{h}}\left[\rho\right]$. Then, from \eqref{eq:SMVarWB2}, and by using $ \left\langle\Delta\Psi_{\vec{p}\vec{q}}^{\vec{p}'\vec{q}'}\right\rangle_{V_{\text{L}}}=\Delta\Psi_{\vec{p}\vec{q}}^{\vec{p}'\vec{q}'}$, and $\Tr[\sigma_\nu^2]=2^m$, we can write 
\begin{align*}
 \Var\left[\partial_\nu C_{\text{L}}\right]=
 \frac{m2^{3(m-1)}}{n^2(2^{2m}-1)^2}\left(\Tr[\rho_h^2]-\frac{1}{2^m}\right)=\frac{m\cdot 2^{3(m-1)}}{n^2(2^{2m}-1)^2} D_{HS}\left(\rho_h,\frac{\id}{2^m}\right)\,,
\end{align*}
where $D_{HS}\left(\rho_h,\id/2^m\right)=\Tr[(\rho_h-\id/2^m)^2]$ is the Hilbert-Schmidt distance between $\rho_h$ and $\id/2^m$. 
If  $D_{HS}\left(\rho_h,\id/2^m\right) \in \Omega\left(1/\poly(n)\right)$,
we have that the variance of the cost function partial derivative is polinomially vanishing with $n$ as
\begin{align}
  \Var[\partial_\nu C_{\text{L}}] \in \Omega\left(\frac{1}{\poly(n)}\right)\,,
\end{align}
and hence in this case $C_{\text{L}}$ presents no  barren plateau.

\section{Proof of Theorem $2$}
\label{sec:proofTheorem2}

First, let us recall that we are considering  $m$-local cost functions  where each operator $O_i$ acts nontrivially on $m$ qubits $O_i=\id_{\overline{m}}\otimes\widehat{O}_i$ (here $\id_{\overline{m}}$ indicates the identity on all but $m$ qubits), and where $\widehat{O}_i$ can be expressed as $\widehat{O}_i=\widehat{O}_i^{\mu_i}\otimes \widehat{O}_{i}^{\mu_i'}$. Hence, we have
\begin{equation}\label{eq:Omlsm}
    O=c_0\id+\sum_{i=1}^N c_i \widehat{O}_i^{\mu_i}\otimes \widehat{O}_{i}^{\mu_i'} \,,
\end{equation}
where  $\widehat{O}_i^{\mu_i}$ are operators acting on $m/2$ qubits which can be written as tensor product of Pauli operators. Here we recall that we have defined  $S_{k}$ as  the $m$-qubit subsystem on which $W_{kL}$ acts and let $\SC = \{S_{k}\}$ be the set of all such subsystems. As detailed in the main text the summation  in Eq.~\eqref{eq:Omlsm} includes two possible cases: First, when  $\widehat{O}_i^{\mu_i}$ ($\widehat{O}_i^{\mu_i'}$) acts on the first (last) $m/2$ qubits of a given $S_{k}$, and second, when $\widehat{O}_i^{\mu_i}$ ($\widehat{O}_i^{\mu_i'}$) acts on the last (first) $m/2$ qubits of a given $S_{k}$ ($ S_{k+1}$). This type of cost function includes any ultralocal (i.e., where the $O_i$ are one-body) cost function, and also VQE Hamiltonians with up to $m/2$ neighbor interactions.

Here we provide a proof of Theorem $2$ in the main text, which we now reiterate for convenience.

\setcounter{theorem}{1}
\begin{theorem}\label{thm2sm}
Consider a trainable parameter  $\theta^\nu$ in a block $W$ of the ansatz in Fig. 3 of the main text.  Let $\Var[\partial_\nu C]$ be the variance of the partial derivative of an $m$-local cost function $C$ (with $O$ given by~\eqref{eq:Omlsm}) with respect to $\theta^\nu$. If $W_{\text{A}}$, $W_{\text{B}}$ of~\eqref{eq:WA-WB}, and each block in $V(\vec{\theta})$ form a local $2$-design, then $\Var[\partial_\nu C]$ is lower bounded by 
\begin{equation}
    G_{n}(L,l)\leq\Var[\partial_\nu C]\,,\label{eq:varMaint2sm}
\end{equation}
with 
\begin{align}\label{eqSM:varMaint22sm}
    G_{n}(L,l)&=
    \frac{2^{m(l+1)-1}}{(2^{2m}-1)^2(2^{m}+1)^{L+l}}\sum_{ i\in i_\LC} \sum_{\substack{(k,k')\in k_{\LC_{\text{B}}} \\ k' \geq k}} c_i^2  \epsilon(\rho_{k,k'}) \epsilon(\widehat{O}_i)\,,
\end{align}
where $i_\LC$ is the set of $i$ indices whose associated operators $\widehat{O}_i$ act on qubits in the forward light-cone $\LC$ of $W$, and $k_{\LC_{\text{B}}}$ is the set of $k$ indices whose associated subsystems $S_k$ are in the backward light-cone $\LC_{\text{B}}$ of $W$. Here we defined the function $\epsilon(M) = D_{\HS}\left(M,\Tr(M)\id / d_M\right)$ where $D_{\HS}$ is the Hilbert-Schmidt distance and $d_M$ is the dimension of the matrix $M$. In addition, $\rho_{k,k'}$ is partial trace of the input state $\rho$ down to the subsystems $S_k S_{k+1}... S_{k'}$. 
\end{theorem}

\begin{proof}

Let us first consider the case when the operators  $O_i$ are of the form~\eqref{eq:Omlsm} and act non trivially in a given subsystem of $\mathcal{S}$. We can expand 
\begin{equation}
    \Var[\partial_\nu C]=\frac{2^{m-1}\Tr[\sigma_\nu^2]}{(2^{2m}-1)^2}\sum_{i,j}\sum_{\substack{\vec{p}\vec{q}\\\vec{p}'\vec{q}'}}c_ic_j\left\langle\Tr[ \Omega_{\vec{q}\vec{p}}^i\Omega_{\vec{q}'\vec{p}'}^j] -\frac{\Tr[ \Omega_{\vec{q}\vec{p}}^i]\Tr[\Omega_{\vec{q}'\vec{p}'}^j]}{2^{m}}\right\rangle_{V_{\text{R}}}\left\langle\Delta\Psi_{\vec{p}\vec{q}}^{\vec{p}'\vec{q}'}\right\rangle_{V_{\text{L}}}\,,\label{eq:varlocal}
\end{equation}
where we defined
\begin{align}
    \Omega_{\vec{q}\vec{p}}^i&=\Tr_{\overline{w}}\left[(\ketbra{\vec{p}}{\vec{q}}\otimes\id_w)V_{\text{R}}\ad O_i V_{\text{R}}\right]\,. 
    \label{eq:Omegamni}
\end{align}

\subsubsection{Expectation value over $V_{\text{R}}$}

We recall here that  we assume without loss of generality that  $V_{\text{R}}$ contains the gates in $\LC$ and all the blocks $W_{kL}$ in the last layer of $V(\thv)$. Hence,  we can express
\begin{align}\label{eq:V-LCb}
    V_{\text{R}} =   \VLb\otimes \VL\,,
\end{align}
where $\VL$ contains  all the blocks in the forward light-cone $\LC$ of $W$, and where $\VLb$ consists of all the remaining blocks in the last layer of the ansatz.  When analyzing the operators $\Omega_{\vec{p}\vec{q}}^i$ we have to consider two cases: 1) when $O_i$ only acts non-trivially on qubits in $S_{\LCb}$, and 2)  when $\widehat{O}_i$ acts non-trivially on qubits in  $S_\LC$. If $O_i$ only acts non-trivially on qubits in $S_{\LCb}$, it is straightforward to show from~\eqref{eq:Omegamni} that $\Omega_{\vec{p}\vec{q}}^i\propto\id_w$. Hence, in this case, we find
\begin{equation}
    \Tr[ \Omega_{\vec{q}\vec{p}}^i\Omega_{\vec{q}'\vec{p}'}^j] -\frac{\Tr[ \Omega_{\vec{q}\vec{p}}^i]\Tr[\Omega_{\vec{q}'\vec{p}'}^j]}{2^{m}}=0\,.
\end{equation}
Let us then define $i_\LC$ as the set of $i$ indices whose associated operators $O_i$ act on qubits in the forward light-cone $\LC$ of $W$.  In what follows we
assume that the indexes $i,j\in i_\LC$, i.e., we assume that $O_i$ and $O_j$ act non-trivially on qubit in $S_\LC$. The latter leads to
\begin{equation}
    \left\langle\Tr[ \Omega_{\vec{q}\vec{p}}^i\Omega_{\vec{q}'\vec{p}'}^j] -\frac{\Tr[ \Omega_{\vec{q}\vec{p}}^i]\Tr[\Omega_{\vec{q}'\vec{p}'}^j]}{2^{m}}\right\rangle_{V_{\text{R}}}=\delta_{(\vec{p},\vec{q})_{S_{\LCb}}}\delta_{(\vec{p}',\vec{q}')_{S_{\LCb}}}\left\langle\Tr[ \Omega_{\vec{q}\vec{p}}^i\Omega_{\vec{q}'\vec{p}'}^j] -\frac{\Tr[ \Omega_{\vec{q}\vec{p}}^i]\Tr[\Omega_{\vec{q}'\vec{p}'}^j]}{2^{m}}\right\rangle_{V_{\LC}}. \label{eq:EV-pre-TN}
\end{equation}
Here  the delta functions arise from~\eqref{eq:Omegamni} by noting that $V_{\text{R}}\ad O_i V_{\text{R}}=V_\LC\ad O_i V_\LC$.

In order to explicitly evaluate the  expectation value in  \eqref{eq:EV-pre-TN} we use the fact that each block in the layered ansatz is a $2$-design, and hence one can algorithmically integrate over each block using the Weingarten calculus. Specifically, as discussed in the main text we employ the tensor network representation of  $\Tr[ \Omega_{\vec{q}\vec{p}}^i\Omega_{\vec{q}'\vec{p}'}^i]$ and $\Tr[ \Omega_{\vec{q}\vec{p}}^i]\Tr[\Omega_{\vec{q}'\vec{p}'}^j]$, and we use the  {\it Random Tensor Network Integrator} (RTNI) package of Ref.~\cite{Fukuda_2019}, which allows for the computation of averages of tensor networks containing multiple Haar-distributed random unitary matrices and deterministic symbolic tensors.

Using this procedure, $\langle\dots\rangle_{V_{\LC}}$ can be computed by integrating over each block inside of $V_\LC$ over the unitary group with respect to the Haar measure.  After each integration the result of the average is a sum of up to four new tensor  according to  \eqref{eq:delta}, and Lemmas~\ref{lemma2} and~\ref{lemma3}. After all the blocks in $V_\LC$ have been integrated, the final result can be expressed as
\begin{equation}
    \left\langle\Tr[ \Omega_{\vec{q}\vec{p}}^i\Omega_{\vec{q}'\vec{p}'}^j] -\frac{\Tr[ \Omega_{\vec{q}\vec{p}}^i]\Tr[\Omega_{\vec{q}'\vec{p}'}^j]}{2^{m}}\right\rangle_{V_{\LC}}=\sum_\tau t_\tau^{ij} 
    \delta_{(\vec{p},\vec{q})_{S_{\overline{\tau}}}}\delta_{(\vec{p}',\vec{q}')_{S_{\overline{\tau}}}}\delta_{(\vec{p},\vec{q}')_{S_\tau}}\delta_{(\vec{p}',\vec{q})_{S_\tau}}\Delta O^{ij}_\tau\,,
    \label{eq:Tn_Sumsm}
\end{equation}  
where $t_\tau\in \mathbb{R}$, $S_\tau\cup S_{\overline{\tau}}=S_\LC\cap S_{\overline{w}}$ (with $S_\tau\neq\emptyset)$, and where we have defined
\begin{equation}
\Delta O^{ij}_\tau=\Tr_{x_\tau y_\tau}\left[\Tr_{z_\tau}\left[O_i\right]\Tr_{z_\tau}\left[O_j\right]\right]-\frac{\Tr_{x_\tau}\left[\Tr_{y_\tau z_\tau}\left[O_i\right]\Tr_{y_\tau z_\tau}\left[O_j\right]\right]}{2^m}\,.
\label{eq:DeltaO}
\end{equation}
Here we use the notation $\Tr_{x_\tau}$ to indicate the trace over the Hilbert space associated with subsystem $S_{x_\tau}$. We also define $\mathcal{S}^\LC=\{S_k:S_k\subset S_\LC\}$ as the set of subspaces $S_k$ which belong to the light-cone $\LC$, so that we have $S_{y_\tau}\in\mathcal{S}^\LC$,  $S_{x_\tau}\cup S_{y_\tau}\cup S_{z_\tau}=S_\LC$, and $S_{x_\tau},S_{z_\tau}\in\mathcal{P}(\mathcal{S}^\LC)$, with $\mathcal{P}(\mathcal{S}^\LC)=\{\emptyset, S_1,\ldots,S_\xi,S_1\cup S_2\ldots\}$ the power set of $\mathcal{S}^\LC$. If one chooses $S_{z_\tau}=\emptyset$ (or $S_{x_\tau}=\emptyset$), with associated Hilbert space $H_{0}=\{\vec{0}\}$, we define $\Tr_{\emptyset}\left[O\right]:=O$. 

\begin{figure}
    \centering
    \includegraphics[width=.95\columnwidth]{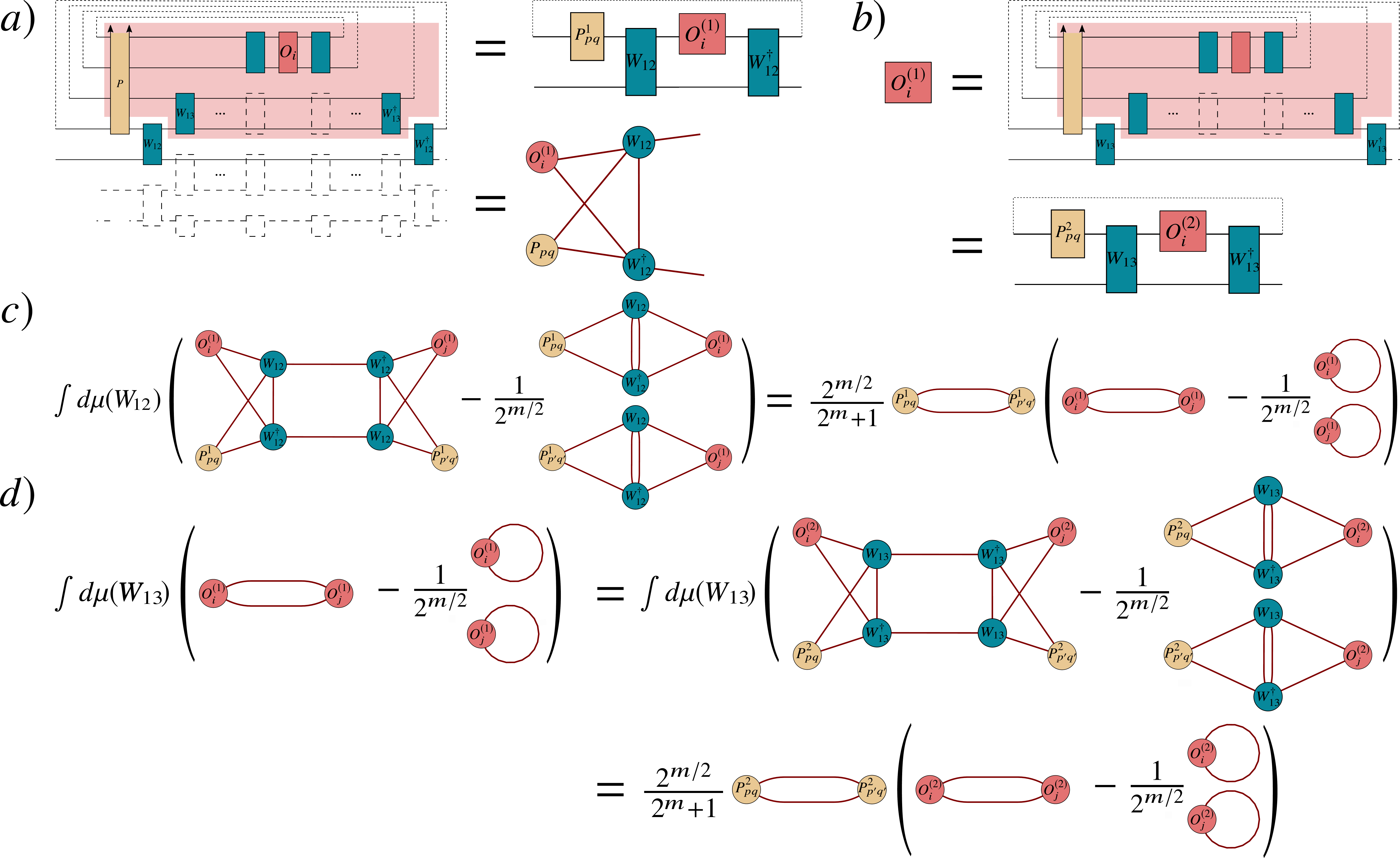}
    \caption{Schematic representation of the method employed to derive Eq.~\eqref{eq:Tn_Sumsm}. Here we consider the case when $O_i$ and $O_j$ act nontrivially only on the topmost $m$ qubits in $V_{\LC}$. (a) Tensor-network representation of the term $\Omega_{\vec{q}\vec{p}}^i$.     All but $L-l$ blocks in $V_\LC$ simplify to identity when computing $V_\LC\ad O_i V_\LC$. One can always define operators  $O_i^{(1)}$, and $O_j^{(1)}$ such that Eq.~\eqref{eq:proof1} holds. Here $W_{12}$ is the topmost block in the first layer of $V_\LC$, and we define $P^1_{\vec{p}\vec{q}}$ as a projector operator acing on $m/2$ qubits.  (b) Tensor-network representation of the operator $O_i^{(1)}$. Note that from $O_i^{(1)}$ we can also define an operator $O_i^{(2)}$ which will be of the same form as $O_i^{(1)}$. (c) By means of the Weingarten calculus we can integrate over $W_{12}$ to compute the expectation value of Eq.~\eqref{eq:proof1}. The result is a single tensor. This result can be obtained by employing Eq.~\eqref{eq:lemma6eq1} of Lemma~\ref{lemma3.5}. (d) In order the integrate the next block we use the operators $O_i^{(2)}$ to express the tensor networks in a form that coincides with panel (c). We can then repeat this procedure $L-l$ times in order to integrate over all blocks and compute the expectation value in \eqref{eq:EV-pre-TN}. The final result is of the form of Eq.~\eqref{eq:proofeq3}.    }
    \label{fig:SM_Proof}
\end{figure}

Equation~\eqref{eq:Tn_Sumsm} can be derived by iteratively applying Lemma~\ref{lemma3.5}  each time a block in  $V_\LC$ is integrated in~\eqref{eq:EV-pre-TN}, and noting that at each integration the resulting tensor can always be written in a way such that Lemma~\ref{lemma3.5} can be applied again. As an example, let us consider the case when $O_i$ and $O_j$ act nontrivially only on the topmost $m$ qubits in $V_{\LC}$. As shown in Sup. Fig.~\ref{fig:SM_Proof}(a), all but $L-l$ blocks in $V_\LC$ simplify to identity when computing $V_\LC\ad O_i V_\LC$. Then, let us 
denote by $W_{12}$ the topmost gate in the first layer in $V_{\LC}$, and let $\widetilde{S}_1$, $\widetilde{S}_2$ be the sets of $m/2$ adjacent qubits (with associated Hilbert spaces $\widetilde{\HC}_1$, $\widetilde{\HC}_2$) such that $W_{12}$ acts on  $\widetilde{\HC}_1\otimes \widetilde{\HC}_2$, and such that $\widetilde{S}_2\subset S_w$. One can always write
\begin{equation}
  \left\langle\Tr[ \Omega_{\vec{q}\vec{p}}^i\Omega_{\vec{q}'\vec{p}'}^j] -\frac{\Tr[ \Omega_{\vec{q}\vec{p}}^i]\Tr[\Omega_{\vec{q}'\vec{p}'}^j]}{2^{m}}\right\rangle_{V_{\LC}}\propto\left\langle \Tr_{\widetilde{\HC}_2}\left[\Tr_{\widetilde{\HC}_1}[\Omega_1] \Tr_{\widetilde{\HC}_1}[\Omega'_1]\right]-\frac{\Tr[\Omega_1]\Tr[\Omega'_1]}{2^m}\right\rangle_{V_{\LC}}\,,\label{eq:proof1}
\end{equation}
 with $\Omega_{1}= W_{12} O_i^{(1)}W_{12}\ad(\ketbra{\vec{p}}{\vec{q}}_{\widetilde{S}_1}\otimes\id_{\widetilde{S}_2})$,  $\Omega_{1}'= W_{12} O_j^{(1)} W_{12}\ad(\ketbra{\vec{p}'}{\vec{q}'}_{\widetilde{S}_1}\otimes\id_{\widetilde{S}_2})$, and where $O_i^{(1)}$, and $O_j^{(1)}$ are defined according to Sup. Fig.~\ref{fig:SM_Proof}(a). The proportionallity factor in Eq.~\eqref{eq:proof1} is given by delta functions over $\vec{p}$, and $\vec{q}$. Moreover, the right-hand side of~\eqref{eq:proof1} is exactly of the form~\eqref{eq:lemma6Proj}, and hence  by applying  Lemma~\ref{lemma3.5} (or more specifically, Eq.~\eqref{eq:lemma6eq1}) one finds the result of Sup. Fig.~\ref{fig:SM_Proof}(c). As schematically depicted in Sup. Fig.~\ref{fig:SM_Proof}(b), one can then define operators $ O_i^{(2)}$ and repeat this calculation $L-l$ times as in Sup. Fig.~\ref{fig:SM_Proof}(d). One finally finds  
\begin{equation}
    \left\langle\Tr[ \Omega_{\vec{q}\vec{p}}^i\Omega_{\vec{q}'\vec{p}'}^j] -\frac{\Tr[ \Omega_{\vec{q}\vec{p}}^i]\Tr[\Omega_{\vec{q}'\vec{p}'}^j]}{2^{m}}\right\rangle_{V_{\LC}}\propto \frac{2^{m(L-l)/2}}{(2^{m}+1)^{L-l}} \left(\Tr\left[O_iO_j\right]-\frac{\Tr\left[O_i\right]\Tr\left[O_j\right]}{2^m}\right)
    \,, \label{eq:proofeq3}
\end{equation}  
where once again the proportionally factor is given by delta functions over $\vec{p}$ and $\vec{q}$. For more general operators $O_i$ and $O_j$, each time a block is integrated, one can always rewrite  the resulting non-trivial terms in the form of either Eq.~\eqref{eq:lemma6Proj}, or Eq.~\eqref{eq:lemma6NoProj}. Remarkably, this means that the final result can always be expressed in terms of contractions of the form~\eqref{eq:DeltaO}, and hence Eq.~\eqref{eq:Tn_Sumsm} holds.

As previously mentioned, Eq.~\eqref{eq:Tn_Sumsm} is valid for arbitrary operators $O_i$ and $O_j$. However, from Lemma~\ref{lemma6} we know that if $O_i$ and $O_j$ have no overlapping support, then $\Delta O^{ij}_\tau=0$, for all $\tau$. Hence we only have to consider the cases when $i=j$. Moreover, if $\widehat{O}_i$ act non trivially in a given subsystem of $\mathcal{S}$ the summation in Eq.~\eqref{eq:Tn_Sumsm} simplifies and we find
\begin{equation}
    \left\langle\Tr[ \Omega_{\vec{q}\vec{p}}^i\Omega_{\vec{q}'\vec{p}'}^j] -\frac{\Tr[ \Omega_{\vec{q}\vec{p}}^i]\Tr[\Omega_{\vec{q}'\vec{p}'}^j]}{2^{m}}\right\rangle_{V_{\LC}}=
    \epsilon(\widehat{O}_i)\sum_\tau \hat{t}_\tau^{ii} \delta_{(\vec{p},\vec{q})_{S_{\overline{\tau}}}}\delta_{(\vec{p}',\vec{q}')_{S_{\overline{\tau}}}}\delta_{(\vec{p},\vec{q}')_{S_\tau}}\delta_{(\vec{p}',\vec{q})_{S_\tau}} \,,
    \label{eq:Tn_Sum-localsm}
\end{equation}  
where we now denote the coefficients as $\hat{t}_\tau^{ii}$ since we now have  $\hat{t}_\tau^{ii}\geq 0$,  and where 
\begin{align}\label{eq:epsiloniSI}
    \epsilon(\widehat{O}_i)=\Tr\left[\widehat{O}_i^2\right]-\frac{\Tr\left[\widehat{O}_i\right]^2}{2^m}=D_{HS}\left(\widehat{O}_i,\Tr[\widehat{O}_i]\frac{\id}{2^m}\right)\,.
\end{align}
Moreover, we also find that the following inequality holds $\forall i$
\begin{equation}
    \sum_\tau \hat{t}_\tau^{ii}\geq \frac{2^{m(L-l)/2}}{(2^{m}+1)^{L-l}}\,, \label{eq:lower-bound-tk}
\end{equation}
where we recall that $L$ is the number of layers in the ansatz, and that the block $W$ is in  the $l$th-layer of $V(\vec{\theta})$.

Combining Eqs.~\eqref{eq:varlocal} and~\eqref{eq:Tn_Sum-localsm}, we find
\begin{equation}
    \Var[\partial_\nu C]=\frac{2^{m-1}\Tr[\sigma_\nu^2]}{(2^{2m}-1)^2}\sum_{\substack{\vec{p}\vec{q}\\\vec{p}'\vec{q}'}}\sum_{i\in i_\LC}c_i^2\epsilon(\widehat{O}_i)\sum_\tau \hat{t}_\tau^{ii} \delta_{(\vec{p},\vec{q})_{S_{\LCb}\cup S_{\overline{\tau}} }}\delta_{(\vec{p}',\vec{q}')_{ S_{\LCb}\cup S_{\overline{\tau}}}}\delta_{(\vec{p},\vec{q}')_{S_\tau}}\delta_{(\vec{p}',\vec{q})_{S_\tau}}
    \left\langle\Delta\Psi_{\vec{p}\vec{q}}^{\vec{p}'\vec{q}'}\right\rangle_{V_{\text{L}}}\,.\label{eq:varlocalfull}
\end{equation}

Before proceeding to compute the expectation value over $V_{\text{L}}$, let us analyze the lower bound  in~\eqref{eq:lower-bound-tk}. For simplicity of notation, let us define $T_\tau:=\sum_\tau \hat{t}_\tau^{ii}$. In what follows we compare $T_\tau$ for two relevant (and extremum) cases: 1) when $O_i$ acts nontrivially on the topmost $m$ qubits in $V_{\LC}$ (see panel (a) of Fig.~\eqref{fig:SM_Proof_New}), and 2) when  $O_i$ acts nontrivially on $S_w$ (see panel (b) of Fig.~\eqref{fig:SM_Proof_New}).  As discussed below, $T_\tau$ is larger for Case 2, that for Case 1. 

\begin{figure}
    \centering
    \includegraphics[width=.95\columnwidth]{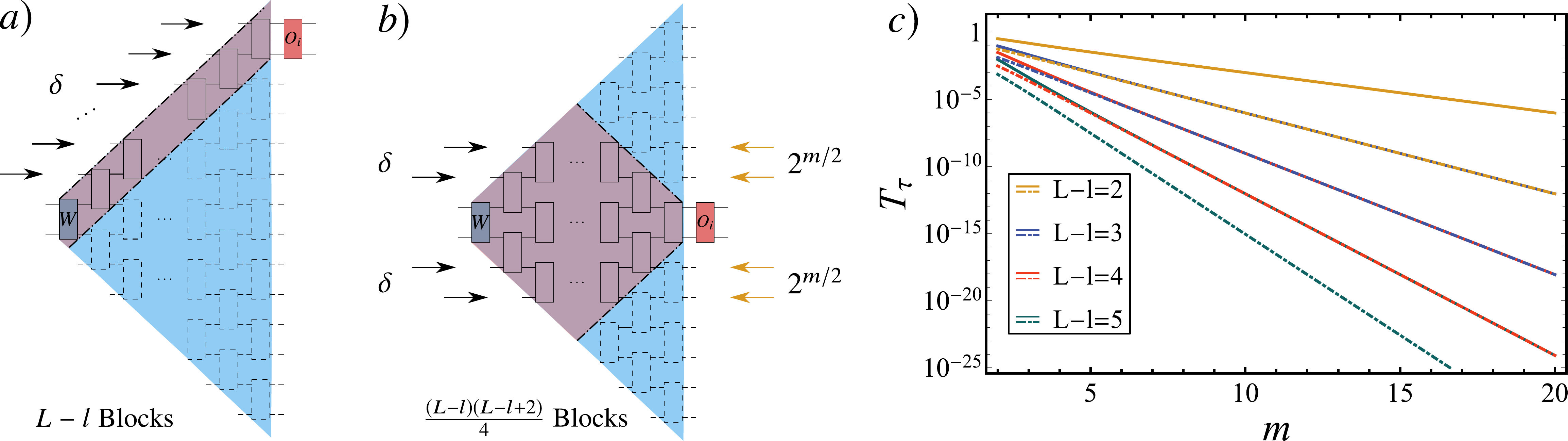}
\caption{a) Schematic representation of Case 1. The operator $O_i$ acts nontrivially on the $m$ topmost qubits in $V_{\LC}$. Note that this scenario corresponds to the one analyzed in Sup. Fig.~\ref{fig:SM_Proof}. Here one needs to integrate over $(L-l)$ blocks. The arrows indicate the tensor indexes that, when contracted after integrating over each block, will lead to  delta functions in $\vec{p}$ and $\vec{q}$. b) Schematic representation of Case 2. The operator $O_i$ acts non trivially on $S_w$. Here one needs to integrate over $(L-l)(L-l+2)/4$ blocks. For the case of $L$ odd, the operator $O_i$ acts nontrivially on the top $m/2$ qubits of $S_w$ and the $m/2$ qubits above those. The arrows indicate the tensor indexes that, when contracted after integrating over each block, will lead to  either delta functions in $\vec{p}$ and $\vec{q}$, or to traces over the identity on $m/2$ qubits. c) Solid lines represent the value of $T_\tau$ versus $m$ for Case 2 in panel b) (i.e., from Eq.~\eqref{eq:Ttau2}).    Dashed lines represent the value of $T_\tau$ versus $m$ for Case 1 in panel a) (i.e., from Eq.~\eqref{eq:Ttau1}). The $y$ axis is shown in a logarithmic scale. For all values of $(L-l)$ considered the values of $T_\tau$ Case 2 are larger that those Case 1. For a given value of $(L-l)$ the difference between values of $T_\tau$ in~\eqref{eq:Ttau1}  and in~\eqref{eq:Ttau2} scales as  $1/2^{m(L-l-1)/2}$. The previous explains why as $m$ increases we the dashed line for $(L-l)$ converges to the solid line with $(L-l+1)$.    }
    \label{fig:SM_Proof_New}
\end{figure}

As schematically depicted in Sup. Fig.~\eqref{fig:SM_Proof_New}(a), for Case 1 one needs to integrate over $(L-l)$ blocks, and we know from Eq.~\eqref{eq:proofeq3} (and Sup. Fig.~\ref{fig:SM_Proof}) that 
\begin{equation}\label{eq:Ttau1}
T_\tau=\frac{2^{m(L-l)/2}}{(2^{m}+1)^{L-l}}\,.    
\end{equation}
On the other hand, for Case 2 one needs to integrate over $(L-l)(L-l+2)/4$ blocks. While a closed formula cannot be derived for this case, we can explicitly integrate over each block to obtain:
\begin{equation}\label{eq:Ttau2}
     T_\tau =
    \begin{cases}
      \frac{2^{3m/2} (1 + 2^{m/2})^2}{(1 + 2^m)^3} & \text{if $L-l=2$}\\
      \frac{2^{2m} (1 + 2^{m/2})^2 (1 + 2^{m/2} + 2^{m})}{1 + d^2)^5} & \text{if $L-l=3$}\\
      \frac{2^{5m/2} (1 + 2^{m/2} + 2^{m}) (1 + 2^{m/2} (3 + 2^{m/2} (8 + 2^{m/2} (3 + 2^{m/2}))))}{(1 + 2^m)^7} & \text{if $L-l=4$}\\
      \frac{2^{3m} (1 + 2^{m/2} (2 + 2^{m/2} (4 + 2^{m/2} (2 + 2^{m/2})))) (1 + 
   2^{m/2} (3 + 2^{m/2} (8 + 2^{m/2} (3 + 2^{m/2}))))}{(1 + 2^m)^9} & \text{if $L-l=5$}\\
    \end{cases} \,.
\end{equation}
For all cases we can see from Sup. Fig.~\eqref{fig:SM_Proof_New}(c) that the values of $ T_\tau$ for Case 2 are larger than the ones obtained for Case 1.

The latter can be understood from two key facts. First, that while integrating over a block leads to coefficients smaller than one according to~\eqref{eq:delta}, the more blocks one integrates over, the more contributions we get for $\epsilon(\widehat{O}_i)$. Second, from the fact that in Case 2 the integration over some blocks will lead to multiplicative factors $2^{m/2}$ which will compensate the factors from~\eqref{eq:delta}. These large factors $2^{m/2}$  essentially arise when integrating blocks in $V_\mathcal{L}$ which act on qubits which one does not measure, and which hence lead to traces over the identity over $m/2$ qubits.  Note that these factors do not arise in Case I. In Sup. Fig.~\eqref{fig:SM_Proof_New}(a) we have schematically indicated by arrows the indexes that  will lead  to delta functions in $\vec{p}$ and $\vec{q}$ after each block is integrated. Each time we integrate a block  containing these indexes one obtains tensor contractions of $\ketbra{\vec{p}}{\vec{q}}$. On the other hand, for Case 2, one will not only get these tensor contractions in $\vec{p}$ and $\vec{q}$, but as indicated in Sup.  Fig.~\eqref{fig:SM_Proof_New}(b) we have also indexes that will lead to identity-tensor contractions and hence to factors $2^{m/2}$.

In fact, for a given $L$ the difference between the $T_\tau$ in~\eqref{eq:Ttau1} and in~\eqref{eq:Ttau2} scales as  $1/2^{m(L-l-1)/2}$, which means that as $(L-l)$ increases the inequality in~\eqref{eq:lower-bound-tk} becomes tighter. Finally, we remark that the previous explains why in Fig.~\eqref{fig:SM_Proof_New}(c) the values of $T_\tau$ from Case 1 and fixed $(L-l)$ converge to the values of Case 2 and $(L-l+1)$.

\subsubsection{Expectation value over $V_{\text{L}}$}

Let us now consider the term
\begin{equation}
    \left\langle\Delta\Psi_{\vec{p}\vec{q}}^{\vec{p}'\vec{q}'}\right\rangle_{V_{\text{L}}}=\left\langle\Tr[\Psi_{\vec{p}\vec{q}}\Psi_{\vec{p}'\vec{q}'}]-\frac{\Tr[\Psi_{\vec{p}\vec{q}}]\Tr[\Psi_{\vec{p}'\vec{q}'}]}{2^{m}}\right\rangle_{V_{\text{L}}}\,.\label{eq:DHS0}
\end{equation}
Since the bitstring $\vec{p}$ ($\vec{q}$) can be expressed as a bit-wise concatenation of the form $\vec{p}=(\vec{p})_{S_{\LCb}\cup S_{\overline{\tau}}}\cdot(\vec{p})_{S_{\tau}}$ (and similarly for $\vec{q}$),  we can first evaluate   
\begin{equation}
    \sum_{(\vec{p}\vec{q})_{S_{\LCb}\cup S_{\overline{\tau}}}}\delta_{(\vec{p},\vec{q})_{S_{\LCb}\cup S_{\overline{\tau}}}}\Psi_{\vec{q}\vec{p}}=\Tr_{\tau}\left[(\ketbra{\vec{p}}{\vec{q}}_{S_\tau}\otimes\id_w)\tilde{\rho}_{Wk}\right]\,,\label{eq:sumPsimm}
\end{equation}
where $\Tr_{\tau}$ is the partial trace over the Hilbert space $\HC_\tau$ of the qubits in $S_\tau$, and where we defined
\begin{equation}
\tilde{\rho}_{Wk}=\Tr_{{\LCb}\cup {\overline{\tau}}}[V_{\text{L}}\rho V_{\text{L}}\ad]\,,
\end{equation}
as the reduced state of  $V_{\text{L}}\rho V_{\text{L}}\ad$ on $\HC_w\otimes \HC_{\tau}$. Then, from the term $\delta_{(\vec{p},\vec{q}')_{S_\tau}}\delta_{(\vec{p}',\vec{q})_{S_\tau}}$ in~\eqref{eq:varlocalfull}, we get
\begin{align}
    \sum_{(\vec{p}\vec{q})_{S_\tau}}\delta_{(\vec{p},\vec{q}')_{S_\tau}}\delta_{(\vec{p}',\vec{q})_{S_\tau}}\left( \Tr[\Psi_{\vec{p}\vec{q}}\Psi_{\vec{q}\vec{p}}]-\frac{\Tr[\Psi_{\vec{p}\vec{q}}]\Tr[\Psi_{\vec{q}\vec{p}}]}{2^{m}}\right)    &= D_{HS}\left(\tilde{\rho}_{w\tau}, \tilde{\rho}_{\tau}\otimes\frac{\id}{2^m}\right)\,, \label{eq:DHSsm}
\end{align}
where $D_{HS}\left(\rho,\sigma\right)=\Tr[(\rho-\sigma)^2]$ is the Hilbert-Schmidt distance, and where we defined 
\begin{equation}
    \tilde{\rho}_{\tau}=\Tr_{w}(\tilde{\rho}_{w\tau})\,, \quad \text{and} \quad \tilde{\rho}_w=\Tr_{\tau}(\tilde{\rho}_{w\tau})=\Tr_{{\overline{w}}}\left[V_{\text{L}}\rho V_{\text{L}}\ad\right]\,,
\end{equation}
as the reduces states of $V_{\text{L}}\rho V_{\text{L}}\ad$  on subsystem $\HC_{\tau}$, and $\HC_w$, respectively. Equation~\eqref{eq:DHSsm} quantifies how far $\tilde{\rho}_{w\tau}$ is from being a tensor product state where the state on subsystem $S_w$ is maximally mixed.  Evidently, if $ \tilde{\rho}_w$ is maximally mixed then it will be impossible to train any angle in $W$. 

Let us now analyze the following chain of inequalities valid for the Hilbert-Schmidt distance $D_{HS}\left(\tilde{\rho}_{w \tau_t}, \tilde{\rho}_{\tau_t}\otimes\frac{\id}{2^m}\right)$ and any choice of $S_\tau$:
\begin{align}
     D_{HS}\left(\tilde{\rho}_{w \tau_t}, \tilde{\rho}_{\tau_t}\otimes\frac{\id}{2^m}\right)&\geq \frac{4D_{T}\left(\tilde{\rho}_{w \tau_t}, \tilde{\rho}_{\tau_t}\otimes\frac{\id}{2^m}\right)^2}{2^m d_\tau}\label{eq:ineq-TD0}\\
    &\geq \frac{4D_{T}\left(\widetilde{\rho}_w,\frac{\id}{2^m}\right)^2}{2^{m(L-l+2)/2}}\nonumber\\
    &\geq \frac{D_{HS}\left(\widetilde{\rho}_w,\frac{\id}{2^m}\right)}{2^{m(L-l+2)/2}}\,,\label{eq:ineq-TD}
\end{align}
with $D_{T}(A,B)=\Tr\left[\sqrt{(A-B)^2}\right]$ the Trace Distance between the Hermitian operators $A$ and $B$. In the first line we employ the matrix norm equivalence, and we denote as $d_\tau$ the dimension of $\HC_\tau$. The second line is derived from the fact that the Trace Distance is non-increasing over partial trace~\cite{nielsen_chuang}, and from the fact that $d_\tau\leq 2^{m(L-l)/2}$ $\forall \tau$. Finally, the inequality in~\eqref{eq:ineq-TD} employs again the matrix norm equivalence.  From Eqs.~\eqref{eq:ineq-TD}, and~\eqref{eq:varlocalfull}, we find that the following lower bound holds
\begin{equation}
    \Var[\partial_\nu C]\geq\frac{\Tr[\sigma_\nu^2]}{2(2^{2m}-1)^2(2^{m}+1)^{L-l}}\sum_{i\in i_\LC}c_i^2\epsilon_i
    \left\langle D_{HS}\left(\widetilde{\rho}_w,\frac{\id}{2^m}\right)\right\rangle_{V_{\text{L}}}\,.\label{eq:varlocalfull-LB}
\end{equation}

Following a similar procedure as the one previously employed to compute expectation values over $V_{\text{R}}$, we can once again leverage the tensor network representation of quantum circuits to algorithmically integrate over each block and compute $\langle D_{HS}\left(\widetilde{\rho}_W,\frac{\id}{2^m}\right)\rangle_{V_{\text{L}}}$.  After considering that all the blocks in $V_{\text{L}}$ which are not in the backpropagated light-cone of $W$ will simplify to identity, we find 
\begin{align}
    \left\langle D_{HS}\left(\tilde{\rho}_w,\frac{\id}{2^m}\right)\right\rangle_{V_{\text{L}}}&=\sum_{\substack{(k,k')\in k_{\LC_{\text{B}}} \\ k' \geq k}} t_{k,k'}  \epsilon(\rho_{k,k'})\,, \label{eq:average-DHS-rho}
\end{align}
where $t_{k,k'}\geq 0$, and where  $k_{\LC_{\text{B}}}$ is the set of $k$ indices whose associated subsystems $S_k$ are in the backward light-cone $\LC_{\text{B}}$ of $W$. Here we defined the function $\epsilon(M) = D_{\HS}\left(M,\Tr(M)\id / d_M\right)$ where $D_{\HS}$ is the Hilbert-Schmidt distance and $d_M$ is the dimension of the matrix $M$. In addition, $\rho_{k,k'}$ is partial trace of the input state $\rho$ down to the subsystems $S_k S_{k+1}... S_{k'}$. 
In addition, we find that the following inequality holds $\forall k,k'$
\begin{equation}
    t_{k,k'}  \geq \frac{2^{ml}}{(2^m+1)^{2l}}\,.
\end{equation}

Hence, we have $\Var[\partial_\nu C]\geq G_{n}(L,l)$, where 
\begin{equation}
    G_{n}(L,l)=
    \frac{2^{m(l+1)-1}}{(2^{2m}-1)^2(2^{m}+1)^{L+l}} \sum_{ i\in i_\LC} \sum_{\substack{(k,k')\in k_{\LC_{\text{B}}} \\ k' \geq k}} c_i^2  \epsilon(\rho_{k,k'}) \epsilon(\widehat{O}_i)\,,
\end{equation}
where we used the fact that $\Tr[\sigma_\nu^2]=2^m$.

This result can be trivially generalized to the case when $\widehat{O}_i$ and $\widehat{O}_j$ are of the general form in Eq.~\eqref{eq:Omlsm} such that they can have overlapping support on at most $m/2$ qubits. Then, from \eqref{eq:DeltaO} it is straightforward to see that $\Delta O^{ij}_\tau=0$ for all $\tau$, $i$, and $j$, as one would always have to compute traces of the form $\Tr[\widehat{O}_i^{\mu_i}]$, which vanish since $\widehat{O}_i^{\mu_i}$ can be written as a tensor product of Pauli operators.

\end{proof}

\section{Proof of Theorem $1$}
\label{sec:proofTheorem1}

The following theorem provides an upper bound on the variance of the partial derivative of a global cost function which can be expressed as the expectation value of an operator of the form
\begin{equation}
O=c_0\id+\sum_{i=1}^{N} c_i \widehat{O}_{i1}\otimes \widehat{O}_{i2}\otimes \dots\otimes \widehat{O}_{i\xi}\,.\label{eqSM:OG}
\end{equation}
Specifically, we consider two cases of interest:  (i) When $N=1$ and each $\widehat{O}_{1k}$ is a non-trivial projector ($\widehat{O}_{1k}^2=\widehat{O}_{1k}\neq\id$) of rank $r_k$ acting on subsystem $S_k$, or (ii) When $N$ is arbitrary and $\widehat{O}_{ik}$ is traceless with  $\Tr[\widehat{O}_{ik}^2]\leq 2^m$ (for example, when $\widehat{O}_{ik}=\bigotimes_{j=1}^m \sigma^\mu_j$ is a tensor product of Pauli operators $\sigma^\mu_j\in\{\id_j,\sigma^x_j,\sigma^y_j,\sigma^z_j\}$, with at least one $\sigma^\mu_j\neq\id$). In Section~\ref{subsec:Projectors} we first consider case (i), while in Section~\ref{subsec:Pauli} we prove Theorem 1 for case (ii).

We now reiterate Theorem 1 for convenience:
\setcounter{theorem}{0}

\begin{theorem}\label{thm1sm}
Consider a trainable parameter  $\theta^\nu$ in a block $W$ of the ansatz in Fig. 3 of the main text.  Let $\Var[\partial_\nu C]$ be the variance of the partial derivative of a global cost function $C$ (with $O$ given by~\eqref{eqSM:OG}) with respect to $\theta^\nu$. If $W_{\text{A}}$, $W_{\text{B}}$ of~\eqref{eq:WA-WB}, and each block in $V(\vec{\theta})$ form a local $2$-design, then $\Var[\partial_\nu C]$ is upper bounded by 
\begin{equation}
    \Var[\partial_\nu C]\leq F_{n}(L,l)\,.\label{eq:varMain1sm} 
\end{equation}
\begin{itemize}
\item[(i)] For $N=1$ and when each $\widehat{O}_{1k}$ is a non-trivial projector, then defining  $R=\prod_{k=1}^{\xi}  r^2_{k}$, we have
\begin{align}
\label{eq:varMain2sm}
F_{n}(L,l) = \frac{2^{2m+(2m-1)(L-l)}}{(2^{2m}-1)\cdot3^{\frac{n}{m}}\cdot2^{(2-\frac{3}{m})n}}c_1^2 R\,.
\end{align}
\item[(ii)] For arbitrary $N$ and when each $\widehat{O}_{ik}$ satisfies $\Tr[\widehat{O}_{ik}]=0$ and $\Tr[\widehat{O}_{ik}^2]\leq 2^m$,  then
\begin{equation}
    F_{n}(L,l)=\frac{2^{2m(L-l+1)+1}}{3^{\frac{2n}{m}}\cdot 2^{\left(3-\frac{4}{m}\right)n}}\sum_{i,j=1}^N c_i c_j \,.
    \label{eq:globalUpperPauliMT}
\end{equation}
\end{itemize}
\end{theorem}

\subsection{$N=1$, and $\widehat{O}_{1k}$ is a non-trivial projector}\label{subsec:Projectors}

For simplicity, let us now use $\widehat{O}_k$ when refering to the operators $\widehat{O}_{ik}$. Let us first consider the case when  each $\widehat{O}_k$ is a non-trivial projector ($\widehat{O}_k^2=\widehat{O}_k$) of rank $r_k$ acting on subsystem $S_k$:
\begin{align}
r_{k}:= \Tr\left[\widehat{O}_{k}\right]=\Tr\left[\widehat{O}^2_{k}\right] = \rank\left[\widehat{O}_{k}\right]\,.
\end{align}

\begin{proof}

In the previous section, we defined $V_{\text{R}}$ as containing all gates in the forward light-cone $\LC$ of $W$ and all the gates in the last layer of the ansatz. Hence, we can express $V_{\text{R}}=V_\LC\otimes V_{\LCb}$, where $V_{\LCb}:\HC_{\LCb}\rightarrow\HC_{\LCb}$ is given by
\begin{equation}
    V_{\LCb}=\bigotimes_{k\in k_{\LCb}} W_{kL}\,,\label{eq:V-LCb-Tp}
\end{equation} 
where where $k_{\LCb}$ is the set of $k$ indices whose associated subsystems of qubits $S_k$ are outside of the forward light-cone $\LC$ of $W$. One can always write  $\ketbra{\vec{q}}{\vec{p}}\otimes \id_w$ as a projector onto $H_\LC$ times a projector onto $H_{\LCb}$:
\begin{align}
\qquad \ketbra{\vec{q}}{\vec{p}}\otimes\id_w=\left(\bigotimes_{k\in k_{\LCb}}\ketbra{\vec{q}}{\vec{p}}_{k}\right)\otimes \ketbra{\vec{q}}{\vec{p}}_{\LC}\otimes\id_w\,.\label{eq:P-Tp}
\end{align}
Combining \eqref{eq:eq-Omegapq} with Eqs.~\eqref{eq:V-LCb-Tp} and~\eqref{eq:P-Tp} leads to
\begin{align}
    \Omega_{\vec{q}\vec{p}}&=\left(\prod_{k\in k_{\LCb}}\Tr_k\left[\ketbra{\vec{p}}{\vec{q}}_{k}W_{kL}\ad \widehat{O}_k W_{kL} \right]\right)\Omega_{\vec{q}\vec{p}}^\LC\,,
\label{eq:OmegamniCG}
\end{align}
where
\begin{equation}
   \Omega_{\vec{q}\vec{p}}^\LC= \Tr_{\LC\cap\overline{w}}\left[(\ketbra{\vec{p}}{\vec{q}}\otimes\id_w)V_\LC O^\LC V_\LC\right]\,,
\end{equation}
and where $O^\LC=\bigotimes_{k\in k_{\LCb}}\widehat{O}_k$. Here $\Tr_k$ indicates the trace over $\HC_k$, while $\Tr_{\LC\cap\overline{w}}$ is the trace over the Hilbert space associated with the qubits in $S_\LC\cap S_{\overline{w}}$.

In order to explicitly evaluate the expectation value $\langle\dots\rangle_{V_{\text{R}}}$ in
\begin{equation}
    \Var[\partial_\nu C]=\frac{2^{m-1}\Tr[\sigma_\nu^2]}{(2^{2m}-1)^2}c_1^2\sum_{\substack{\vec{p}\vec{q}\\\vec{p}'\vec{q}'}}\left\langle\left(\Tr[ \Omega_{\vec{q}\vec{p}}\Omega_{\vec{q}'\vec{p}'}] -\frac{\Tr[ \Omega_{\vec{q}\vec{p}}]\Tr[\Omega_{\vec{q}'\vec{p}'}]}{2^{m}}\right)\right\rangle_{V_{\text{R}}}\left\langle\Delta\Psi_{\vec{p}\vec{q}}^{\vec{p}'\vec{q}'}\right\rangle_{V_{\text{L}}}\,,\label{eq:varlocal2}
\end{equation}
we use the fact that the blocks in $V(\vec{\theta})$ are independent, and hence $\langle\dots\rangle_{V_{\text{R}}}=\langle\dots\rangle_{V_{\LCb},V_\LC}$. Then, we have
\begin{align}
&\left\langle\Tr[ \Omega_{\vec{q}\vec{p}}\Omega_{\vec{q}'\vec{p}'}] -\frac{\Tr[ \Omega_{\vec{q}\vec{p}}]\Tr[\Omega_{\vec{q}'\vec{p}'}]}{2^{m}}\right\rangle_{V_{\text{R}}}= \left\langle\Tr[ \Omega_{\vec{q}\vec{p}}^\LC\Omega_{\vec{q}'\vec{p}'}^\LC] -\frac{\Tr[ \Omega_{\vec{q}\vec{p}}^\LC]\Tr[\Omega_{\vec{q}'\vec{p}'}^\LC]}{2^{m}}\right\rangle_{V_\LC} \left(\prod_{k/S_k\subset S_{\LCb}}\left\langle\Omega_k\right\rangle_{W_{kL}}\right)\,,
\label{eq:OmegaLCsm}
\end{align}
with
\begin{equation}\label{eq:omegak}
\Omega_k=\Tr_k\left[\ketbra{\vec{p}}{\vec{q}}_{k}W_{kL}\ad \widehat{O}_{k} W_{kL} \right]\Tr_k\left[\ketbra{\vec{p}'}{\vec{q}'}_{k}W_{kL}\ad \widehat{O}_{k} W_{kL} \right]\,.
\end{equation}

Let us first compute the expectation value of each $\Omega_k$.
From Lemma~\ref{lemma3}, we have:
\begin{align}
\left\langle\Omega_k\right\rangle_{W_{kL}}&=\frac{r_k}{2^{2m}-1}\left(\left(r_k-\frac{1}{2^m}\right)\delta_{(\vec{p},\vec{q})_{S_k}}\delta_{(\vec{p'},\vec{q'})_{S_k}}+\left(1-\frac{r_k}{2^m}\right)\delta_{(\vec{p},\vec{q'})_{S_k}}\delta_{(\vec{p'},\vec{q})_{S_k}}\right)\nonumber\\
&\leq \frac{1}{2^{2m}-1}r_k^2\left(\delta_{(\vec{p},\vec{q})_{S_k}}\delta_{(\vec{p'},\vec{q'})_{S_k}}+\delta_{(\vec{p},\vec{q'})_{S_k}}\delta_{(\vec{p'},\vec{q})_{S_k}}\right)\,,
\end{align}
where in the inequality we have dropped the negative terms and used the fact that $r_k\leq r_k^2$. Then, we have 
\begin{align*}
    \prod_{k\in k_{\LCb}}\left\langle\Omega_k\right\rangle_{W_{kL}}\leq \frac{1}{(2^{2m}-1)^{\xi_{\LCb}}}\prod_{k\in k_{\LCb}}r_k^2 \left(\delta_{(\vec{p},\vec{q})_{S_k}}\delta_{(\vec{p'},\vec{q'})_{S_k}}+\delta_{(\vec{p},\vec{q'})_{S_k}}\delta_{(\vec{p'},\vec{q})_{S_k}}\right)\,.
    \end{align*}
Combining this result with Eq.~\eqref{eq:varlocal2} leads to the upper bound  
\begin{equation}
\begin{aligned}
    \Var[\partial_\nu C ]&\leq \frac{2^{m-1}\Tr[\sigma_\nu^2]}{(2^{2m}-1)^2(2^{2m}-1)^{\xi_{\LCb}}}c_1^2\sum_{\substack{\vec{p}\vec{q}\\\vec{p}'\vec{q}'}}  \Bigg(\left\langle\Tr[ \Omega_{\vec{q}\vec{p}}^\LC\Omega_{\vec{q}'\vec{p}'}^\LC] -\frac{\Tr[ \Omega_{\vec{q}\vec{p}}^\LC]\Tr[\Omega_{\vec{q}'\vec{p}'}^\LC]}{2^{m}}\right\rangle_{V_\LC}\\
    &\quad \times\left(\prod_{k\in k_{\LCb}}r_k^2 \left(\delta_{(\vec{p},\vec{q})_{S_k}}\delta_{(\vec{p'},\vec{q'})_{S_k}}+\delta_{(\vec{p},\vec{q'})_{S_k}}\delta_{(\vec{p'},\vec{q})_{S_k}}\right)\right) \left\langle\Delta\Psi_{\vec{p}\vec{q}}^{\vec{p}'\vec{q}'}\right\rangle_{V_{\text{L}}}\Bigg)\,.\label{eq:var-Glob-UB}
\end{aligned}
\end{equation}

As discussed in the previous section, one can compute  the expectation values in Eq.~\eqref{eq:var-Glob-UB}  by systematically integrating over each block over the unitary group with the respect to the Haar measure. From Eq.~\eqref{eq:Tn_Sumsm} we can find
 \begin{equation}
\begin{aligned}
    \Var[\partial_\nu C ]&\leq \frac{2^{m-1}\Tr[\sigma_\nu^2]}{(2^{2m}-1)^2(2^{2m}-1)^{\xi_{\LCb}}}c_1^2\sum_{\substack{\vec{p}\vec{q}\\\vec{p}'\vec{q}'}}\sum_\tau  \Bigg( t_\tau 
    \delta_{(\vec{p},\vec{q})_{S_{\overline{\tau}}}}\delta_{(\vec{p}',\vec{q}')_{S_{\overline{\tau}}}}\delta_{(\vec{p},\vec{q}')_{S_\tau}}\delta_{(\vec{p}',\vec{q})_{S_\tau}}\Delta O_\tau\\
    &\quad \times \left(\prod_{k\in k_{\LCb}}r_k^2\left(\delta_{(\vec{p},\vec{q})_{S_k}}\delta_{(\vec{p'},\vec{q'})_{S_k}}+\delta_{(\vec{p},\vec{q'})_{S_k}}\delta_{(\vec{p'},\vec{q})_{S_k}}\right)\right) \left\langle\Delta\Psi_{\vec{p}\vec{q}}^{\vec{p}'\vec{q}'}\right\rangle_{V_{\text{L}}}\Bigg)\,.\label{eq:var-Glob-UB2}
\end{aligned}
\end{equation}
Note that by expanding the product $\prod_{k\in k_{\LCb}}\left(\delta_{(\vec{p},\vec{q})_{S_{k}}}\delta_{(\vec{p'},\vec{q'})_{S_{k}}}+\delta_{(\vec{p},\vec{q'})_{S_{k}}}\delta_{(\vec{p'},\vec{q})_{S_{k}}}\right)\delta_{(\vec{p},\vec{q})_{S_{\overline{\tau}}}}\delta_{(\vec{p}',\vec{q}')_{S_{\overline{\tau}}}}\delta_{(\vec{p},\vec{q}')_{S_\tau}}\delta_{(\vec{p}',\vec{q})_{S_\tau}}$, one obtains a sum of $2^{\xi_{\LCb}}$ terms, and according to Eqs.~\eqref{eq:DHS0}--\eqref{eq:DHSsm}, each term in the summation leads to a Hilbert-Schmidt distances between two quantum states. Then, since $D_{HS}(\rho_1,\rho_2)\leq 2$ for any pair of states $\rho_1$, $\rho_2$, we find 
\begin{equation}
    \sum_{\substack{\vec{p}\vec{q}\\\vec{p}'\vec{q}'}}  \delta_{(\vec{p},\vec{q})_{S_{\overline{\tau}}}}\delta_{(\vec{p}',\vec{q}')_{S_{\overline{\tau}}}}\delta_{(\vec{p},\vec{q}')_{S_\tau}}\delta_{(\vec{p}',\vec{q})_{S_\tau}} \left( \prod_{k\in k_{\LCb}}\left(\delta_{(\vec{p},\vec{q})_{S_k}}\delta_{(\vec{p'},\vec{q'})_{S_k}}+\delta_{(\vec{p},\vec{q'})_{S_k}}\delta_{(\vec{p'},\vec{q})_{S_k}}\right)\right) \left\langle\Delta\Psi_{\vec{p}\vec{q}}^{\vec{p}'\vec{q}'}\right\rangle_{V_{\text{L}}}\leq 2\cdot 2^{\xi_{\LCb}}\,. 
\end{equation}
Replacing this result in~\eqref{eq:var-Glob-UB2} leads to 
\begin{align}   
    \Var[\partial_\nu C ]&\leq \frac{2^{m}2^{\xi_{\LCb}}\Tr[\sigma_\nu^2]}{(2^{m}-1)^2(2^{2m}-1)^{\xi_{\LCb}}}c_1^2\left(\prod_{k\in k_{\LCb}} r^2_{k}\right) \sum_{\tau}t_{\tau}\Delta O_{\tau}  \,.\label{eq_varCG-upper}
\end{align}

Next, we consider the terms $\Delta O_\tau$. From Eq.~\eqref{eq:DeltaO} we can show that
\begin{align}
    \Delta O_{\tau} \leq \prod_{k\in k_{\LC}}r^2_k\,,
\end{align}
so that~\eqref{eq_varCG-upper} becomes 
\begin{align}
    \Var[\partial_\nu C ]&\leq \frac{2^{m}2^{\xi_{\LCb}}\Tr[\sigma_\nu^2]}{(2^{2m}-1)^2(2^{2m}-1)^{\xi_{\LCb}}}  c_1^2 R \sum_\tau t_\tau \,,
\end{align}
where $R=\prod_{k=1}^{\xi} r^2_{k}$.

Let us finally show that $\sum_\tau t_\tau \leq 2$ $\forall l ,L$. We recall from Eq.~\eqref{eq:varlocal2} that the coefficients $t_\tau$ are obtained by integrating each block in $V(\vec{\theta})$ over the unitary group with the respect to the Haar measure. Each time a block is integrated one obtains four new tensors weighted by the coefficients $\eta_i$, with $i=1,\ldots,4$ such that $\sum_{\mu=1}^4|\eta_\mu|\leq 1$ for all $m$. Consider now the average $\left\langle\Tr[\Omega_{\vec{q}\vec{p}}^i\Omega_{\vec{q}'\vec{p}'}^j]\right\rangle_{V_{\text{R}}}$, once all the blocks have been integrated we find 
\begin{equation}
    \left\langle\Tr[\Omega_{\vec{q}\vec{p}}^i\Omega_{\vec{q}'\vec{p}'}^j]\right\rangle_{V_{\text{R}}}=\sum_{\mu_1=1}^{4}\eta_{\mu_1}\left(\sum_{\mu_2=1}^{4}\eta_{\mu_2}\left(\dots\left(\sum_{\mu_m=1}^{4}\eta_{\mu_m}\right)\right)\right) T_{\mu_1,\mu_2,\ldots,\mu_m}(O)\,,\label{eq:absval}
\end{equation}
where $\mu_m$ is the number of averaged blocks, and where $T_{\mu_1,\mu_2,\ldots,\mu_m}(O)$ is a tensor contraction of $O$. A similar equation can be obtained for $\left\langle\Tr[ \Omega_{\vec{q}\vec{p}}^i]\Tr[\Omega_{\vec{q}'\vec{p}'}^j]\right\rangle_{V_{\text{R}}}$ as 
\begin{equation}
    \left\langle\Tr[ \Omega_{\vec{q}\vec{p}}^i]\Tr[\Omega_{\vec{q}'\vec{p}'}^j]\right\rangle_{V_{\text{R}}}=\sum_{\mu_1'=1}^{4}\eta_{\mu_1'}\left(\sum_{\mu_2'=1}^{4}\eta_{\mu_2'}\left(\dots\left(\sum_{\mu_m'=1}^{4}\eta_{\mu_m'}\right)\right)\right)T_{\mu_1',\mu_2',\ldots,\mu_m'}(O)\,,\label{eq:absval2}
\end{equation}
so that 
\begin{equation}\label{eq:ttau2}
   \sum_\tau t_\tau =\sum_{\mu_1=1}^{4}\eta_{\mu_1}\left(\sum_{\mu_2=1}^{4}\eta_{\mu_2}\left(\dots\left(\sum_{\mu_m=1}^{4}\eta_{\mu_m}\right)\right)\right)-\frac{1}{2^m}\sum_{\mu_1'=1}^{4}\eta_{\mu_1'}\left(\sum_{\mu_2'=1}^{4}\eta_{\mu_2'}\left(\dots\left(\sum_{\mu_m'=1}^{4}\eta_{\mu_m'}\right)\right)\right)\,.
\end{equation}
Taking the absolute value on both side and using the fact that $\sum_{\mu=1}^4|\eta_\mu|\leq 1$, one gets $\sum_\tau t_\tau=|\sum_\tau t_\tau|\leq1+\frac{1}{2^m}\leq 2$.

Therefore, by using $\Tr\left[\sigma_\nu^2\right]=2^m$, we have
\begin{align*}
    \Var[\partial_\nu C ]&\leq\frac{2^{2m+ \frac{n}{m}+l-L}}{(2^{2m}-1)^2(2^{2m}-1)^{ \frac{n}{m}+l-L-1}}c_1^2 R\,.
\end{align*}
Moreover, we also find that 
\begin{align}
\frac{2^{2m+ \frac{n}{m}+l-L}}{(2^{2m}-1)^2(2^{2m}-1)^{ \frac{n}{m}+l-L-1}}c_1^2R = \frac{2^{2m}}{2^{2m}-1}\cdot\left(2^{2m-1}-\frac{1}{2}\right)^{L-l}\cdot \frac{2^{\frac{n}{m}}}{2^{2n}(1-2^{-2m})^{\frac{n}{m}}}c_1^2 R\,,
\end{align}
and since $2^{2m-1}-\frac{1}{2}< 2^{2m-1}$ and $1\leq\frac{1}{1-2^{-2m}}\leq \frac{4}{3}$, $\Var[\partial_\nu C ]$ can be upper bounded as $\Var[\partial_\nu C ]< F_n(L,l)$, where
\begin{align}
    F_n(L,l)= \frac{2^{2m+(2m-1)(L-l)}}{(2^{2m}-1)\cdot3^{\frac{n}{m}}\cdot2^{(2-\frac{3}{m})n}}c_1^2 R\,.
\label{eq:varglobalupper}
\end{align}
\end{proof}

\subsection{Arbitrary $N$ and $\widehat{O}_{ik}$ is a traceless operators  such that  $\Tr[\widehat{O}_{ik}^2]\leq 2^m$}\label{subsec:Pauli}

We now consider the case when $\widehat{O}_{ik}$ is traceless. 
Specifically, we assume
\begin{align}\label{eq:SMtraceless}
    \Tr\left[\widehat{O}_{ik}\right]=0\,, \quad \text{and}\quad \Tr\left[\widehat{O}_{ik}^2\right]\leq 2^m\,.
\end{align}
Note that if $\widehat{O}_{ik}$ is a tensor product of Pauli operators then $\Tr\left[\widehat{O}_{ik}^2\right]= 2^m$.

\begin{proof}
Let us first analyze the case of $N=1$. Combining Lemma~\ref{lemma3} with Eqs.~\eqref{eq:OmegaLCsm}, and~\eqref{eq:omegak}, we find
\begin{align}
\left\langle\Omega_k\right\rangle_{W_{kL}}&\leq\frac{2^m}{2^{2m}-1}\left(\delta_{(\vec{p'},\vec{q})_{S_k}}\delta_{(\vec{p},\vec{q'})_{S_k}}-\frac{1}{2^m}\delta_{(\vec{p},\vec{q})_{S_k}}\delta_{(\vec{p'},\vec{q'})_{S_k}}\right)\leq \frac{2^m}{2^{2m}-1}\delta_{(\vec{p'},\vec{q})_{S_k}}\delta_{(\vec{p},\vec{q'})_{S_k}}\,,
\end{align}
which yields
\begin{align}\label{eq:omegakpauli}
    \prod_{k\in k_{\LCb}} \left\langle\Omega_k\right\rangle_{W_{kL}} \leq \frac{2^{\xi_{\LCb}m}}{(2^{2m}-1)^{\xi_{\LCb}}}\prod_{k\in k_{\LCb}} \delta_{(\vec{p'},\vec{q})_{S_k}}\delta_{(\vec{p},\vec{q'})_{S_k}}\,.
\end{align}
Therefore, from \eqref{eq:var-Glob-UB2}, and~\eqref{eq:omegakpauli}, we can obtain 
\begin{align}
    \Var\left[\partial_\nu C\right] \leq& \frac{2^{2m-1}}{(2^{2m}-1)^2(2^{2m}-1)^{\xi_{\LCb}}}c_1^2\cdot \frac{2^{\xi_{\LCb}m}}{(2^{2m}-1)^{\xi_{\LCb}}}\\
    &\times\sum_{\substack{\vec{p}\vec{q}\\\vec{p}'\vec{q}'}}\sum_\tau   t_\tau \Delta O_\tau \left(\prod_{k\in k_{\LCb}}\delta_{(\vec{p'},\vec{q})_{S_k}}\delta_{(\vec{p},\vec{q'})_{S_k}}\right)
    \delta_{(\vec{p},\vec{q})_{S_{\overline{\tau}}}}\delta_{(\vec{p}',\vec{q}')_{S_{\overline{\tau}}}}\delta_{(\vec{p},\vec{q}')_{S_\tau}}\delta_{(\vec{p}',\vec{q})_{S_\tau}} \left\langle\Delta\Psi_{\vec{p}\vec{q}}^{\vec{p}'\vec{q}'}\right\rangle_{V_{\text{L}}}\,.
\end{align}
Let us now consider the terms $\Delta O_{\tau}$ from \eqref{eq:DeltaO}. It is straightforward to see that $\Delta O_{\tau}=0$ if $S_{\tau_z}\neq\emptyset$, due to the fact that $\Tr[\widehat{O}_{ik}]=0$.  Hence, let us define $\tau_p$ as the set of indexes such that  $S_{\tau_z}=\emptyset$. Then, we have 
\begin{align}
\Delta O_{\tau} \leq
\begin{cases}
2^{m(L-l+1)}&~~(\tau\in\tau_p)\\
0&~~(\tau\not\in\tau_p)
\end{cases}\,.
\end{align}
The latter implies  
\begin{align}
\sum_\tau t_\tau \Delta O_{\tau}\leq\sum_{\tau\in\tau_p} t_\tau \Delta O_{\tau}\leq 2\times 2^{m(L-l+1)}\,,
\end{align}
where we used the fact that $\sum t_{\tau} \leq 2$ (see Eq.~\eqref{eq:ttau2} in the previous section).

According to Eqs.~\eqref{eq:DHS0}--\eqref{eq:DHSsm}, the product $\left(\prod_{k\in k_{\LCb}}\delta_{(\vec{p'},\vec{q})_{S_{k}}}\delta_{(\vec{p},\vec{q'})_{S_{k}}}\right)\delta_{(\vec{p},\vec{q})_{S_{\overline{\tau}}}}\delta_{(\vec{p}',\vec{q}')_{S_{\overline{\tau}}}}\delta_{(\vec{p},\vec{q}')_{S_\tau}}\delta_{(\vec{p}',\vec{q})_{S_\tau}}$, leads to a Hilbert-Schmidt distances between two quantum states, which is always upper bounded by $2$. We then find  
\begin{align}\label{eq:varPauli}
    \Var\left[\partial_\nu C\right] \leq \frac{2^{2m+1}}{(2^{2m}-1)^2(2^{2m}-1)^{\xi_{\LCb}}}c_1^2\cdot \frac{2^{\xi_{\LCb}m}}{(2^{2m}-1)^{\xi_{\LCb}}}\cdot 2^{m(L-l+1)}\,.
\end{align}
Since $n=m(L-l+1)+m\xi_{\LCb}$, we can rewrite~\eqref{eq:varPauli} as
\begin{align}
    \Var\left[\partial_\nu C\right] \leq \frac{2^{2m+n+1}(2^{2m}-1)^{2(L-l)}}{2^{4n}(1-2^{-2m})^{\frac{2n}{m}}}c_1^2\,.
\end{align}
Finally, noting that $2^{2m}-1\leq2^{2m}$ and $1\leq \frac{1}{1-2^{-2m}}\leq \frac{4}{3}$, we have 
\begin{align}\label{eq:boundvarproof2}
    \Var\left[\partial_\nu C\right] \leq \frac{2^{2m(L-l+1)+1}}{3^{\frac{2n}{m}} 2^{\left(3-\frac{4}{m}\right)n}}c_1^2\,.
\end{align}

Let us now consider the arbitrary $N$ case, where  the operator $O$ can be expressed as
\begin{equation}\label{eq:Op-O-PS-sm}
    O=c_0\id+\sum_{i=1}^{N} c_i O_i \,, 
\end{equation}
and where each $O_i$ can be expressed as $O_i=\widehat{O}_{i1}\otimes \widehat{O}_{i2}\otimes \dots\otimes \widehat{O}_{i\xi}$ with $\widehat{O}_{ik}$ satisying~\eqref{eq:SMtraceless}.

Here it is convinient to define $C_i=\Tr[O_iV(\thv)\rho V\ad(\thv)]$, so that $C=c_0+\sum_ic_iC_i$. When computing the variance we now have to consider the cross terms arising from $C_i$, and $C_j$ with $i\neq j$:
\begin{equation}\label{eq:expCij}
\Var[\partial_\nu C]=\langle(\partial_\nu C)^2\rangle_V=\sum_ic_i^2 \langle(\partial_\nu C_i)^2\rangle_V+\sum_{i\neq j}c_ic_j\langle\partial_\nu C_i\partial_\nu C_j\rangle_V\,.
\end{equation}
First, let us remark that 
\begin{align}
    \Tr\left[\widehat{O}_{ik}\right]=0\,, \quad \text{and}\quad \Tr\left[\widehat{O}_{ik}\widehat{O}_{jk'}\right]\leq 2^m\,, \,\, \forall{i,j,k,k'} \quad \text{and}\quad 
\begin{cases}
\Delta O_{\tau}^{ij} \leq2^{m(L-l+1)}&~~(\tau\in\tau_p)\\
\Delta O_{\tau}^{ij} =0&~~(\tau\not\in\tau_p)
\end{cases}\,,
\end{align}
where $\Delta O_{\tau}^{ij}$ was defined in \eqref{eq:DeltaO}, and where $\Tr\left[\widehat{O}_{ik}\widehat{O}_{jk'}\right]\leq 2^m$ follows from the fact that $\Tr\left[\widehat{O}_{ik}^2\right]\leq 2^m$ and $\Tr\left[\widehat{O}_{jk'}^2\right]\leq 2^m$. Hence, it is straigtforward to see that the upper bound in Eq.~\eqref{eq:boundvarproof2} also holds for the cross terms in \eqref{eq:expCij}, and we then have
\begin{equation}\label{eq:upperCij}
    \Var[\partial_\nu C]\leq F_n(L,l)\,.
\end{equation}
where
\begin{equation}\label{eq:upperCij2}
  F_n(L,l) = \frac{2^{2m(L-l+1)+1}}{3^{\frac{2n}{m}} 2^{\left(3-\frac{4}{m}\right)n}} \sum_{ij}c_ic_j\,.
\end{equation}
\end{proof}

\section{Proofs of Corollaries}
\label{sec:proofcor}

\subsection{Proof of Corollary 1}
\label{sec:Proofcor1}

\begin{corollary}\label{cor1sm}
Consider the function $F_n(L,l)$. 

\begin{itemize}
    \item[(i)] Let $N=1$ and let each $\widehat{O}_{1k}$ be a non-trivial projector, as in case (i) of  Theorem~\ref{thm1sm}. If $c_1^2 R \in\OC(2^n)$ and if the number of layers $L\in\OC(\poly(\log(n)))$, then
    \begin{equation}
F_{n}\left(L,l\right)\in\OC\left(2^{-\left(1-\frac{1}{m}\log_23\right) n}\right)\,,
\end{equation}
which implies that $\Var[\partial_\nu C]$ is exponentially vanishing in $n$ if $m\geq2$.

\item[(ii)] Let $N$ be arbitrary, and let each $\widehat{O}_{ik}$ satisfy $\Tr[\widehat{O}_{ik}]=0$ and $\Tr[\widehat{O}_{ik}^2]\leq 2^m$, as in case (ii) of  Theorem~\ref{thm1sm}. If $N\in\OC(2^n)$,  $c_i\in\OC(1)$, and if the number of layers $L\in\OC(\poly(\log(n)))$, then
\begin{equation}
F_{n}\left(L,l\right)\in\OC\left(\frac{1}{2^{\left(1-\frac{1}{m}\right)n}}\right)\,,
\end{equation}
which implies that $\Var[\partial_\nu C]$ is exponentially vanishing in $n$ if $m\geq2$.
\end{itemize}

\end{corollary}

Let us first consider case (i) in Theorem~\ref{thm1sm}.
\begin{proof}
Let us assume $L\in\OC\left(\poly(\log(n))\right)$, so that we have $2^{(2m-1)(L-l)}\in\OC\left(2^{\poly(\log(n))}\right)$. Here, note that  
\begin{align*}
    \lim_{n\to \infty}\left|{\frac{\poly(\log(n))}{n}}\right|=0,
\end{align*}
which means that $n$ grows faster than the any polylogarithmic functions of $n$. Therefore, we can write  
$\OC\left(\poly(\log(n))\right)\subset\OC\left(n\right)$. For a nonzero constant $a$, we can write $\OC\left(|a|n\right)=\OC\left(n\right)$, and we can choose $a=\frac{1}{m}\log_2\left(9/8\right)$. Therefore, we can also write 
$\OC\left(2^{\poly(\log(n))}\right)\subset\OC\left(\left(9/8\right)^{\frac{n}{m}}\right)$.  In addition, if we have $c_1^2 R\in\OC\left(2^n\right)$, from \eqref{eq:varglobalupper}, one can obtain  
\begin{align}
    F_{n}(L,l)\in\frac{1}{3^{\frac{n}{m}} 2^{\left(2-\frac{3}{m}\right)n}}\OC\left(\left(9/8\right)^{\frac{n}{m}}\right)\OC\left(2^n\right)=\OC\left(\frac{1}{2^{\left(1-\frac{1}{m}\log_23\right) n}}\right),
\end{align}
where $1-\frac{1}{m}\log_2 3>0$ for $m\geq 2$. Hence the upper bound of $\Var[\partial_\nu C_G]$ exponentially vanishing when $m\geq 2$. 
\end{proof}

Let us now consider case (ii) in Theorem~\ref{thm1sm}.
\begin{proof}
We now have  $2^{2m(L-l+1)+1)}\in\OC\left(2^{\poly(\log(n))}\right)\subset\OC\left(\left(9/8\right)^{\frac{n}{m}}\right)$. We have from Eq.~\eqref{eq:upperCij2} that if  $c_i,c_j \in\OC\left(1\right)$, the following bound holds 
\begin{align*}
    F_{n}(L,l)\in \frac{1}{3^{\frac{2n}{m}}2^{\left(3-\frac{4}{m}\right)n}}\OC\left(\left(9/8\right)^{\frac{n}{m}}\right) \OC\left(2^{2n}\right)=
    \OC\left(\frac{1}{2^{\left(1-\frac{1}{m}\right)n}}\right)\,.
\end{align*}
Hence, the upper bound of $\Var[\partial_\nu C_G]$ exponentially vanishing when $m\geq 2$.
\end{proof}

\subsection{Proof of Corollary 2}
\label{sec:Proofcor2}

\begin{corollary}\label{cor2SM}
Consider the function $F_n(L,l)$. Let $O$ be an operator of the form~\eqref{eq:Omlsm}, as in Theorem~\ref{thm2sm}. If at least one term $c_i^2  \epsilon(\rho_{k,k'}) \epsilon(\widehat{O}_i)$  in the sum in \eqref{eqSM:varMaint22sm} vanishes no faster than $\Omega(1/\poly(n))$, and if the number of layers $L$ is $\OC(\log(n))$, then 
\small
\begin{equation}
    G_n(L,l) \in \Omega\left(\frac{1}{\poly(n)}\right)\,.\label{eq:varGlobalCor2SM}
\end{equation}
On the other hand, if at  least one term $c_i^2  \epsilon(\rho_{k,k'}) \epsilon(\widehat{O}_i)$ in the sum in~\eqref{eqSM:varMaint22sm} vanishes no faster than  $\Omega\left(1/2^{\poly(\log(n))}\right)$, and if the number of layers is $\OC(\poly(\log(n)))$, then 
\begin{equation}
    G_n(L,l) \in \Omega\left(\frac{1}{2^{\poly(\log(n))}}\right)\,.\label{eq:varGlobalCor3SM}
\end{equation}
\end{corollary}

\begin{proof}
Let us  assume that at least one term $ c_i^2  \epsilon(\rho_{k,k'}) \epsilon(\widehat{O}_i)$ in~\eqref{eqSM:varMaint22sm} vanishes no faster than $\Omega\left(1/\poly(n)\right)$. 
Then, if  we also assume $L\in \OC\left(\log(n)\right)$,  we have $(2^m+1)^{-(L+l)}\in\Omega\left(1/\poly(n)\right)$. The latter  implies
\begin{align*}
    G_{n}(L,l)\in\Omega\left(\frac{1}{\poly(n)}\right) \Omega\left(\frac{1}{\poly(n)}\right)=\Omega\left(\frac{1}{\poly(n)}\right)\,. 
\end{align*} 

On the other hand, if at  least one term $ c_i^2  \epsilon(\rho_{k,k'}) \epsilon(\widehat{O}_i)$ in~\eqref{eqSM:varMaint22sm} vanishes no faster than  $\Omega\left(1/2^{\poly(\log(n))}\right)$, and if $L\in \OC\left(\poly(\log(n))\right)$, we have $(2^m+1)^{-(L+l)}\in \Omega\left(1/2^{\poly(\log(n))}\right)$. Therefore, we obtain 
\begin{align*}
    G_{n}(L,l)\in\Omega\left(\frac{1}{2^{\poly(\log(n))}}\right) \Omega\left(\frac{1}{2^{\poly(\log(n))}}\right)=\Omega\left(\frac{1}{2^{\poly(\log(n))}}\right)\,.
\end{align*}
\end{proof}

\section{Faithfulness of local cost function for Quantum autoencoder}
\label{sec:faithfulness}

Recall that the global cost function $(C_G')$ and local cost function $(C_{\text{L}}')$ for the quantum autoencoder are defined as 
\begin{align*}
    \begin{split}
        C_G' &=1- \Tr\left[\dya{\vec{0}}\rho_{\text{B}}^{\text{out}}\right]\\
        C_{\text{L}}' & =  1- \frac{1}{n_{\text{B}}}\sum_{j=1}^{n_{\text{B}}}\Tr\left[\left(\dya{0}_j\otimes\id_{\overline{j}}\right)\rho_{\text{B}}^{\text{out}}\right]\,.
    \end{split}
\end{align*}
We relate the $C_{\text{L}}'$ and $C_G'$ cost functions in the following proposition. Because this proposition establishes that $C_{\text{L}}'$ and $C_G'$ vanish under the same conditions, and because $C_G'$ is faithful~\cite{Romero}, this in turn proves that $C_{\text{L}}'$ is a faithful cost function.
\begin{proposition}
The cost functions for the quantum autoencoder satisfy \begin{equation}
\label{eq:faithfulness}
    C_{\text{L}}'\leq C_G' \leq n_{\text{B}} C_{\text{L}}'\,.
\end{equation}
\end{proposition}
\begin{proof}
Let us first prove $C_{\text{L}}'\leq C_G'$.  Given the state $\rho_{\text{B}}^{\text{out}}$, we can define $E_j$ as the event of qubit $j$ being measured on the $\ket{0}_j$ state, such that the  probability of $E_j$ is given by $\Pr\left(E_j\right) = \Tr\left[\widehat{O}_j^L\rho_{\text{B}}^{\text{out}}\right]$,
where $\widehat{O}_j^L = \dya{0}_j\otimes \id_{\overline{j}}$. 
Then, we can write
\begin{align}\label{eq:CGprob}
    C_G' =1-\Pr\left(\bigcap_{j=1}^{n_{\text{B}}}E_j\right)
\end{align}
Similarly, for the local cost function, we have
\begin{align}\label{eq:CLprob}
    C_{\text{L}}' =1-\frac{1}{n_{\text{B}}}\sum_{j=1}^{n_{\text{B}}}\Pr\left(E_j\right).
\end{align}
Then, it is known that for any set of events $\{E_1,\cdots, E_{n_{\text{B}}}\}$, the following property always holds
\begin{align}\label{eq:unionprob1}
1-\Pr\left(\bigcap_{j=1}^{n_{\text{B}}} E_j\right) \ge \frac{1}{n_{\text{B}}}\sum_{j=1}^{n_{\text{B}}}\left(1-\Pr\left(E_j\right)\right)\,.
\end{align}
From the definition of $C_G'$ and $C_{\text{L}}'$ in Eqs.~\eqref{eq:CGprob}--~\eqref{eq:CLprob}, and from Eq.~\eqref{eq:unionprob1} we have
\begin{align}\label{eq:boundlower}
    C_{\text{L}}'\leq C_G'\,.
\end{align}

Next, let us prove $C_G'\leq n_{\text{B}} C_{\text{L}}'$. 
Since for any set of events $\{E_1,\cdots,E_{n_{\text{B}}}\}$, we have 
\begin{align}\label{eq:unionprob2}
  1-\Pr\left(\bigcap_{j=1}^{n_{\text{B}}} E_j\right) \leq \sum_{j=1}^{n_{\text{B}}}\left(1-\Pr\left(E_j\right)\right).
\end{align}
By definition, we finally have
\begin{align}\label{eq:boundupper}
C_G'\leq n_{\text{B}} C_{\text{L}}'. 
\end{align}
Combining Eqs.~\eqref{eq:boundlower} and~\eqref{eq:boundupper}, we get~\eqref{eq:faithfulness}, which indicates $C_{\text{L}}'=0~\Longleftrightarrow~C_G'=0$.
\end{proof}

\end{document}